\documentclass{llncs}
\makeatletter
\@ifundefined{lhs2tex.lhs2tex.sty.read}%
  {\@namedef{lhs2tex.lhs2tex.sty.read}{}%
   \newcommand\SkipToFmtEnd{}%
   \newcommand\EndFmtInput{}%
   \long\def\SkipToFmtEnd#1\EndFmtInput{}%
  }\SkipToFmtEnd

\newcommand\ReadOnlyOnce[1]{\@ifundefined{#1}{\@namedef{#1}{}}\SkipToFmtEnd}
\usepackage{amstext}
\usepackage{amssymb}
\usepackage{stmaryrd}
\DeclareFontFamily{OT1}{cmtex}{}
\DeclareFontShape{OT1}{cmtex}{m}{n}
  {<5><6><7><8>cmtex8
   <9>cmtex9
   <10><10.95><12><14.4><17.28><20.74><24.88>cmtex10}{}
\DeclareFontShape{OT1}{cmtex}{m}{it}
  {<-> ssub * cmtt/m/it}{}

\DeclareFontShape{OT1}{cmtt}{bx}{n}
  {<5><6><7><8>cmtt8
   <9>cmbtt9
   <10><10.95><12><14.4><17.28><20.74><24.88>cmbtt10}{}
\DeclareFontShape{OT1}{cmtex}{bx}{n}
  {<-> ssub * cmtt/bx/n}{}
\newcommand{\Conid}[1]{\mathit{#1}}
\newcommand{\Varid}[1]{\mathit{#1}}
\newcommand{\anonymous}{\kern0.06em \vbox{\hrule\@width.5em}}

\newcommand{\bind}{\mathbin{>\!\!\!>\mkern-6.7mu=}}
\newcommand{\sequ}{\mathbin{>\!\!\!>}}

\usepackage{polytable}

\@ifundefined{mathindent}%
  {\newdimen\mathindent\mathindent\leftmargini}%
  {}%

\def\resethooks{%
  \global\let\SaveRestoreHook\empty
  \global\let\ColumnHook\empty}
\newcommand*{\savecolumns}[1][default]%
  {\g@addto@macro\SaveRestoreHook{\savecolumns[#1]}}
\newcommand*{\restorecolumns}[1][default]%
  {\g@addto@macro\SaveRestoreHook{\restorecolumns[#1]}}
\newcommand*{\aligncolumn}[2]%
  {\g@addto@macro\ColumnHook{\column{#1}{#2}}}

\resethooks

\newcommand{\onelinecommentchars}{\quad-{}- }
\newcommand{\commentbeginchars}{\enskip\{-}
\newcommand{\commentendchars}{-\}\enskip}

\newcommand{\visiblecomments}{%
  \let\onelinecomment=\onelinecommentchars
  \let\commentbegin=\commentbeginchars
  \let\commentend=\commentendchars}

\newcommand{\invisiblecomments}{%
  \let\onelinecomment=\empty
  \let\commentbegin=\empty
  \let\commentend=\empty}

\visiblecomments

\newlength{\blanklineskip}
\newcommand{\hsindent}[1]{\quad}% default is fixed indentation
\let\hspre\empty
\let\hspost\empty

\EndFmtInput
\makeatother
\ReadOnlyOnce{polycode.fmt}%
\makeatletter

\newcommand{\hsnewpar}[1]%
  {{\parskip=0pt\parindent=0pt\par\vskip #1\noindent}}

\newcommand{\hscodestyle}{}

\newcommand{\sethscode}[1]%
  {\expandafter\let\expandafter\hscode\csname #1\endcsname
   \expandafter\let\expandafter\endhscode\csname end#1\endcsname}

  {\par\noindent
   \advance\leftskip\mathindent
   \hscodestyle
   \let\\=\@normalcr
   \let\hspre\(\let\hspost\)%
   \pboxed}%
  {\endpboxed\)%
   \par\noindent
   \ignorespacesafterend}

  {\hsnewpar\abovedisplayskip
   \advance\leftskip\mathindent
   \hscodestyle
   \let\hspre\(\let\hspost\)%
   \pboxed}%
  {\endpboxed%
   \hsnewpar\belowdisplayskip
   \ignorespacesafterend}

  {\hsnewpar\abovedisplayskip
   \advance\leftskip\mathindent
   \hscodestyle
   \let\\=\@normalcr
   \(\pboxed}%
  {\endpboxed\)%
   \hsnewpar\belowdisplayskip
   \ignorespacesafterend}

\newcommand{\plainhs}{\sethscode{plainhscode}}

\plainhs

  {\hsnewpar\abovedisplayskip
   \advance\leftskip\mathindent
   \hscodestyle
   \let\\=\@normalcr
   \(\parray}%
  {\endparray\)%
   \hsnewpar\belowdisplayskip
   \ignorespacesafterend}

  {\parray}{\endparray}

  {\(\parray}{\endparray\)}

\def\codeframewidth{\arrayrulewidth}
\RequirePackage{calc}

  {\parskip=\abovedisplayskip\par\noindent
   \hscodestyle
   \arrayrulewidth=\codeframewidth
   \tabular{@{}|p{\linewidth-2\arraycolsep-2\arrayrulewidth-2pt}|@{}}%
   \hline\framedhslinecorrect\\{-1.5ex}%
   \let\endoflinesave=\\
   \let\\=\@normalcr
   \(\pboxed}%
  {\endpboxed\)%
   \framedhslinecorrect\endoflinesave{.5ex}\hline
   \endtabular
   \parskip=\belowdisplayskip\par\noindent
   \ignorespacesafterend}

\newcommand{\framedhslinecorrect}[2]%
  {#1[#2]}

  {\(\def\column##1##2{}%
   \let\>\undefined\let\<\undefined\let\\\undefined
   \newcommand\>[1][]{}\newcommand\<[1][]{}\newcommand\\[1][]{}%
   \def\fromto##1##2##3{##3}%
   }{\) }%

  {\let\orighscode=\hscode
   \let\origendhscode=\endhscode
   \def\endhscode{\def\hscode{\endgroup\def\@currenvir{hscode}\\}\begingroup}
   \orighscode\def\hscode{\endgroup\def\@currenvir{hscode}}}%
  {\origendhscode
   \global\let\hscode=\orighscode
   \global\let\endhscode=\origendhscode}%

\makeatother
\EndFmtInput

\usepackage{paralist, balance}
\usepackage{graphicx,subfig,wrapfig}
\usepackage{amsmath,amssymb}%,stmaryrd}
\usepackage{tikz}
\newcommand{\hcancel}[1]{%
    \tikz[baseline=(tocancel.base)]{
        \node[inner sep=0pt,outer sep=0pt] (tocancel) {#1};
        \draw[red, thick] (tocancel.south west) -- (tocancel.north east);
    }%
}%

\usepackage{url}
\usepackage{mathpartir}
\usepackage[utf8]{inputenc}
\newcommand{\ignore}[1]{}
\newcommand{\Red}[1]{{\color{red} #1}}
\newcommand{\Blue}[1]{{\color{blue} #1}}
\newcommand{\concept}[1]{}%{\Blue{\bf #1:}}}
\long\def\comment#1{}
\newcommand{\tocite}[1]{\Red{[cite]}}

\long\def\cut#1{}

\definecolor{gray}{RGB}{84,84,84}
\definecolor{dark-green}{RGB}{ 0,100,  0}

\newif\ifextended
\extendedfalse

\ifextended
\newcommand{\inputproof}[1]{\input{#1}}
\newcommand{\inlong}[1]{\Red{#1}}
\else
\newcommand{\inputproof}[1]{}
\newcommand{\inlong}[1]{}
\fi

\usepackage{soul} % for highlighting
\usepackage[square,comma,numbers,sort&compress,sectionbib]{natbib}

\newcommand{\nsu}{{no-sensitive-upgrade}}
\newcommand{\Nsu}{{No-sensitive-upgrade}}
\newcommand{\pu}{{permissive-upgrade}}
\newcommand{\Pu}{{Permissive-upgrade}}

\newcommand{\flows}{\sqsubseteq}
\newcommand{\lub}{\sqcup}
\newcommand{\glb}{\sqcap}

\newcommand{\lto}{\longrightarrow}

\newcommand{\lattice}{\ensuremath{\ell}}

\newcommand{\lcurr}{\ensuremath{l_{\textrm{cur}}}}

\newcommand{\Coloneqq}{::=} % txfonts

\newcommand{\fresh}[1]{\ensuremath{\textrm{fresh}(#1)}}

\newcommand{\ruleref}[1]{(\textsc{#1})}

\newcommand{\conf}[2]{\langle #1 \ensuremath{|} #2\rangle}
\newcommand{\tconf}[3]{\{ #1, #2 \ensuremath{|} #3\}}
\newcommand{\thread}[3]{\langle #1, #2, #3\rangle}

\newcommand{\lio}{\ensuremath{\lambda^{\textsf{\tiny LIO}}_{\ensuremath{\lattice}}}}
\newcommand{\liofs}{\ensuremath{\lambda^{\textsf{\tiny LIO}}_{\ensuremath{\lattice,\textsc{fs}}}}}
\newcommand{\lioconc}{\ensuremath{\lambda^{\textsf{\tiny $\parallel$-LIO}}_{\ensuremath{\lattice}}}}
\newcommand{\lioafs}{\ensuremath{\lambda^{\textsf{\tiny LIO}}_{\ensuremath{\lattice,\textsc{fs+au}}}}}
\usepackage{thmtools}
\begin{document}

\frontmatter

\title{On Dynamic Flow-Sensitive Floating-Label Systems}
\ifextended
\subtitle{Extended Version}
\fi

\author{Pablo Buiras\inst{1} \and Deian Stefan\inst{2} \and Alejandro Russo\inst{1}}
\institute{Chalmers University of Technology \and Stanford University}
% \address[Pablo Buiras]{Chalmers University of Technology}
% \address[Alejandro Russo]{Chalmers University of Technology}
% \address[Deian Stefan]{Stanford University}

\maketitle

\begin{abstract}
  Flow-sensitive analysis for information-flow control (IFC) allows data
  structures to have mutable security labels, i.e., labels that can change over
  the course of the computation.
  This feature is often used to boost the permissiveness of the IFC monitor, by
  rejecting fewer runs of programs, and to reduce the burden of explicit label
  annotations.
  However, adding flow-sensitive constructs (e.g., references or files) to a
  dynamic IFC system is subtle and may also introduce high-bandwidth covert
  channels.
  %when added naively, in a purely dynamic setting, mutable labels can
  %expose a high bandwidth covert channel.
  %
  %\Red{
  %  Outside the wording ``IFC monitor`` the first two sentences make
  %  it seem like we're talking about static IFC.
  %   I would change ``flow-sensitive analysis`` and remove ``rejecting fewer
  %   programs.''
  %   Also adding mutable labels is not really the feature that is being
  %   added to the language---maybe:
  %   Unfortunately, adding flow-sensitive constructs (e.g., references
  %   or files) to a dynamic IFC system is subtle and may also introduce
  %   to high-bandwidth covert channels.
  % }
  %
  In this work, we extend LIO---a language-based floating-label
  system---with flow-sensitive references.
  The key insight to safely manipulating the label of a reference is to not
  only consider the label on the data stored in the reference, i.e., the
  reference label, but also the \emph{label on the reference label} itself.
  Taking this into consideration, we provide an \emph{upgrade} primitive that
  can be used to change the label of a reference in a safe manner.
  We additionally provide a mechanism for automatic upgrades
  to eliminate the burden of determining when a reference should be
  upgraded.
  This approach naturally extends to a concurrent setting, which has
  not been previously considered by dynamic flow-sensitive systems.
  For both our sequential and concurrent calculi we prove non-interference by
  embedding the flow-sensitive system into the original,
  flow-insensitive LIO calculus---a surprising result on its own.
\end{abstract}

\section{Introduction}
\label{sec:intro}

Modern software systems are composed of many complex components that handle
sensitive data.
In many cases (e.g., mobile and web applications) these disparate
components are provided by different authors, of varying
trustworthiness.
Unfortunately, because today's software development tools do not
provide a means for protecting sensitive data from untrusted code,
data theft and corruption is prevalent.

% \Red{IMO the third sentence doesn't add much (I do think it's simply
%   because of the way it's written though): I think everybody
%   reading the first two will get that there is risk. Maybe hint at the
%   fact that it is possible to reduce/remove this risk e.g., by
%   phrasing the sentence is a way that makes it seem like the way we
%   build software today is the reason this is a risk.}

%% Information-flow control as a promising technology
Information-flow control (IFC) is a promising approach to security
that provides data confidentiality and integrity in the presence of
untrusted code.
At a high level, IFC tracks and controls the flow of information through a
system according to a security policy, usually
\emph{non-interference}~\cite{Goguen:Meseguer:Noninterference}.
%\Red{Alt:At a high level, IFC ...}
%
Non-interference states that public events should not depend on
sensitive data and dually, trusted data should not be affected by
untrusted events.
Hence, with IFC, the program is guaranteed to preserve data
confidentiality and integrity, even when composed of untrusted
components.
Indeed, this appealing guarantee has recently led to significant
research and development efforts that use IFC
to secure web
applications (e.g.,~\cite{DeGroef:2012:FWB:2382196.2382275,
giffin:hails, yang:2013:towards, Hedin13}) and mobile platforms
(e.g.,~\cite{Enck:2010,android:esorics13}).
%, where untrusted code is the norm.

%% Language-based enforcement
% Language-based systems that employ dynamic execution monitors to enforce IFC
% have become popular~\cite{Hedin2011}.
% %
% This is due, in part, to the permissiveness of dynamic techniques
% (when compared to static approaches~\cite{Sabelfeld:Russo:PSI09}),
% %the ability to specify fine-grained policies,
% %Ale: Type-system handles fine-grained policies too
%  and the ability to handle complex language features like dynamic code evaluation---a feature
% common to many scripting languages.
% \Red{This is a pretty weird paragraph. First, the transition into the
%   paragraph ``In particular'' doesn't make sense--in particular
%   what? The first citation is a bit weird, it makes it seem like
% it's backing up this claim of popularity. Third, the topic sentence
% makes it seem like we're about to back up the rise of language-based
% systems vs. dynamic IFC systems. Fourth, we're not really doing
% anything with code eval so it's a bit of bait-and-switch. I would
% honestly cut this paragraph --- I this it adds more confusion than
% anything else.
% }

To ensure data confidentiality and integrity, these dynamic IFC
systems associate \emph{security labels} with data and monitor where
such data can flow~\cite{myers:dlm, Stefan:2011}.
In this paper, we use the labels \ensuremath{\textbf{H}} and \ensuremath{\textbf{L}}, to respectively
denote secret and public data, and ensure that information cannot flow
from a secret entity into a public one, i.e., the labels are ordered
such that \ensuremath{\textbf{L}\;\flows\;\textbf{H}} and \ensuremath{\textbf{H}\;\not\flows\;\textbf{L}}.
In general, the partial order \ensuremath{\flows} (label check) is used to
govern the allowed flows.
We remark that our results apply to arbitrary lattices that may
 also express integrity concerns~\cite{myers:dlm, Stefan:2011}, we
 only use the two-point lattice for simplicity of exposition.
% \Red{generalized lattice? want to say: results are in terms of an
% arbitrary lattice.}

One of the facets of IFC analysis lies in how such labels, when associated with
objects, are treated~\cite{Hunt:2006}.
Specifically, some IFC systems (e.g.,~\cite{Breeze, Aeolus, stefan:lio,
stefan:addressing-covert, zeldovich:histar, Efstathopoulos:2005, krohn:flume})
treat labels on objects as \emph{immutable} and do not allow for changes over
the lifetime of the program, i.e., labels of objects are \emph{flow-insensitive}.
%
% In contrast, other systems (e.g.,~\cite{jif, FlowCaml,
% Ale: FlowCaml and JIF are flow-insensitive (they have label inference tough)
In contrast, other systems (e.g.,~\cite{Zdancewic02programminglanguages, Austin:Flanagan:PLAS09,
Austin:Flanagan:PLAS10}) are \emph{flow-sensitive}, i.e., they allow object
labels to change, in certain conditions, according to the sensitivity of the
data that is stored in the object.
In general, these flow-sensitive systems are more permissive, i.e., they allow
programs that flow-insensitive monitors would reject.
%\Red{I think that just the difference in citation count will raise
%eyebrows. Maybe note that the complexity of it is a reason it has not
%been considered in IFC OSes.}

Consider, for instance, a web application that writes to a labeled
log while servicing user requests.
If the label of the log is \ensuremath{\textbf{L}}, a flow-insensitive IFC monitor would disallow
writing any sensitive data (e.g., error messages containing user-supplied data)
to the log, since this would constitute a leak.
However, in a flow-sensitive system, the label of the log can change (to
\ensuremath{\textbf{H}}), as to accommodate the kinds of data being written to the log.
For many applications, allowing labels to change in such a way is very
desirable---it alleviates the burden of having to a-priori determine
the precise labels of objects (e.g., the log).

\begin{wrapfigure}{r}{0.40\columnwidth}
%\vspace{-15pt}
  \small
\begin{hscode}\SaveRestoreHook
\column{B}{@{}>{\hspre}l<{\hspost}@{}}%
\column{6}{@{}>{\hspre}l<{\hspost}@{}}%
\column{E}{@{}>{\hspre}l<{\hspost}@{}}%
\>[B]{}\Varid{l}{}\<[6]%
\>[6]{}\mathbin{:=}\Conid{True}{}\<[E]%
\\
\>[B]{}\Varid{tmp}{}\<[6]%
\>[6]{}\mathbin{:=}\Conid{False}{}\<[E]%
\\
\>[B]{}\mathbf{if}\;\Varid{h}\;\mathbf{then}\;\Varid{tmp}\mathbin{:=}\Conid{True}{}\<[E]%
\\
\>[B]{}\mathbf{if}\;\neg \;\Varid{tmp}\;\mathbf{then}\;\Varid{l}\mathbin{:=}\Conid{False}{}\<[E]%
\ColumnHook
\end{hscode}\resethooks
%\vspace{-pt}
%%|l := 1 ; t := 0 ;| \\
%%|if h then t := 1;| \\
%%|if t /= 1 then l := 0| \\
%%|publish(l)|
\caption{\small\label{fig:fsattackintro} Flow-sensitive attack}
\end{wrapfigure}
Unfortunately, naively introducing flow-sensitive objects to a dynamic
IFC system can turn label changes into a covert channel~\cite{Russo:2010}.
% \Red{I'm not a fan of this ``purely dynamic'' wording; it's hinting at
% the fact that you can do better if you're static/hybrid, but not
% how/why so why not just stick to ``dynamic'' and leave the comparison
% to related work.}
%
Consider the code fragment of Figure~\ref{fig:fsattackintro}
where references \ensuremath{\Varid{l}} and \ensuremath{\Varid{h}} are respectively labeled \ensuremath{\textbf{L}} and
\ensuremath{\textbf{H}}.
By naively allowing arbitrary label changes---even if the new label is
more restricting---we can leak the contents of \ensuremath{\Varid{h}} into \ensuremath{\Varid{l}}.
In particular, suppose that the temporary variable \ensuremath{\Varid{tmp}} is initially
labeled \ensuremath{\textbf{L}}.
If the value stored in \ensuremath{\Varid{h}} is \ensuremath{\Conid{True}}, then in the first conditional we
assign \ensuremath{\Conid{True}} into \ensuremath{\Varid{tmp}} and raise its label to \ensuremath{\textbf{H}}, reflecting the
fact that the branch condition depends on sensitive data.
Since the \ensuremath{\Varid{tmp}} is \ensuremath{\Conid{True}}, the branch condition for the second conditional is
\ensuremath{\Conid{False}} and thus the value and label of \ensuremath{\Varid{l}} are left intact, i.e., \ensuremath{\Conid{True}} at
\ensuremath{\textbf{L}}.
However, if \ensuremath{\Varid{h}} is \ensuremath{\Conid{False}}, then the value and label of \ensuremath{\Varid{tmp}} do not
change---the first assignment is not executed.
Instead, the second assignment, which sets \ensuremath{\Varid{l}} to \ensuremath{\Conid{False}}, is
performed; since the label of the branch condition is \ensuremath{\textbf{L}}, the
label of \ensuremath{\Varid{l}} remains \ensuremath{\textbf{L}}.
Note that in both cases the label of \ensuremath{\Varid{l}} stays \ensuremath{\textbf{L}}, but the value of \ensuremath{\Varid{l}} is
the same as the secret \ensuremath{\Varid{h}}.
(Hence why \emph{label change} is considered a covert channel.)
In systems such as LIO and Breeze, which allow labels to be inspected, this
attack can be further simplified by simply checking the label of \ensuremath{\Varid{tmp}} after
the first assignment---if the secret is true then the label will be \ensuremath{\textbf{H}},
otherwise it will be \ensuremath{\textbf{L}}.

This attack is not new, and, to ensure that the covert channel is
not introduced when adding flow-sensitive references in such a way,
several solutions have already been proposed.
These solutions fall into roughly three categories.
First, the IFC monitor can incorporate static information to ensure that such
leaks are disallowed~\cite{Russo:2010}.
Second, the IFC monitor can forbid certain label changes, depending on the
context (e.g., the program counter (\emph{pc})
label~\cite{sabelfeld:language-based-iflow}).
For instance, the \emph{no-sensitive upgrades} policy disallows raising the
label of a public reference in a sensitive context (e.g., when a branch
condition is \ensuremath{\textbf{H}})~\cite{Zdancewic02programminglanguages,
Austin:Flanagan:PLAS09}.
And, third, the monitor can disallow branches that depend on certain variables,
for which the label was mutated, as done by the \emph{permissive upgrades}
policy~\cite{Austin:Flanagan:PLAS10}.

In this paper, we take a fresh perspective on flow-sensitivity %the no-sensitive upgrades
% discipline
in the context of coarse-grained floating-label systems, in
particular, the LIO IFC system~\cite{stefan:lio,
stefan:addressing-covert}. % Ale: remove Haskell
LIO brings ideas from IFC Operating Systems---notably,
HiStar~\cite{zeldovich:histar}---into a language-based setting.
In particular, LIO takes an OS-like coarse-grained approach by associating a
single ``current'' floating-label with a computation (and everything in scope),
instead of heterogeneously labeling every variable, as typically done by
language-based systems (e.g.,~\cite{jif,FlowCaml}).
This floating-label is raised (e.g., from \ensuremath{\textbf{L}} to \ensuremath{\textbf{H}}) to accommodate
reading sensitive data and thus serves as a form of ``taint'' reflecting the
sensitive of data in context, i.e., LIO is flow-sensitive in the current label.
(This can be seen as raising the \emph{pc} in more traditional language-based
systems.)
In turn, the LIO monitor uses the ``current'' label to
restrict where the computation can write (e.g., once the current label
is raised to \ensuremath{\textbf{H}}, it can no longer write to references labeled
\ensuremath{\textbf{L}}).
However, like other IFC systems, % already cited some in a previous paragraph
LIO is \emph{flow-insensitive} in
object labels. 

%% Contributions
%Inspired by previous results~\citep{Austin:Flanagan:PLAS10},
This work extends the LIO IFC system, both the sequential and concurrent
versions, to incorporate flow-sensitive references.
A key insight of this work is to consider labels of references as
being composed of two elements: the reference label describing the
confidentiality (integrity) of the stored value, and another label, called
\emph{the label on the label}, which describes the confidentiality
(integrity) of the reference label itself.
Our monitor, then only forbids changing a label of a reference if
\emph{the label on the label is below the current floating-label}.
Inspired by~\cite{Hedin13}, we add a primitive for safely and
explicitly \emph{upgrading} labels.
%when permitted by our monitor.
%
This boosts the permissiveness of LIO, and, for instance, allows programs, such
at the logging web application described above, which would otherwise be
rejected by the IFC monitor.

To reduce the burden of introducing upgrade annotations, our
calculus provides a means for automatically upgrading references for which the
computation is about to ``lose'' write access, i.e., before tainting the
computation by raising the current label, we first upgrade all the references
whose labels are below the (new) current label.
While secure, this feature facilitates a form of \emph{label creep},
wherein all flow-sensitive references might end up with labels that
are ``too high.''
To further address this, we propose a block-structured primitive which
only upgrades the labels of declared flow-sensitive references, while
disallowing writes to undeclared ones.

By taking a fresh perspective on flow-sensitivity, we also show that our
flow-sensitive extension
%the {\nsu}
%policy and our upgrade operation
can be entirely encoded using existing flow-insensitive
constructs---the key insight is to explicitly model flow-sensitive values as
\emph{nested flow-insensitive labeled references}.
In the context of LIO, this encoding has the added benefit of allowing us to prove
non-interference by simply invoking previous results.
Equally important, the sequential semantics for LIO with flow-sensitive
references directly extend to the concurrent setting.

The contributions of this paper are as follows:
\begin{itemize}

\item We extend LIO to incorporate flow-sensitive objects, with a focus on
references.
Specifically, we introduce two explicit primitives to safely raise (\ensuremath{\textbf{upgrade}})
or downgrade (\ensuremath{\textbf{downgrade}}) the security label of references.
This extension not only increases LIO's permissiveness, but also provides a
means for safely combining flow-insensitive and flow-sensitive references.

\item We present a uniform treatment for flow-insensitive and
flow-sensitive references in both sequential and concurrent settings. To the
best of our knowledge, we are the first to analyze the challenges of purely
dynamic monitors with flow-sensitive references in the presence of concurrency.
%\Red{Is mechanism the right word here?}

\item A non-interference proof for the different calculi that leverages the
encoding of flow-sensitive references using flow-insensitive constructs.

\end{itemize}
%

%
%\hl{(Ale): Indicate the difference between CSF and Journal version}
The novel aspect of this article, with respect to its conference
version~\cite{buiras:on-dynamic}, is the extension of our formal
results to consider a \ensuremath{\textbf{downgrade}} primitive that further boosts
permissiveness. Additionally, we compare our approach with
\emph{no-sensitive-upgrade}~\cite{Zdancewic02programminglanguages} and
\emph{permissive-upgrade}~\cite{Austin:Flanagan:PLAS09}---two known
policies for label changes.
% \Red{novetly is the wrong word. The novel aspect of this article...
% And the comment on comparison really just makes it seem like the
% conference version was not doing a good job at related work or that it
% ha a bug but we don't say what it is: I would remove it, this is not
% really a contribution.}
%
% \Red{Alternative (maybe combine with next 1-line paragraph):
% %
% This paper extends an earlier conference
% version~\cite{buiras:on-dynamic} by considering a |downgrade|
% primitive that further boosts permissiveness.
% }

We remark that while our development focuses on LIO, we believe that
our results generalize to other sequential and concurrent
floating-label systems (e.g.,~\cite{Breeze, Efstathopoulos:2005,
zeldovich:histar, krohn:flume}).

%\hl{Structure of the paper, i.e., the paper is organized as follows}
The rest of the paper is organized as follows. Section~\ref{sec:background}
provides an introduction to LIO and its formalization.
Section~\ref{sec:flow-sensitive} presents our flow-sensitivity extensions and
enforcement mechanism. Section~\ref{sec:conc} extends this approach to the
concurrent setting. Section~\ref{sec:soundness} presents the embedding of our
enforcement using flow-insensitive constructs, from which our formal security
guarantees follow.  We discuss related work in Section~\ref{sec:related} and
conclude in Section~\ref{sec:conclusion}.

% Local Variables:
% TeX-master: "main.tex"
% TeX-command-default: "Make"
% End:
% 1 1/2 Ale

\section{Introduction to LIO}
\label{sec:background}
%\hl{DONE: make sure we explain TCB is not part of the surface syntax}
%\hl{DONE, in the description of Figure 2}

%\concept{LIO library as LIO-monad}
LIO is a language-level IFC system, implemented as a library in
Haskell.  The library provides a new \emph{monad}, \ensuremath{\Conid{LIO}}, atop which
programmers implement computations, which may use the LIO
API to perform side effects (e.g., mutate a reference or write to a
file).
%
%Concretely, LIO provides programmers with a new monad, called |LIO|,
%that is similar to---and is intended to be used in place of---the
%standard Haskell |IO| monad, but additionally monitors and controls
%the flow of information.
%
%In turn, various practical features (e.g., mutable references and
%threads) are implemented as combinators in this monad.
%
%\concept{Current label used to govern flows}

The \ensuremath{\Conid{LIO}} monad implements a purely dynamic execution monitor.
%\Red{Again ``purely.''}
%
Specifically, \ensuremath{\Conid{LIO}} encapsulates the state necessary to enforce
IFC for the computation under evaluation.
Part of this state is the current (floating) label.
Intuitively, the current label serves a role similar to the program
counter (\ensuremath{\Varid{pc}}) of more-traditional IFC systems (e.g.,~\cite{FlowCaml}):
it is used to restrict the current
computation from performing side-effects that may compromise the
confidentiality or integrity of data (e.g., by restricting where the
current computation may write).
%Henceforth, we assume that all computations are encoded in |LIO|.

%
%For instance, if the current label is |lcurr|, LIO restricts the
%computation to only reading data labeled |l_d|, where |l_d canFlowTo
%lcurr|.
%
%For instance, it restricts the computation from writing to entities
%labeled |l_e|, unless |lcurr canFlowTo l_e|.

%\concept{Current label is label on all in scope}
To soundly reason about IFC, every piece of data \emph{must} be
labeled, including literals, terms, and labels themselves.
However and different from most language-based systems
(e.g.,~\cite{myers:jif, FlowCaml, Breeze}) where every value is
explicitly labeled,
LIO takes a coarse-grained approach and uses the current label to protect all
values in scope.
As in IFC operating systems~\cite{Efstathopoulos:2005,
zeldovich:histar}, in LIO, the current label \ensuremath{\lcurr} is the label on
all the non-explicitly labeled values in the context of a computation.
%
%% It is repetitive with the previous paragraph
%Since the current label is used to restrict the current computation
%from performing arbitrary side-effects---this ensures that that the
%confidentiality (and integrity) of all values in scope are preserved.
%

%\concept{Floating label}

To allow for computations on resistive data, LIO raises the current
label to protect newly read data.
That is, the current label is raised to ``float'' above the labels of
all the objects read by the current computation.
Raising the current label allows computations to flexibly read data,
at the cost of being more limited in where they can subsequently
write.
Concretely, a computation with current label \ensuremath{\lcurr} can read data
labeled \ensuremath{\Varid{l}_{\Varid{d}}} by raising its current label to \ensuremath{\lcurr'\mathrel{=}\lcurr\;\lub\;\Varid{l}_{\Varid{d}}}, but can thereafter only write to entities labeled \ensuremath{\Varid{l}_{\Varid{e}}} if
\ensuremath{\lcurr'\;\flows\;\Varid{l}_{\Varid{e}}}.
Hence, for example, a public LIO computation can read secret data by
first raising \ensuremath{\lcurr} from \ensuremath{\textbf{L}} to \ensuremath{\textbf{H}}. Importantly, however, the
new current label prevents the computation from subsequently writing
to public entities.

% \concept{Labeled objects}
% \hl{TODO:talk about labeled values and references}
% \hl{Ale:I think it is not needed, it can be done along the way}

\subsection{\lio: A coarse-grained IFC calculus}

%\concept{\lio{} syntax}

\begin{figure}[t]%{0.5\columnwidth}
\small
\centering
\begin{hscode}\SaveRestoreHook
\column{B}{@{}>{\hspre}l<{\hspost}@{}}%
\column{9}{@{}>{\hspre}l<{\hspost}@{}}%
\column{14}{@{}>{\hspre}c<{\hspost}@{}}%
\column{14E}{@{}l@{}}%
\column{16}{@{}>{\hspre}c<{\hspost}@{}}%
\column{16E}{@{}l@{}}%
\column{19}{@{}>{\hspre}l<{\hspost}@{}}%
\column{E}{@{}>{\hspre}l<{\hspost}@{}}%
\>[B]{}\mathrm{Values}\;{}\<[9]%
\>[9]{}\Varid{v}{}\<[14]%
\>[14]{}\Coloneqq{}\<[14E]%
\>[19]{}\Conid{True}\mid \Conid{False}\mid ()\mid \lambda \Varid{x}.\Varid{t}\mid \lattice\mid LIO^{\Red{\textsf{{\tiny TCB}}}}\;\Varid{t}{}\<[E]%
\\
\>[B]{}\mathrm{Terms}\;{}\<[9]%
\>[9]{}\Varid{t}{}\<[14]%
\>[14]{}\Coloneqq{}\<[14E]%
\>[19]{}\Varid{v}\mid \Varid{x}\mid \Varid{t}\;\Varid{t}\mid \mathbf{fix}\;\Varid{t}\mid \mathbf{if}\;\Varid{t}\;\mathbf{then}\;\Varid{t}\;\mathbf{else}\;\Varid{t}{}\<[E]%
\\
\>[14]{}\hsindent{2}{}\<[16]%
\>[16]{}\mid {}\<[16E]%
\>[19]{}\Varid{t}\;\otimes\;\Varid{t}\mid \mathbf{return}\;\Varid{t}\mid \Varid{t}\bind \Varid{t}\mid \mathbf{getLabel}{}\<[E]%
\\
\>[B]{}\mathrm{Types}\;{}\<[9]%
\>[9]{}\tau{}\<[14]%
\>[14]{}\Coloneqq{}\<[14E]%
\>[19]{}\Conid{Bool}\mid ()\mid \tau\to \tau\mid \lattice\mid \Conid{LIO}\;\tau{}\<[E]%
\\
\>[B]{}\mathrm{Ops}_\lattice\;{}\<[9]%
\>[9]{}\otimes{}\<[14]%
\>[14]{}\Coloneqq{}\<[14E]%
\>[19]{}\lub\mid \glb\mid \flows{}\<[E]%
\ColumnHook
\end{hscode}\resethooks
\caption{\label{fig:sos:base} Syntactic categories for base \lio.}
\end{figure}
We give the precise semantics for LIO by extending the simply-typed,
call-by-name $\lambda$-calculus;
we call this extended IFC calculus \lio.
The formal syntax of the core \lio{} calculus, parametric in the label
type \ensuremath{\lattice}, is given in Fig.~\ref{fig:sos:base}.
Syntactic categories \ensuremath{\Varid{v}}, \ensuremath{\Varid{t}}, and \ensuremath{\tau} represent values, terms, and
types, respectively.
Values include standard primitives (Booleans, unit, and $\lambda$-abstractions)
and terminals corresponding to labels (\ensuremath{\lattice}) and monadic values (\ensuremath{LIO^{\Red{\textsf{{\tiny TCB}}}}\;\Varid{t}}).\footnote{We restrict our formalization to computations implemented in the
  \ensuremath{\Conid{LIO}} monad and only consider Haskell features relevant to IFC, similar to the
  presentation of LIO in~\cite{stefan:2012:arxiv-flexible}.}
% \Red{This
%   is not the right citation; the arxiv paper has the old JFP which is
% not using the small steps semantics we have in the updated paper.
% Since JFP is not yet published we may want to first update the arxiv
% paper. Ale?}}
% \Red{I think we should
%   cite JFP here since the semantics are really from this paper; prior to
% this we did not distinguish pure from monadic contexts}.  }
We note
that values of the form \ensuremath{LIO^{\Red{\textsf{{\tiny TCB}}}}\;\Varid{t}} denote computations subject to security
checks. (In fact, security checks are only applied to such values.)
Terms are composed of standard constructs (values, variables \ensuremath{\Varid{x}}, function
application, the \ensuremath{\mathbf{fix}} operator, and conditionals), terminals corresponding to
label operations (\ensuremath{\Varid{t}\;\otimes\;\Varid{t}}, where \ensuremath{\lub} is the join, \ensuremath{\glb} is the meet, and
\ensuremath{\flows} is the partial-order on labels), standard monadic operators (\ensuremath{\mathbf{return}\;\Varid{t}} and \ensuremath{\Varid{t}\bind \Varid{t}}), and \ensuremath{\mathbf{getLabel}}, a term for inspecting the current
label, as further explained below.
We do not consider terms annotated with \ensuremath{\cdot\;{}^{\Red{\textsf{{\tiny{TCB}}}}}} as part of the
surface syntax, i.e., such syntax nodes are not made available to
programmers and are solely used internally in our semantic
description.
Types consist of Booleans, unit, function types, labels, and \ensuremath{\Conid{LIO}}
computations; since the \lio{} type system is standard, we do not
discuss it further.

% Monadic operators
%
We include monadic terms in our calculus since (in Haskell) monads
dictate the evaluation order of a program and encapsulate all
side-effects, including I/O~\cite{moggi:monads, wadler:monads};
LIO leverages monads to precisely control what (side-effecting)
operations the programmer is allowed to perform at any given time.
\begin{wrapfigure}{r}{0.30\columnwidth}
\vspace{-12pt}
\begin{hscode}\SaveRestoreHook
\column{B}{@{}>{\hspre}l<{\hspost}@{}}%
\column{5}{@{}>{\hspre}l<{\hspost}@{}}%
\column{E}{@{}>{\hspre}l<{\hspost}@{}}%
\>[B]{}\mathbf{do}\;{}\<[5]%
\>[5]{}\Varid{x}\leftarrow \Varid{t}{}\<[E]%
\\
\>[5]{}\mathbf{return}\;(\Varid{x}\mathbin{+}\mathrm{1}){}\<[E]%
\ColumnHook
\end{hscode}\resethooks
%\vspace{-15pt}
\caption{\small\label{fig:do}\ensuremath{\mathbf{do}}-notation}
\end{wrapfigure}
In particular, an LIO program is simply a computation in the \ensuremath{\Conid{LIO}} monad,
composed from simpler monadic terms using \emph{return} and \emph{bind}.
Term \ensuremath{\mathbf{return}\;\Varid{t}} produces a computation which simply returns the value
denoted by \ensuremath{\Varid{t}}.
Term \ensuremath{\bind }, called \emph{bind}, is used to sequence
LIO computations. Specifically, term \ensuremath{\Varid{t}\bind (\lambda \Varid{x}.\Varid{t}')} takes the result
produced by term \ensuremath{\Varid{t}} and applies function \ensuremath{\lambda \Varid{x}.\Varid{t}'} to it.  (This
operator allows computation \ensuremath{\Varid{t}'} to depend on the value produced by
\ensuremath{\Varid{t}}.) 
We sometimes use Haskell’s \ensuremath{\mathbf{do}}-notation to write such monadic
computations. For example, the term \ensuremath{\Varid{t}\bind \lambda \Varid{x}.\mathbf{return}\;(\Varid{x}\mathbin{+}\mathrm{1})}, which
simply adds 1 to the value produced by the term \ensuremath{\Varid{t}}, can be written
using \ensuremath{\mathbf{do}}-notation as shown in Figure~\ref{fig:do}.
%

%\concept{configurations}
%The reduction rules for monadic terms deserve some attention.
%
%%Rather than modeling LIO state, such as the current label, as
%%purely-functional monad layer, we take an operational and more
%%imperative approach by modeling state as a separate component in a
%%configuration.
%
%%Following a more imperative approach, we model the state kept by LIO as a
%%separate component in a configuration.
%in an operational and
%more-imperative style as a separate component.
A top-level \lio{} computation is a \emph{configuration} of the form
\ensuremath{\conf{\Sigma}{\Varid{t}}}, where \ensuremath{\Varid{t}} is the monadic term and \ensuremath{\Sigma} is the state associated with
the term.
As in~\cite{stefan:lio,stefan:addressing-covert}, we take an
imperative approach to modeling the LIO state as a separate component
of the configuration (as opposed to being part of the term).
We partially define the state of \lio~ to (at least) contain the
current label \ensuremath{\lcurr}, i.e., \ensuremath{\Sigma\mathrel{=}(\lcurr,\mathbin{...})}; here, \ensuremath{\mathbin{...}} denotes
other parts of the state not relevant at this point.
Under this definition, a top-level well-typed \lio{} term has the form \ensuremath{\Delta,\Gamma\vdash\Varid{t}\mathbin{:}\Conid{LIO}\;\tau}, where \ensuremath{\Delta} is the store typing, and \ensuremath{\Gamma} is the usual type
environment.
%
% As a result of this encoding, the definition for |return| and |(>>=)|
% are trivial: the former simply reduces to a monadic value by wrapping
% the term with the |LIOTCB| constructor, while the latter evaluates the
% left-hand term and supplies the result to the right-hand term, as
% usual.

\begin{figure}[t]%{0.52\columnwidth} % sos:rules-abr
\small
\begin{hscode}\SaveRestoreHook
\column{B}{@{}>{\hspre}l<{\hspost}@{}}%
\column{5}{@{}>{\hspre}l<{\hspost}@{}}%
\column{E}{@{}>{\hspre}l<{\hspost}@{}}%
\>[B]{}\Red{E}{}\<[5]%
\>[5]{}\Coloneqq\Red{E}\;\Varid{t}\mid \mathbf{fix}\;\Red{E}\mid \mathbf{if}\;\Red{E}\;\mathbf{then}\;\Varid{t}\;\mathbf{else}\;\Varid{t}\mid \Red{E}\;\otimes\;\Varid{t}\mid \Varid{v}\;\otimes\;\Red{E}{}\<[E]%
\\
\>[B]{}\textbf{\Blue{E}}{}\<[5]%
\>[5]{}\Coloneqq[\mskip1.5mu \mskip1.5mu]\mid \Red{E}\mid \textbf{\Blue{E}}\bind \Varid{t}{}\<[E]%
\ColumnHook
\end{hscode}\resethooks
\begin{mathpar}
%{\centering ...}
\\
\inferrule[getLabel]
{ \ensuremath{\Sigma\mathrel{=}(\lcurr,\mathbin{...})}  }
{
\ensuremath{\conf{\Sigma}{\textbf{\Blue{E}}\;[\mskip1.5mu \mathbf{getLabel}\mskip1.5mu]}\lto\conf{\Sigma}{\textbf{\Blue{E}}\;[\mskip1.5mu \mathbf{return}\;\lcurr\mskip1.5mu]}}
}
\end{mathpar}
\caption{Evaluation contexts and \ensuremath{\mathbf{getLabel}} reduction rule.\label{fig:sos:rules-abr}}
\end{figure}
%\concept{eval ctx}
We use evaluation contexts in the style of Felleisen and Hieb to specify the
reduction rules for \lio~\cite{felleisen1992revised}.
Figure~\ref{fig:sos:rules-abr} defines the evaluation contexts for pure
terms (\ensuremath{\Red{E}}) and monadic terms (\ensuremath{\textbf{\Blue{E}}}) for the base \lio.
The definitions are standard; we solely highlight that monadic terms
are evaluated only at the outermost use of bind (\ensuremath{\textbf{\Blue{E}}\bind \Varid{t}}), as in
Haskell.
For the base \lio, we also give the reduction rule for the monadic
term \ensuremath{\mathbf{getLabel}}, which simply retrieves the current label.
As shown later, it is precisely this label that is used to restrict
the reads/writes performed by the current computation.
%
%Rule \ruleref{getLabel} defines the LIO library function |getLabel|.
%
%\hl{The following sentence seems out of place, like it's missing a follow-up sentence.}
%Recall that it is the current label that protects all the values in scope.
%
The rest of the reduction rules for the base calculus are straight
forward and given Appendix~\ref{sec:app:sem}.
%

%
%However, since neither |return| nor |>>=| perform any security checks,
%all LIO functions must inspect the current label to enforce IFC.
%
%Ale: for space reasons only
% For example, we can define a |publish| function that writes the
% supplied value to a public channel in terms of a |publishTCB| function
% (that unsafely writes any data to the channel) as follows:
% \begin{code}
% publish v = do  lcurr <- getLabel
%                 if (lcurr canFlowTo low)  then  publishTCB v
%                                           else  return ()
% \end{code}
% Hence, if the current computation has read any sensitive data---and
% thus the current label is |high|---it cannot subsequently leak it with
% publish.\footnote{
%   As with other TCB functions, we assume that |publishTCB| is not part
%   of the surface syntax.
% }

\subsection{Labeled values}

%\concept{|lcurr| is too coarse-grained}
Using \ensuremath{\lcurr} as the label on all terms in scope makes it trivial to deal with
implicit flows. Branch conditions, which are simply values of type \ensuremath{\Conid{Bool}}, are
already implicitly labeled with \ensuremath{\lcurr}. Consequently, all the subsequent writes cannot leak
this bit---the current label restricts all the possible writes.
%reason about issues that have traditionally been hard to deal with in
%IFC systems.
%
%For example, LIO does not have to explicitly handle the implicit flows
%problem: branch conditions (which are are simply values of type
%|Bool|) are labeled |lcurr|, and thus any subsequent writes cannot leak
%the sensitive bit---the current label is used when governing all
%writes (even those in a branch).
%
However, this coarse-grained labeling approach suffers from a severe restriction: a piece
of code cannot, for example, %and independently of any secret,
write a public
value (e.g., \ensuremath{\mathrm{42}}) to a public channel labeled \ensuremath{\textbf{L}} after observing secret
data, even if the value is independent from the secret---once secret data is read, the current label is raised to \ensuremath{\textbf{H}} thereby
``over tainting'' the public data in scope.

\begin{figure}[t] % labeled
\small
\begin{hscode}\SaveRestoreHook
\column{B}{@{}>{\hspre}l<{\hspost}@{}}%
\column{6}{@{}>{\hspre}l<{\hspost}@{}}%
\column{17}{@{}>{\hspre}l<{\hspost}@{}}%
\column{E}{@{}>{\hspre}l<{\hspost}@{}}%
\>[B]{}\Varid{v}{}\<[6]%
\>[6]{}\Coloneqq\cdots{}\<[17]%
\>[17]{}\mid Lb^{\Red{\textsf{{\tiny TCB}}}}\;\Varid{l}\;\Varid{t}{}\<[E]%
\\
\>[B]{}\Varid{t}{}\<[6]%
\>[6]{}\Coloneqq\cdots{}\<[17]%
\>[17]{}\mid \textbf{label}\;\Varid{t}\;\Varid{t}\mid \textbf{unlabel}\;\Varid{t}\mid \textbf{labelOf}\;\Varid{t}{}\<[E]%
\\
\>[B]{}\tau{}\<[6]%
\>[6]{}\Coloneqq\cdots{}\<[17]%
\>[17]{}\mid \Conid{Labeled}\;\tau{}\<[E]%
\\
\>[B]{}\Red{E}{}\<[6]%
\>[6]{}\Coloneqq\cdots{}\<[17]%
\>[17]{}\mid \textbf{label}\;\Red{E}\;\Varid{t}\mid \textbf{unlabel}\;\Red{E}\mid \textbf{labelOf}\;\Red{E}{}\<[E]%
\ColumnHook
\end{hscode}\resethooks

\begin{mathpar}
\inferrule[label]
{ \ensuremath{\Sigma\mathrel{=}(\lcurr,\mathbin{...})}\\
  \ensuremath{\lcurr\;\flows\;\Varid{l}}
}
{
\ensuremath{\conf{\Sigma}{\textbf{\Blue{E}}\;[\mskip1.5mu \textbf{label}\;\Varid{l}\;\Varid{t}\mskip1.5mu]}\lto\conf{\Sigma}{\textbf{\Blue{E}}\;[\mskip1.5mu \mathbf{return}\;(Lb^{\Red{\textsf{{\tiny TCB}}}}\;\Varid{l}\;\Varid{t})\mskip1.5mu]}}
}
\and

\inferrule[unlabel]
{ \ensuremath{\Sigma\mathrel{=}(\lcurr,\mathbin{...})}\\
  \ensuremath{\lcurr'\mathrel{=}\lcurr\;\lub\;\Varid{l}}\\
  \ensuremath{\Sigma'\mathrel{=}(\lcurr',\mathbin{...})}
}
{
\ensuremath{\conf{\Sigma}{\textbf{\Blue{E}}\;[\mskip1.5mu \textbf{unlabel}\;(Lb^{\Red{\textsf{{\tiny TCB}}}}\;\Varid{l}\;\Varid{t})\mskip1.5mu]}\lto\conf{\Sigma'}{\textbf{\Blue{E}}\;[\mskip1.5mu \mathbf{return}\;\Varid{t}\mskip1.5mu]}}
}
\and
\inferrule[labelOf]
{ }
{ \ensuremath{\Red{E}\;[\mskip1.5mu \textbf{labelOf}\;(Lb^{\Red{\textsf{{\tiny TCB}}}}\;\Varid{l}\;\Varid{t})\mskip1.5mu]\lto\Red{E}\;[\mskip1.5mu \Varid{l}\mskip1.5mu]} }
% \and
% \inferrule[upgrade]
% { |c = (lcurr, ...)|\\
%   |l_u = lcurr lub l lub l'|}
% { |conf c (E[upgrade (LabeledTCB l t) l']) --> conf c (E[label l_u t])| }
% \and
% \inferrule[downgrade]
% { |c = (lcurr, ...)|\\
%   |l_d = lcurr glb l glb l'|}
% { |conf c (E[downgrade (LabeledTCB l t) l']) --> conf c (E[label l_d undefined])| }
\end{mathpar}
\caption{Extending base \lio{} with labeled values\label{fig:sos:labeled}.}
\end{figure}

%\concept{labeled values}
To address this limitation, LIO provides \ensuremath{\Conid{Labeled}} values.
A \ensuremath{\Conid{Labeled}} value is a term that is explicitly protected by a label,
other than the current label.
%
%%Intuitively, a term can be %heterogeneously
%%labeled by a label at least as restrictive as |lcurr|.
%  if we can ``transfer'' protection from the current label
% to this |Labeled| value (e.g., by ensuring that the protected value
% cannot be inspected in a (future) less-sensitive context).
%
Figure~\ref{fig:sos:labeled} shows the extension of the base \lio{}
with \ensuremath{\Conid{Labeled}} values.

%\concept{label}
The \ensuremath{\textbf{label}} terminal is used to explicitly label a term.
As rule \ruleref{label} shows,
\ensuremath{\textbf{label}\;\Varid{l}\;\Varid{t}} associates the supplied
label \ensuremath{\Varid{l}} with term \ensuremath{\Varid{t}} by wrapping the term with the \ensuremath{Lb^{\Red{\textsf{{\tiny TCB}}}}}
constructor.
Importantly, it first asserts that the new label (\ensuremath{\Varid{l}}), which will
used to protect \ensuremath{\Varid{t}}, is at least as restricting as the current label,
i.e., \ensuremath{\lcurr\;\flows\;\Varid{l}}.

%\concept{unlabel}
Dually, terminal \ensuremath{\textbf{unlabel}} unwraps explicitly labeled values.
As defined in rule \ruleref{unlabel}, given a labeled value
\ensuremath{Lb^{\Red{\textsf{{\tiny TCB}}}}\;\Varid{l}\;\Varid{t}}, \ensuremath{\textbf{unlabel}} returns the wrapped term \ensuremath{\Varid{t}}.
Since the returned term is no longer explicitly labeled by \ensuremath{\Varid{l}}, and
is instead protected by the current label, \ensuremath{\lcurr} must be at least
as restricting as \ensuremath{\Varid{l}}.
To ensure this, the current label is raised from \ensuremath{\lcurr} to \ensuremath{\lcurr\;\lub\;\Varid{l}}, capturing the fact that the remaining computation might depend on
\ensuremath{\Varid{t}}.
This rule highlights the fact that the current label always
``floats'' above the labels of the values observed by the current
computation.

%\concept{labelOf}
The \ensuremath{\textbf{labelOf}} function provides a means for inspecting the label of a
labeled value.
As detailed by reduction rule \ruleref{labelOf}, given a labeled value
\ensuremath{Lb^{\Red{\textsf{{\tiny TCB}}}}\;\Varid{l}\;\Varid{t}}, the function returns the label \ensuremath{\Varid{l}} protecting term
\ensuremath{\Varid{t}}.
This allows code to check the label of a labeled value before deciding
to unlabel it, and thereby raising the current label.
It it worth noting that regardless of the current label in the
configuration, the label of a labeled value can be
inspected---hence labels are effectively ``public.''\footnote{
  Since labeled values can be nested, this only applies to
  the labels of top-level labeled values.
  Indeed, even these labels are not public---they are protected by the
  current label.
  However, since code can always observe objects labeled at the
  current label, this is akin to being public.
}

A common problem with dynamic IFC systems is \emph{label
creep}~\cite{sabelfeld:language-based-iflow}---the raising of the
current label to a point where the computation can no longer do
anything useful.
%
% Label creep does not compromise security, since the current
% label still protects all data in lexical scope.
%
%Despite not compromising security, label creep could make LIO overly
%restricting when building practical applications.
%---in general, being able to perform computations on
%sensitive data without raising the current label is crucial to
%building practical applications.
%
%\concept{toLabeled}
To avoid label creep, LIO provides \ensuremath{\textbf{toLabeled}} as a way to
allow the current label to be \emph{temporarily} raised during the execution of
a given computation.
% and subsequently restore it.
% executes an
% |LIO| action (that may raise the current label) and subsequently
% restores the current label.
%
We extend the terms and the pure evaluation context as \ensuremath{\Varid{t}\Coloneqq\cdots\;|\;\textbf{toLabeled}\;\Varid{t}\;\Varid{t}} and \ensuremath{\Red{E}\Coloneqq\cdots\;|\;\textbf{toLabeled}\;\Red{E}\;\Varid{t}}, respectively, and give the precise semantics for
\ensuremath{\textbf{toLabeled}} as follows:
\begin{mathpar}
\inferrule[toLabeled]
{
\ensuremath{\Sigma\mathrel{=}(\lcurr,\mathbin{...})}\\
\ensuremath{\lcurr\;\flows\;\Varid{l}} \\
\ensuremath{\conf{\Sigma}{\Varid{t}}\lto^*\conf{\Sigma'}{LIO^{\Red{\textsf{{\tiny TCB}}}}\;\Varid{t}'}} \\
\ensuremath{\Sigma'\mathrel{=}(\lcurr',\mathbin{...})}\\
\ensuremath{\lcurr'\;\flows\;\Varid{l}}\\
\ensuremath{\Sigma''\mathrel{=}\Sigma\;\!\ltimes\!\;\Sigma'}
}
{
\ensuremath{\conf{\Sigma}{\textbf{\Blue{E}}\;[\mskip1.5mu \textbf{toLabeled}\;\Varid{l}\;\Varid{t}\mskip1.5mu]}\lto\conf{\Sigma''}{\textbf{\Blue{E}}\;[\mskip1.5mu \textbf{label}\;\Varid{l}\;\Varid{t}'\mskip1.5mu]}}
}
\end{mathpar}
If the current label at the point of executing
\ensuremath{\textbf{toLabeled}\;\Varid{l}\;\Varid{t}} is \ensuremath{\lcurr}, \ensuremath{\textbf{toLabeled}} evaluates \ensuremath{\Varid{t}} to completion
(\ensuremath{\conf{\Sigma}{\Varid{t}}\lto^*\conf{\Sigma'}{LIO^{\Red{\textsf{{\tiny TCB}}}}\;\Varid{t}'}}) and
restores the current label to \ensuremath{\lcurr}, i.e., \ensuremath{\textbf{toLabeled}} provides a
separate context in which \ensuremath{\Varid{t}} is evaluated.
(Here, the state merge function \ensuremath{\!\ltimes\!} is defined as: \ensuremath{\Sigma\;\!\ltimes\!\;\Sigma'\triangleq{}\Sigma}, in the next section we present an alternative definition.)
We note that returning the result of evaluating \ensuremath{\Varid{t}} directly (e.g., as
\ensuremath{\conf{\Sigma}{\textbf{\Blue{E}}\;[\mskip1.5mu \textbf{toLabeled}\;\Varid{l}\;\Varid{t}\mskip1.5mu]}\lto\conf{\Sigma''}{\textbf{\Blue{E}}\;[\mskip1.5mu \Varid{t}'\mskip1.5mu]}}) would allow for
trivial leaks; thus, \ensuremath{\textbf{toLabeled}} labels \ensuremath{\Varid{t}'} with \ensuremath{\Varid{l}} (\ensuremath{\conf{\Sigma''}{\textbf{\Blue{E}}\;[\mskip1.5mu \textbf{label}\;\Varid{l}\;\Varid{t}'\mskip1.5mu]}}).
This effectively states that the result of \ensuremath{\Varid{t}} is protected by label
\ensuremath{\Varid{l}}, as opposed to the current label (\ensuremath{\lcurr'}) at the point \ensuremath{\Varid{t}} completed.
Importantly, this requires that the result not be more sensitive than
\ensuremath{\Varid{l}}, i.e., \ensuremath{\lcurr'\;\flows\;\Varid{l}}.

% %% This can be removed!
% \concept{upgrade with toLabeled}
% Note that with |toLabeled| we can now define |upgrade| as follows:
% \begin{code}
% upgrade t1 t2 @= getLabel >>= \l.
%   toLabeled (labelOf t1 lub l lub t2) (unlabel t1)
% \end{code}
% However, since |toLabeled| needs to be modified to accommodate for
% flow-sensitive references with auto-upgrading, we leave our definition
% |upgrade| as a given in Figure~\ref{fig:sos:labeled}.\footnote{
% In section~\ref{sec:flow-sensitive}, we define an alternative version
% of |toLabeled|, in terms of which we can define the upgrade function
% such that its semantics are equivalent to those given
% Figure~\ref{fig:sos:labeled}. This definition is:
% |upgrade t1 t2 @= getLabel >>= \l.
%   toLabeled' (labelOf t1 lub l lub t2) nil (unlabel t1)|
% }

\subsection{Labeled references}

%\concept{FI refs}
To complete the description of LIO, we extend the \lio{} calculus  with
mutable, flow-insensitive references.
Conceptually, flow-insensitive references are simply mutable \ensuremath{\Conid{Labeled}} values.
Like labeled values, the label of a reference is immutable and serves to
protect the underlying term.
%
%(Recall that |upgrade| for labeled values create another immutable
%value with the new label.)
%
The immutable label makes the semantics straightforward: writing a
term to a reference amounts to ensuring that the reference label is as
restrictive as the current label, i.e., the reference label must be
above the current label; reading from a reference taints the current
label with the reference label.
%
% In this paper we also introduce a flow-sensitive version of
% references, where the reference label itself is also mutable; the
% semantics for these references are described in
% Section~\ref{sec:flow-sensitive}.

\begin{figure}[t]
\small
\begin{hscode}\SaveRestoreHook
\column{B}{@{}>{\hspre}l<{\hspost}@{}}%
\column{6}{@{}>{\hspre}l<{\hspost}@{}}%
\column{17}{@{}>{\hspre}l<{\hspost}@{}}%
\column{E}{@{}>{\hspre}l<{\hspost}@{}}%
\>[B]{}\Varid{v}{}\<[6]%
\>[6]{}\Coloneqq\cdots{}\<[17]%
\>[17]{}\mid Ref^{\Red{\textsf{{\tiny TCB}}}}_{\textsc{fi}}\;\Varid{l}\;\Varid{a}{}\<[E]%
\\
\>[B]{}\Varid{t}{}\<[6]%
\>[6]{}\Coloneqq\cdots{}\<[17]%
\>[17]{}\mid \textbf{newRef}_{\Varid{s}}\;\Varid{t}\;\Varid{t}\mid \textbf{writeRef}_{\Varid{s}}\;\Varid{t}\;\Varid{t}\mid \textbf{readRef}_{\Varid{s}}\;\Varid{t}{}\<[E]%
\\
\>[17]{}\mid \textbf{labelOf}_{\Varid{s}}\;\Varid{t}\mid \textbf{copyRef}\;\Varid{t}\;\Varid{t}{}\<[E]%
\\
\>[B]{}\tau{}\<[6]%
\>[6]{}\Coloneqq\cdots{}\<[17]%
\>[17]{}\mid Ref_{\Varid{s}}\;\tau{}\<[E]%
\\
\>[B]{}\Red{E}{}\<[6]%
\>[6]{}\Coloneqq\cdots{}\<[17]%
\>[17]{}\mid \textbf{newRef}_{\Varid{s}}\;\Red{E}\;\Varid{t}\mid \textbf{writeRef}_{\Varid{s}}\;\Red{E}\;\Varid{t}\mid \textbf{readRef}_{\Varid{s}}\;\Red{E}{}\<[E]%
\\
\>[17]{}\mid \textbf{labelOf}_{\Varid{s}}\;\Red{E}\mid \textbf{copyRef}\;\Red{E}\;\Varid{t}\mid \textbf{copyRef}\;\Varid{v}\;\Red{E}{}\<[E]%
\ColumnHook
\end{hscode}\resethooks
\begin{center}
\begin{mathpar}
\inferrule[newRef-\ensuremath{\textsc{fi}}]
{ \ensuremath{\Sigma\mathrel{=}(\lcurr,\mu_\textsc{fi},\mathbin{...})} \and
  \ensuremath{\lcurr\;\flows\;\Varid{l}} \ \\
  \ensuremath{\mu_\textsc{fi}'\mathrel{=}\mu_\textsc{fi}\;[\mskip1.5mu \Varid{a}\;\mapsto\;Lb^{\Red{\textsf{{\tiny TCB}}}}\;\Varid{l}\;\Varid{t}\mskip1.5mu]} \and
  \ensuremath{\Sigma'\mathrel{=}(\lcurr,\mu_\textsc{fi}',\mathbin{...})}
}
{ \ensuremath{\conf{\Sigma}{\textbf{\Blue{E}}\;[\mskip1.5mu \textbf{newRef}_{\textsc{fi}}\;\Varid{l}\;\Varid{t}\mskip1.5mu]}\lto\conf{\Sigma'}{\textbf{\Blue{E}}\;[\mskip1.5mu \mathbf{return}\;(Ref^{\Red{\textsf{{\tiny TCB}}}}_{\textsc{fi}}\;\Varid{l}\;\Varid{a})\mskip1.5mu]}} }
{ \ensuremath{\fresh{\Varid{a}}} }

\and

\inferrule[readRef-\ensuremath{\textsc{fi}}]
{ \ensuremath{\Sigma\mathrel{=}(\lcurr,\mu_\textsc{fi},\mathbin{...})}
}
{ \ensuremath{\conf{\Sigma}{\textbf{\Blue{E}}\;[\mskip1.5mu \textbf{readRef}_{\textsc{fi}}\;(Ref^{\Red{\textsf{{\tiny TCB}}}}_{\textsc{fi}}\;\Varid{l}\;\Varid{a})\mskip1.5mu]}\lto\conf{\Sigma}{\textbf{\Blue{E}}\;[\mskip1.5mu \textbf{unlabel}\;\mu_\textsc{fi}\;(\Varid{a})\mskip1.5mu]}} }

\and
\inferrule[writeRef-\ensuremath{\textsc{fi}}]
{ \ensuremath{\Sigma\mathrel{=}(\lcurr,\mu_\textsc{fi},\mathbin{...})} \and
  \ensuremath{\lcurr\;\flows\;\Varid{l}} \\
  \ensuremath{\mu_\textsc{fi}'\mathrel{=}\mu_\textsc{fi}\;[\mskip1.5mu \Varid{a}\;\mapsto\;Lb^{\Red{\textsf{{\tiny TCB}}}}\;\Varid{l}\;\Varid{t}\mskip1.5mu]} \and
  \ensuremath{\Sigma'\mathrel{=}(\lcurr,\mu_\textsc{fi}',\mathbin{...})}
}
{ \ensuremath{\conf{\Sigma}{\textbf{\Blue{E}}\;[\mskip1.5mu \textbf{writeRef}_{\textsc{fi}}\;(Ref^{\Red{\textsf{{\tiny TCB}}}}_{\textsc{fi}}\;\Varid{l}\;\Varid{a})\;\Varid{t}\mskip1.5mu]}\lto\conf{\Sigma'}{\textbf{\Blue{E}}\;[\mskip1.5mu \mathbf{return}\;()\mskip1.5mu]}} }

\and

\inferrule[labelOf-\ensuremath{\textsc{fi}}]
{ }
{ \ensuremath{\Red{E}\;[\mskip1.5mu \textbf{labelOf}_{\textsc{fi}}\;(Ref^{\Red{\textsf{{\tiny TCB}}}}_{\textsc{fi}}\;\Varid{l}\;\Varid{a})\mskip1.5mu])\lto\Red{E}\;[\mskip1.5mu \Varid{l}\mskip1.5mu]} }

\and

\inferrule[copyRef]
{ \ensuremath{\Sigma\mathrel{=}(\lcurr,\mu_\textsc{fi},\mathbin{...})} \and \ensuremath{\Varid{l}_{1}\;\flows\;\Varid{l}_{2}}
  \and \ensuremath{\lcurr\;\flows\;\Varid{l}_{2}} \\
  \and
  \ensuremath{Lb^{\Red{\textsf{{\tiny TCB}}}}\;\Varid{l}_{1}\;\Varid{v}_{1}\mathrel{=}\mu_\textsc{fi}\;(\Varid{a}_{1})} \and
  \ensuremath{\mu_\textsc{fi}'\mathrel{=}\mu_\textsc{fi}\;[\mskip1.5mu \Varid{a}_{2}\;\mapsto\;Lb^{\Red{\textsf{{\tiny TCB}}}}\;\Varid{l}_{2}\;\Varid{v}_{1}\mskip1.5mu]}
  \and \ensuremath{\Sigma'\mathrel{=}(\lcurr,\mu_\textsc{fi}',\mathbin{...})}
}
{ \ensuremath{\conf{\Sigma}{\textbf{\Blue{E}}\;[\mskip1.5mu \textbf{copyRef}\;(Ref^{\Red{\textsf{{\tiny TCB}}}}_{\textsc{fi}}\;\Varid{l}_{1}\;\Varid{a}_{1})\;(Ref^{\Red{\textsf{{\tiny TCB}}}}_{\textsc{fi}}\;\Varid{l}_{2}\;\Varid{a}_{2})\mskip1.5mu]}\lto\conf{\Sigma'}{\textbf{\Blue{E}}\;[\mskip1.5mu \mathbf{return}\;()\mskip1.5mu]}} }

\end{mathpar}
\caption{Extending \lio{} with references.
\label{fig:sos:refs}}
\end{center}
\end{figure}
%\concept{new syntax}
%
The syntactic extensions to our calculus are shown in Figure~\ref{fig:sos:refs}.
We use meta-variable \ensuremath{\Varid{s}} to distinguish flow-insensitive (\textrm{FI})
and flow-sensitive (\textrm{FS}) productions---the latter are described
in Section \ref{sec:flow-sensitive}.
%
%, the semantics for latter are
%described in Section~\ref{sec:flow-sensitive}.
%
We also extend configurations to contain a reference (memory)
store \ensuremath{\mu_\textsc{fi}}: \ensuremath{\Sigma\mathrel{=}(\lcurr,\mu_\textsc{fi},\mathbin{...})};
\ensuremath{\mu_\textsc{fi}} maps memory addresses---spanned over by metavariable \ensuremath{\Varid{a}}---to \ensuremath{\Conid{Labeled}} values.
%
%% Second, we add a flow-insensitive reference value (|LIORefTCB I l a|), and
%% terms for creating (|newRef s|), reading (|readRef s|), modifying (|writeRef
%% s|), and inspecting the label of (|labelOfR s|) a reference.
%

%\concept{semantics for refs}
%
When creating a flow-insensitive reference, \ensuremath{\textbf{newRef}_{\textsc{fi}}\;\Varid{l}\;\Varid{t}} creates a labeled
value that guards \ensuremath{\Varid{t}} with label \ensuremath{\Varid{l}} (\ensuremath{Lb^{\Red{\textsf{{\tiny TCB}}}}\;\Varid{l}\;\Varid{t}}) and stores it in the
memory store at a fresh address \ensuremath{\Varid{a}} (\ensuremath{\mu_\textsc{fi}\;[\mskip1.5mu \Varid{a}\;\mapsto\;Lb^{\Red{\textsf{{\tiny TCB}}}}\;\Varid{l}\;\Varid{t}\mskip1.5mu]}).
Subsequently, the function returns a value of the form \ensuremath{Ref^{\Red{\textsf{{\tiny TCB}}}}_{\textsc{fi}}\;\Varid{l}\;\Varid{a}} which simply encapsulates the reference label and address where the
term is stored.
We remark that since any references created within a \ensuremath{\textbf{toLabeled}}
block may outlive the \ensuremath{\textbf{toLabeled}} block computation,
the merge function used in rule~\ruleref{toLabeled} must also account
for this, i.e.,
\ensuremath{(\lcurr,\mu_\textsc{fi},\mathbin{...})\;\!\ltimes\!\;(\lcurr',\mu_\textsc{fi}',\mathbin{...})\mathrel{=}(\lcurr,\mu_\textsc{fi}',\mathbin{...})}.
%
%% This function restores the current label to that of the outer
%% configuration, but retains the latest memory store.

Rule \ruleref{readRef-\ensuremath{\textsc{fi}}} gives the semantics for reading a
labeled reference; reading the term stored at address \ensuremath{\Varid{a}} simply
amounts to unlabeling the value \ensuremath{\mu\;(\Varid{a})} stored at the underlying
address (\ensuremath{\textbf{unlabel}\;\mu_\textsc{fi}\;(\Varid{a})}).

Terminal \ensuremath{\textbf{writeRef}_{\textsc{fi}}} is used to update the memory store with a new
term.
%labeled term  |t| for the reference at location |a|.
%
Note that \ensuremath{\textbf{writeRef}_{\textsc{fi}}} \emph{leaves the label of the reference
intact}, i.e., the label of a flow-insensitive reference is never
changed, but, as rule~\ruleref{writeRef-fi} shows in turn, requires
the current label to be below the reference label when performing the
write (\ensuremath{\lcurr\;\flows\;\Varid{l}}).

Terminal \ensuremath{\textbf{labelOf}_{\textsc{fi}}} has the benefit of allowing code to always inspect the
label of a reference.

Terminal \ensuremath{\textbf{copyRef}} is to copy the contents of one reference to
another, without inspecting the contents of either reference.
As given by rule~\ruleref{copyRef}, the function copies the contents
of a labeled reference into another one, as long as the
source-reference label (\ensuremath{\Varid{l}_{1}}) flows to the target-reference label
(\ensuremath{\Varid{l}_{2}}) and the usual condition for writing to an entity with label
\ensuremath{\Varid{l}_{2}} also holds (\ensuremath{\lcurr\;\flows\;\Varid{l}_{2}}). Since the computation does not
read the source reference, the current label remains unchanged.  We
remarks that while \ensuremath{\textbf{copyRef}} can be encoded using \ensuremath{\textbf{toLabeled}}, we
introduce \ensuremath{\textbf{copyRef}} explicitly since the use of \ensuremath{\textbf{toLabeled}} is
prohibited in concurrent settings and our results rely on such a
feature in both contexts (see Section~\ref{sec:soundness}),

\section{Flow-sensitivity extensions}
\label{sec:flow-sensitive}

%\concept{a more-LIO focused example for why FS}
The flow-insensitive references described in the previous section
are inflexible.
Consider, for example, an application that uses a reference as a log.
Since the log may contain sensitive information, it is important that the
reference be labeled.
Equally important is to be able to read the log at any point in the program to,
for instance, save it to a file.
Although labeling the reference with the top element in the security
lattice ($\top$) would always allow writes to the log, and \ensuremath{\textbf{toLabeled}}
can be used to read the log and then write it to a file, this is
unsatisfactory: it assumes the existence of a top element, which in
some practical IFC systems, including HiStar~\cite{zeldovich:histar}
and Hails~\cite{giffin:hails}, does not exist.
Moreover, it almost always over-approximates the sensitivity of the
log.
Hence, for example, a computation that never reads sensitive data, yet
wishes to read the log content as to send error message to a user over
the network (e.g., as done in a web application) cannot do so---LIO
prevents the computation from reading the log, which results in the
computation getting tainted by $\top$, and subsequently writing to the
network.\footnote{
  Here, as in most IFC systems, we assume the network is public.
}
It is clear that even for such a simple use case, having references
with labels that vary according to the sensitivity of what is stored
in the reference is useful.

%\concept{naive approach doesn't work}
% Ale: I don't like implicit since it might be associated to implicit flows
However, naively implementing flow-sensitive references can effectively
introduce label changes as a covert channel.
%similar to other language-based system to LIO.
%
Suppose that we allow for the label of a reference to be raised to the
current label at the time of the \ensuremath{\textbf{writeRef}_{{}}}.
So, for example, if the label of our log reference is \ensuremath{\textbf{L}} and the
computation has read sensitive data (such that the current label is
\ensuremath{\textbf{H}}), subsequently writing to the log will raise the label of the
reference to \ensuremath{\textbf{H}}.
Unfortunately, while this may appear safe, as previously shown
in~\cite{Austin:Flanagan:PLAS09,Russo:2010, Austin:Flanagan:PLAS10},
the approach is unsound.

The code fragment in Figure~\ref{fig:fs-attack} defines a function,
\ensuremath{\Varid{leakRef}}, that can be used to leak the contents of a reference by
leveraging the newly introduced covert channel: the label of references.
(In this and future examples we use function \ensuremath{\textbf{when}} to denote an \ensuremath{\mathbf{if}}
statement without the \ensuremath{\mathbf{else}} branch and \ensuremath{(\mathbin{\$})} as lightweight notation
for function application, i.e., \ensuremath{\Varid{f}\mathbin{\$}\Varid{x}} is the same as \ensuremath{\Varid{f}\;(\Varid{x})}.)
\begin{figure}[t]%{0.7\columnwidth}
\vspace{-0.5em}
\small
\begin{hscode}\SaveRestoreHook
\column{B}{@{}>{\hspre}l<{\hspost}@{}}%
\column{3}{@{}>{\hspre}l<{\hspost}@{}}%
\column{9}{@{}>{\hspre}l<{\hspost}@{}}%
\column{24}{@{}>{\hspre}l<{\hspost}@{}}%
\column{E}{@{}>{\hspre}l<{\hspost}@{}}%
\>[B]{}\Varid{leakRef}\mathbin{::}Ref_{\textsf{{\tiny FS}}}\;{}\;\Conid{Bool}\to \Conid{LIO}\;\Conid{Bool}{}\<[E]%
\\
\>[B]{}\Varid{leakRef}\;\Varid{href}\mathrel{=}\mathbf{do}{}\<[E]%
\\
\>[B]{}\hsindent{3}{}\<[3]%
\>[3]{}\Varid{tmp}{}\<[9]%
\>[9]{}\leftarrow \textbf{newRef}_{{}}\;\textbf{L}\;(){}\<[E]%
\\
\>[B]{}\hsindent{3}{}\<[3]%
\>[3]{}\textbf{toLabeled}\;\textbf{H}\mathbin{\$}\mathbf{do}\;{}\<[24]%
\>[24]{}\Varid{h}\leftarrow \textbf{readRef}_{{}}\;\Varid{href}{}\<[E]%
\\
\>[24]{}\textbf{when}\;\Varid{h}\mathbin{\$}\textbf{writeRef}_{{}}\;\Varid{tmp}\;(){}\<[E]%
\\
\>[B]{}\hsindent{3}{}\<[3]%
\>[3]{}\mathbf{return}\mathbin{\$}\textbf{labelOf}_{{}}\;\Varid{tmp}\equiv \textbf{H}{}\<[E]%
\ColumnHook
\end{hscode}\resethooks
\cut{$}
\caption{Attack on LIO with naive treatment of flow-sensitive references. We omit
subscripts for clarity.  \label{fig:fs-attack}}
\end{figure}
To illustrate an attack, suppose that the current label is public
(\ensuremath{\textbf{L}}) and \ensuremath{\Varid{leakRef}} is called with a secret (\ensuremath{\textbf{H}}) reference
(\ensuremath{\Varid{href}}).
\ensuremath{\Varid{leakRef}} first creates a public reference \ensuremath{\Varid{tmp}} and, then, within the
\ensuremath{\textbf{toLabeled}} block---which is used to ensure that the current label remains
\ensuremath{\textbf{L}}---the label of this reference is changed to \ensuremath{\textbf{H}} if the secret stored in
\ensuremath{\Varid{href}} is \ensuremath{\Conid{True}}, and left intact (\ensuremath{\textbf{L}}) if the secret is \ensuremath{\Conid{False}}.
The value stored in \ensuremath{\Varid{href}} is revealed by simply inspecting the label
of the \ensuremath{\Varid{tmp}} reference.\footnote{%
The use of \ensuremath{\textbf{labelOf}_{{}}} is not fundamental to this attack and in
Appendix~\ref{sec:app:fs-attack2} we show an alternative attack that
does not rely on such label inspection.
}

%\concept{KEY INSIGHT:}
Fundamentally, the label protecting the \emph{label} of an object,
such as a reference or labeled value, is the current label \ensuremath{\lcurr} at
the time of creation.
Hence, to modify the label of the object within some context (e.g., \ensuremath{\textbf{toLabeled}}
block) wherein the current label is \ensuremath{\lcurr'}, \emph{it must be the case that} \ensuremath{\lcurr'\;\flows\;\lcurr}, i.e., we must be able to write data at sensitivity level
\ensuremath{\lcurr'} into an entity---the label of the object---labeled \ensuremath{\lcurr}.
This restriction is especially important if \ensuremath{\lcurr\;\sqsubset\;\lcurr'} and we can restore the current label from \ensuremath{\lcurr'} to \ensuremath{\lcurr},
since a leak would then be observable within the program itself.
In the case where the label of the object is immutable, as is the case for
flow-insensitive references (and labeled values), this is not a concern: even if
the current label is raised to \ensuremath{\lcurr'} and then restored to \ensuremath{\lcurr}, we do not
learn any information more sensitive than \ensuremath{\lcurr}---the label of the label at the
time of creation---by inspecting the label of the reference (or
value): the label has not changed!

\begin{figure} % semantics
\small
\begin{hscode}\SaveRestoreHook
\column{B}{@{}>{\hspre}l<{\hspost}@{}}%
\column{6}{@{}>{\hspre}l<{\hspost}@{}}%
\column{17}{@{}>{\hspre}l<{\hspost}@{}}%
\column{E}{@{}>{\hspre}l<{\hspost}@{}}%
\>[B]{}\Varid{v}{}\<[6]%
\>[6]{}\Coloneqq\cdots{}\<[17]%
\>[17]{}\mid Ref^{\Red{\textsf{{\tiny TCB}}}}_{\textsc{fs}}\;\Varid{t}{}\<[E]%
\\
\>[B]{}\Varid{t}{}\<[6]%
\>[6]{}\Coloneqq\cdots{}\<[17]%
\>[17]{}\mid \textbf{upgrade}_\textsc{fs}\;\Varid{t}\;\Varid{t}\mid {\Uparrow}\mid \textbf{downgrade}_\textsc{fs}\;\Varid{t}\;\Varid{t}{}\<[E]%
\\
\>[B]{}\Red{E}{}\<[6]%
\>[6]{}\Coloneqq\cdots{}\<[17]%
\>[17]{}\mid \textbf{upgrade}_\textsc{fs}\;\Red{E}\;\Varid{t}\mid \textbf{upgrade}_\textsc{fs}\;\Varid{v}\;\Red{E}{}\<[E]%
\\
\>[17]{}\mid \textbf{downgrade}_\textsc{fs}\;\Red{E}\;\Varid{t}\mid \textbf{downgrade}_\textsc{fs}\;\Varid{v}\;\Red{E}{}\<[E]%
\ColumnHook
\end{hscode}\resethooks

\begin{mathpar}
\inferrule[newRef-\ensuremath{\textsc{fs}}]
{ \ensuremath{\Sigma\mathrel{=}(\lcurr,\mu_\textsc{fi},\mu_\textsc{fs})}\\
  \ensuremath{\lcurr\;\flows\;\Varid{l}}\\
  \ensuremath{\fresh{\Varid{a}}}\\
  \ensuremath{\mu_\textsc{fs}'\mathrel{=}\mu_\textsc{fs}\;[\mskip1.5mu \Varid{a}\;\mapsto\;Lb^{\Red{\textsf{{\tiny TCB}}}}\;\lcurr\;(Lb^{\Red{\textsf{{\tiny TCB}}}}\;\Varid{l}\;\Varid{t})\mskip1.5mu]}\\
  \ensuremath{\Sigma'\mathrel{=}(\lcurr,\mu_\textsc{fi},\mu_\textsc{fs}')}
}
{ \ensuremath{\conf{\Sigma}{\textbf{\Blue{E}}\;[\mskip1.5mu \textbf{newRef}_{\textsc{fs}}\;\Varid{l}\;\Varid{t}\mskip1.5mu]}\lto\conf{\Sigma'}{\textbf{\Blue{E}}\;[\mskip1.5mu \mathbf{return}\;(Ref^{\Red{\textsf{{\tiny TCB}}}}_{\textsc{fs}}\;\Varid{a})\mskip1.5mu]}} }

\and

\inferrule[readRef-\ensuremath{\textsc{fs}}]
{ \ensuremath{\Sigma\mathrel{=}(\lcurr,\mu_\textsc{fi},\mu_\textsc{fs})} \\
  \ensuremath{\mu_\textsc{fs}\;(\Varid{a})\mathrel{=}Lb^{\Red{\textsf{{\tiny TCB}}}}\;\Varid{l}\;(Lb^{\Red{\textsf{{\tiny TCB}}}}\;\Varid{l'}\;\Varid{t})}\\
  \ensuremath{\Varid{l''}\mathrel{=}\Varid{l}\;\lub\;\Varid{l'}}
}
{ \ensuremath{\conf{\Sigma}{\textbf{\Blue{E}}\;[\mskip1.5mu \textbf{readRef}_{\textsc{fs}}\;(Ref^{\Red{\textsf{{\tiny TCB}}}}_{\textsc{fs}}\;\Varid{a})\mskip1.5mu]}\lto\conf{\Sigma}{\textbf{\Blue{E}}\;[\mskip1.5mu \textbf{unlabel}\;(Lb^{\Red{\textsf{{\tiny TCB}}}}\;\Varid{l''}\;\Varid{t})\mskip1.5mu]}} }

\and

\inferrule[writeRef-\ensuremath{\textsc{fs}}]
{ \ensuremath{\Sigma\mathrel{=}(\lcurr,\mu_\textsc{fi}\;\mu_\textsc{fs})}\\
  \ensuremath{\mu_\textsc{fs}\;(\Varid{a})\mathrel{=}Lb^{\Red{\textsf{{\tiny TCB}}}}\;\Varid{l}\;(Lb^{\Red{\textsf{{\tiny TCB}}}}\;\Varid{l'}\;\Varid{t}')} \\
  \ensuremath{\lcurr\;\flows\;(\Varid{l}\;\lub\;\Varid{l'})}\\
  \ensuremath{\mu_\textsc{fs}'\mathrel{=}\mu_\textsc{fs}\;[\mskip1.5mu \Varid{a}\;\mapsto\;Lb^{\Red{\textsf{{\tiny TCB}}}}\;\Varid{l}\;(Lb^{\Red{\textsf{{\tiny TCB}}}}\;\Varid{l'}\;\Varid{t})\mskip1.5mu]}\\
  \ensuremath{\Sigma'\mathrel{=}(\lcurr,\mu_\textsc{fi},\mu_\textsc{fs}')}
}
{ \ensuremath{\conf{\Sigma}{\textbf{\Blue{E}}\;[\mskip1.5mu \textbf{writeRef}_{\textsc{fs}}\;(Ref^{\Red{\textsf{{\tiny TCB}}}}_{\textsc{fs}}\;\Varid{a})\;\Varid{t}\mskip1.5mu]}\lto\conf{\Sigma'}{\textbf{\Blue{E}}\;[\mskip1.5mu \mathbf{return}\;()\mskip1.5mu]}} }

\and

\inferrule[writeRef-\ensuremath{\textsc{fs}}-fail]
{ \ensuremath{\Sigma\mathrel{=}(\lcurr,\mu_\textsc{fi}\;\mu_\textsc{fs})}\\
  \ensuremath{\mu_\textsc{fs}\;(\Varid{a})\mathrel{=}Lb^{\Red{\textsf{{\tiny TCB}}}}\;\Varid{l}\;(Lb^{\Red{\textsf{{\tiny TCB}}}}\;\Varid{l'}\;\Varid{t}')} \\
  \ensuremath{\lcurr\;\not\flows\;(\Varid{l}\;\lub\;\Varid{l'})}
}
{ \ensuremath{\conf{\Sigma}{\textbf{\Blue{E}}\;[\mskip1.5mu \textbf{writeRef}_{\textsc{fs}}\;(Ref^{\Red{\textsf{{\tiny TCB}}}}_{\textsc{fs}}\;\Varid{a})\;\Varid{t}\mskip1.5mu]}\lto\conf{\Sigma'}{\textbf{\Blue{E}}\;[\mskip1.5mu \textbf{unlabel}\;(Lb^{\Red{\textsf{{\tiny TCB}}}}\;\Varid{l}\;{\Uparrow})\mskip1.5mu]}} }

\and

\inferrule[labelOf-\ensuremath{\textsc{fs}}]
{ \ensuremath{\Sigma\mathrel{=}(\lcurr,\mu_\textsc{fi},\mu_\textsc{fs})}\\
  \ensuremath{\mu_\textsc{fs}\;(\Varid{a})\mathrel{=}Lb^{\Red{\textsf{{\tiny TCB}}}}\;\Varid{l}\;(Lb^{\Red{\textsf{{\tiny TCB}}}}\;\Varid{l'}\;\Varid{t})}
}
{ \ensuremath{\conf{\Sigma}{\textbf{\Blue{E}}\;[\mskip1.5mu \textbf{labelOf}_{\textsc{fs}}\;(Ref^{\Red{\textsf{{\tiny TCB}}}}_{\textsc{fs}}\;\Varid{a})\mskip1.5mu]}\lto\conf{\Sigma}{\textbf{\Blue{E}}\;[\mskip1.5mu \textbf{unlabel}\;(Lb^{\Red{\textsf{{\tiny TCB}}}}\;\Varid{l}\;\Varid{l'})\mskip1.5mu]}}}

\and

\inferrule[upgradeRef]
{ \ensuremath{\Sigma\mathrel{=}(\lcurr,\mu_\textsc{fi},\mu_\textsc{fs})}\\
  \ensuremath{\mu_\textsc{fs}\;(\Varid{a})\mathrel{=}Lb^{\Red{\textsf{{\tiny TCB}}}}\;\Varid{l}\;(Lb^{\Red{\textsf{{\tiny TCB}}}}\;\Varid{l''}\;\Varid{v})}\\
  \ensuremath{\lcurr\;\flows\;\Varid{l}} \\
%  |lcurr canFlowTo (lcurr lub l'' lub l')|
  \ensuremath{\mu_\textsc{fs}'\mathrel{=}\mu_\textsc{fs}\;[\mskip1.5mu \Varid{a}\;\mapsto\;Lb^{\Red{\textsf{{\tiny TCB}}}}\;\Varid{l}\;(Lb^{\Red{\textsf{{\tiny TCB}}}}\;(\Varid{l''}\;\lub\;\Varid{l'})\;\Varid{v})\mskip1.5mu]}\\
  \ensuremath{\Sigma'\mathrel{=}(\lcurr,\mu_\textsc{fi},\mu_\textsc{fs}')}
}
{ \ensuremath{\conf{\Sigma}{\textbf{\Blue{E}}\;[\mskip1.5mu \textbf{upgrade}_\textsc{fs}\;(Ref^{\Red{\textsf{{\tiny TCB}}}}_{\textsc{fs}}\;\Varid{a})\;\Varid{l'}\mskip1.5mu]}\lto\conf{\Sigma'}{\textbf{\Blue{E}}\;[\mskip1.5mu \mathbf{return}\;()\mskip1.5mu]}} }

\and

\inferrule[downgradeRef]
{ \ensuremath{\Sigma\mathrel{=}(\lcurr,\mu_\textsc{fi},\mu_\textsc{fs})}\\
  \ensuremath{\mu_\textsc{fs}\;(\Varid{a})\mathrel{=}Lb^{\Red{\textsf{{\tiny TCB}}}}\;\Varid{l}\;(Lb^{\Red{\textsf{{\tiny TCB}}}}\;\Varid{l''}\;\Varid{v})}\\
  \ensuremath{\lcurr\;\flows\;\Varid{l}} \\
%  |lcurr canFlowTo (lcurr lub l'' lub l')|
  \ensuremath{\mu_\textsc{fs}'\mathrel{=}\mu_\textsc{fs}\;[\mskip1.5mu \Varid{a}\;\mapsto\;Lb^{\Red{\textsf{{\tiny TCB}}}}\;\Varid{l}\;(Lb^{\Red{\textsf{{\tiny TCB}}}}\;(\Varid{l}\;\lub\;(\Varid{l''}\;\glb\;\Varid{l'}))\;{\Uparrow})\mskip1.5mu]}\\
  \ensuremath{\Sigma'\mathrel{=}(\lcurr,\mu_\textsc{fi},\mu_\textsc{fs}')}
}
{ \ensuremath{\conf{\Sigma}{\textbf{\Blue{E}}\;[\mskip1.5mu \textbf{downgrade}_\textsc{fs}\;(Ref^{\Red{\textsf{{\tiny TCB}}}}_{\textsc{fs}}\;\Varid{a})\;\Varid{l'}\mskip1.5mu]}\lto\conf{\Sigma'}{\textbf{\Blue{E}}\;[\mskip1.5mu \mathbf{return}\;()\mskip1.5mu]}} }
\end{mathpar}
\caption{\liofs{}: extension of \lio{} with flow-sensitive references.\label{fig:sos:fs}}
\end{figure}

%\concept{\liofs{}}
Thus, to extend LIO with flow-sensitive references, we must account for the label
on the label of the reference at the time of creation, \ensuremath{\lcurr}.
(This label is, however, immutable.)
In turn, when changing the label of the reference, we must ensure that no data
from the context at the time of the change, whose label is \ensuremath{\lcurr'}, is leaked
into the label of the reference by ensuring that \ensuremath{\lcurr'\;\flows\;\lcurr},
i.e., we can write data labeled \ensuremath{\lcurr'} into the label that is labeled \ensuremath{\lcurr}.

Formally, we extend the \lio{} syntax and reduction rules as shown in
Figure~\ref{fig:sos:fs}; we call this calculus \liofs{}.
To create a flow-sensitive reference \ensuremath{\textbf{newRef}_{\textsc{fs}}\;\Varid{l}\;\Varid{t}} creates a labeled
value that guards \ensuremath{\Varid{t}} with label \ensuremath{\Varid{l}} (\ensuremath{Lb^{\Red{\textsf{{\tiny TCB}}}}\;\Varid{l}\;\Varid{t}}).
Since we wish to allow programmers to modify the label \ensuremath{\Varid{l}} of the
reference, we additionally store the label on \ensuremath{\Varid{l}}, i.e., the current
label \ensuremath{\lcurr}, by simply labeling the already-guarded term (\ensuremath{\mu_\textsc{fs}'\mathrel{=}\mu_\textsc{fs}\;[\mskip1.5mu \Varid{a}\;\mapsto\;Lb^{\Red{\textsf{{\tiny TCB}}}}\;\lcurr\;(Lb^{\Red{\textsf{{\tiny TCB}}}}\;\Varid{l}\;\Varid{t})\mskip1.5mu]}), as shown in
rule~\ruleref{newRef-\ensuremath{\textsc{fs}}}.
Primitive \ensuremath{\textbf{newRef}_{\textsc{fs}}} returns a \ensuremath{Ref^{\Red{\textsf{{\tiny TCB}}}}_{\textsc{fs}}\;\Varid{a}} which simply encapsulates the fresh
reference address where the doubly-labeled term is stored.
Different from the constructor \ensuremath{Ref^{\Red{\textsf{{\tiny TCB}}}}_{\textsc{fi}}},  the constructor \ensuremath{Ref^{\Red{\textsf{{\tiny TCB}}}}_{\textsc{fs}}}
does not encapsulate the label of the reference.
This is precisely because the label of a flow-sensitive reference is
mutable and must be looked up in the store.
As given by rule~\ruleref{labelOf-\ensuremath{\textsc{fs}}}, \ensuremath{\textbf{labelOf}_{\textsc{fs}}} returns the label of the
reference after raising the current label (with \ensuremath{\textbf{unlabel}}) to account for the
fact that the label of the reference \ensuremath{\Varid{l'}} is a value at sensitivity level \ensuremath{\Varid{l}},
i.e., we raise the current label to the join of the current label and the label
on the label.

The rule for reading flow-sensitive references is standard.
As given by rule \ruleref{readREf-fs}, \ensuremath{\textbf{readRef}_{\textsc{fs}}} simply raises the
current label to the join of the reference label and label on the
reference label (\ensuremath{\Varid{l}\;\lub\;\Varid{l'}}) and returns the protected value.
This reflects the fact that the computation is observing both data at
level \ensuremath{\Varid{l}} (the label on the reference) and \ensuremath{\Varid{l'}} (the actual term).

The rule for writing flow-sensitive references deserves more
attention.
\begin{wrapfigure}{r}{0.4\columnwidth}
\begin{hscode}\SaveRestoreHook
\column{B}{@{}>{\hspre}l<{\hspost}@{}}%
\column{5}{@{}>{\hspre}l<{\hspost}@{}}%
\column{E}{@{}>{\hspre}l<{\hspost}@{}}%
\>[B]{}\mathbf{do}\;{}\<[5]%
\>[5]{}\Varid{r}\leftarrow \textbf{newRef}_{\textsc{fs}}\;\textbf{H}\;(){}\<[E]%
\\
\>[5]{}\textbf{readRef}_{\textsc{fs}}\;\Varid{r}{}\<[E]%
\\
\>[5]{}\textbf{writeRef}_{\textsc{fs}}\;\Varid{r}\;(){}\<[E]%
\ColumnHook
\end{hscode}\resethooks
\caption{\small \label{fig:permissive} Permissiveness test.}
\end{wrapfigure}
First, \ensuremath{\textbf{writeRef}_{\textsc{fs}}} as given by rule~\ruleref{writeRef-fs}, ensures
that the current computation can write to the reference by checking
that \ensuremath{\lcurr\;\flows\;(\Varid{l}\;\lub\;\Varid{l'})}.
We impose this condition instead of the two conditions \ensuremath{\lcurr\;\flows\;\Varid{l}} and \ensuremath{\lcurr\;\flows\;\Varid{l'}}---which respectively check that
the current computation can modify both, the label of the reference,
and the reference itself---since it is more permissive, yet still
safe.
When imposing the two conditions independently, certain programs, such
as the one given in Figure~\ref{fig:permissive}, would fail.
In this program, we first create a flow-sensitive reference labeled
\ensuremath{\textbf{H}} when the current label is \ensuremath{\textbf{L}} (and thus the label on \ensuremath{\textbf{H}}
is \ensuremath{\textbf{L}}).
Then, we raise the label by reading from the reference.
Finally, we attempt to write to the reference.
Under our semantics, this program behaves as expected; however, when
imposing the two conditions independently, the write fails---the
current label does not flow to the label on the label of the
reference.

Another case for \ensuremath{\textbf{writeRef}_{\textsc{fs}}} which we must handle is when current
label does not flow to the join of the reference label, i.e.,
\ensuremath{\lcurr\;\not\flows\;\Varid{l}\;\lub\;\Varid{l'}}, and the write is disallowed.
If the semantics simply got stuck, the current label (at the point of
the stuck term) would
not reflect the fact that the success of applying such rule depends on the
label \ensuremath{\Varid{l'}}, which is itself protected by \ensuremath{\Varid{l}}.
Indeed, this might lead to information leaks and we thus provide an
explicit rule, \ruleref{writeref-fs-fail}, for this failure case that
first raises the current label (via \ensuremath{\textbf{unlabel}}) to \ensuremath{\Varid{l}} and then
diverges; in the rule, \ensuremath{{\Uparrow}} represents a divergent term for which
we do not provide a reduction rule.

% there is a termination leak which reveals
% information about the possible changes on label |l'| (itself protected by |l|).
% Observe that those changes could have occurred in a more sensitive context than
% |lcurr|. For sequential LIO, this is not a problem since the security guarantees
% ignore termination leaks~\cite{stefan:lio}. However, in a concurrent setting,
% such leaks are a concern~\cite{stefan:addressing-covert}. Aiming to apply our
% flow-sensitive approach for concurrency almost without modifications (see
% Section~\ref{sec:conc}), we introduce a failure rule for |writeRef S|
% which
%does not reflect the fact
%that the rule relies, and thus observes, the label on the reference |l'| which
%is itself protected by |l|
%
%The effects of this rule in a sequential settings is analogous to getting
%stuck\footnote{This is not true if we consider exceptions, a feature
%out of the scope of this paper}.
Note that \ensuremath{\textbf{writeRef}_{\textsc{fs}}} does not modify the label of the reference.
This is, in part, because we wish to keep the difference between
flow-insensitive and flow-sensitive references as small as possible.
Instead, we provide \ensuremath{\textbf{upgrade}_\textsc{fs}} precisely for this purpose; this primitive is
used to raise the label of a reference.
Rule \ruleref{upgradeRef} is straight forward---it simply ensures that the
current computation can modify the label of the reference by checking that the
current label flows to the label on the label (\ensuremath{\lcurr\;\flows\;\Varid{l}}).
Similarly, \ensuremath{\textbf{downgrade}_\textsc{fs}} is used to lower the label of the reference,
destroying its contents, i.e., replacing its value with \ensuremath{{\Uparrow}}.
%\Red{Why don't we just use |diverge| instead of |undefined|?}
%\Red{\textbf{IMPORTANT:} The rule should paranthesize |l'' glb l'|
%since as is it reads left to right and that's not correct.}
%Rule \ruleref{downgradeRef} is analogous to
%rule \ruleref{upgradeRef}, where the primitive |downgrade| is used
%instead of |upgrade|.
%
Rules \ruleref{upgradeRef} and \ruleref{downgradeRef} are analogous;
the main difference is that the former uses the join operation to
combine the old and new labels (\ensuremath{\Varid{l''}\;\lub\;\Varid{l'}}), whereas the latter uses
the meet operation (\ensuremath{\Varid{l''}\;\glb\;\Varid{l'}}).
The \ensuremath{\textbf{downgrade}_\textsc{fs}} primitive is useful when one
wishes to store information that is less sensitive into a reference.
%\hl{(Pablo) Should we also make a
%  forward-reference to the encoding of NSU?}
%
Both \ensuremath{\textbf{upgrade}_\textsc{fs}} and \ensuremath{\textbf{downgrade}_\textsc{fs}} highlight that it
is safe to raise or lower the label of a flow-sensitive reference,
if that the label on the label still flows to the final label in the
nested \ensuremath{Lb^{\Red{\textsf{{\tiny TCB}}}}} structure.

% \hl{(Ale)
%   Explain downgrade like ``Similar, downgrade...''. We should emphasize that
%   change upwards or downwards in the label is possible as long as the labeled
%   value does not go below the lol!}

\subsection{Automatic upgrades}
\label{sec:flow-sensitive:auto}

%\concept{upgrade is burdensome}
We can use \liofs{} to implement various applications that rely on
flow-sensitive references, even those that rely on policies such as the
popular no-sensitive upgrades~\cite{Austin:Flanagan:PLAS09}.
%
%\hl{Ale: to do}
%(In Section~\ref{sec:related}, we describe the encoding of a policy that is
%similar to, but more permissive than, no-sensitive upgrades.)
%
Using \liofs{}, we can also safely implement our logging application using a
flow-sensitive reference.
Unfortunately, our system (and others like it) requires that we insert
\ensuremath{\textbf{upgrade}}s before we raise the current label so that it is possible to
write references in a more-sensitive context, e.g., to modify a public reference
after reading a secret.
In the case of the logging example, we would need to upgrade the label
before reading any sensitive data, if we later wish to write to
the log.

\begin{figure}[t] % auto-fs
\small
\begin{mathpar}
\inferrule[upgradeStore]
{\ensuremath{\Sigma\mathrel{=}(\lcurr,\mu_\textsc{fi},\mu_\textsc{fs})}\\
 \ensuremath{\mu_\textsc{fs}\mathrel{=}\{\mskip1.5mu \Varid{a}_{1}\;\mapsto\;\Varid{v}_{1},\mathbin{...},\Varid{a}_n\;\mapsto\;\Varid{v}_n\mskip1.5mu\}}\\
 \ensuremath{\Varid{t}_\Varid{i}\mathrel{=}\textbf{upgrade}_\textsc{fs}\;(Ref^{\Red{\textsf{{\tiny TCB}}}}_{\textsc{fs}}\;\Varid{a}_\Varid{i})\;\Varid{l},\Varid{i}\mathrel{=}\mathrm{1},\mathbin{...},\Varid{n}}
}
{\ensuremath{\conf{\Sigma}{\textbf{\Blue{E}}\;[\mskip1.5mu \textbf{upgradeStore}_\textsc{fs}\;\Varid{l}\mskip1.5mu]}\lto\conf{\Sigma}{\textbf{\Blue{E}}\;[\mskip1.5mu \Varid{t}_{1}\sequ \mathbin{...}\sequ \Varid{t}_\Varid{n}\mskip1.5mu]}}}
\and
\inferrule[unlabel-au]
{ \ensuremath{\Sigma\mathrel{=}(\lcurr,\mu_\textsc{fi},\mu_\textsc{fs})}\\
  \ensuremath{\lcurr'\mathrel{=}\lcurr\;\lub\;\Varid{l}}\\
  \ensuremath{\conf{\Sigma}{\textbf{upgradeStore}_\textsc{fs}\;\lcurr'}\lto^*\conf{\lcurr,\mu_\textsc{fi},\mu_\textsc{fs}'}{LIO^{\Red{\textsf{{\tiny TCB}}}}\;()}}\\
  \ensuremath{\Sigma'\mathrel{=}(\lcurr',\mu_\textsc{fi},\mu_\textsc{fs}')}
}
{
\ensuremath{\conf{\Sigma}{\textbf{\Blue{E}}\;[\mskip1.5mu \textbf{unlabel}\;(Lb^{\Red{\textsf{{\tiny TCB}}}}\;\Varid{l}\;\Varid{t})\mskip1.5mu]}\lto\conf{\Sigma'}{\textbf{\Blue{E}}\;[\mskip1.5mu \mathbf{return}\;\Varid{t}\mskip1.5mu]}}
}
\end{mathpar}
\caption{\lioafs{}: Extending \liofs{} with auto-upgrades.\label{fig:sos:afs}}
\end{figure}

\begin{figure}[t] % auto-fs
\small
\begin{hscode}\SaveRestoreHook
\column{B}{@{}>{\hspre}l<{\hspost}@{}}%
\column{6}{@{}>{\hspre}l<{\hspost}@{}}%
\column{17}{@{}>{\hspre}l<{\hspost}@{}}%
\column{27}{@{}>{\hspre}l<{\hspost}@{}}%
\column{30}{@{}>{\hspre}l<{\hspost}@{}}%
\column{51}{@{}>{\hspre}l<{\hspost}@{}}%
\column{E}{@{}>{\hspre}l<{\hspost}@{}}%
\>[B]{}\Varid{v}{}\<[6]%
\>[6]{}\Coloneqq\cdots{}\<[17]%
\>[17]{}\mid \overline{{\Varid{v},\mathbin{...}}}\mid \epsilon_{\,\overline{\tau, ...}}{}\<[E]%
\\
\>[B]{}\Varid{t}{}\<[6]%
\>[6]{}\Coloneqq\cdots{}\<[17]%
\>[17]{}\mid \overline{{\Varid{t},\mathbin{...}}}\mid \textbf{withRefs}_\textsc{fs}\;\Varid{t}\;\Varid{t}{}\<[E]%
\\
\>[B]{}\tau{}\<[6]%
\>[6]{}\Coloneqq\cdots{}\<[17]%
\>[17]{}\mid \overline{{\tau,\mathbin{...}}}{}\<[E]%
\\
\>[B]{}\Red{E}{}\<[6]%
\>[6]{}\Coloneqq\cdots{}\<[17]%
\>[17]{}\mid \overline{{\Red{E},\Varid{t},\mathbin{...}}}\mid \overline{{\Varid{v},\Red{E},\Varid{t},\mathbin{...}}}\mid \textbf{withRefs}_\textsc{fs}\;\Red{E}\;\Varid{t}{}\<[E]%
\\[\blanklineskip]%
\>[B]{}\Varid{addrs}(\epsilon_{\,\overline{\tau, ...}}){}\<[51]%
\>[51]{}\triangleq{}\emptyset{}\<[E]%
\\
\>[B]{}\Varid{addrs}(\overline{{Ref^{\Red{\textsf{{\tiny TCB}}}}_{\textsc{fs}}\;\Varid{a}_{1},Ref^{\Red{\textsf{{\tiny TCB}}}}_{\textsc{fs}}\;\Varid{a}_{2},\mathbin{...}}}){}\<[51]%
\>[51]{}\triangleq{}\{\mskip1.5mu \Varid{a}_{1},\Varid{a}_{2},\mathbin{...}\mskip1.5mu\}{}\<[E]%
\\[\blanklineskip]%
\>[B]{}\Varid{addrs}^+_{\mu}(\epsilon_{\,\overline{\tau, ...}}){}\<[30]%
\>[30]{}\triangleq{}\emptyset{}\<[E]%
\\
\>[B]{}\Varid{addrs}^+_{\mu}(\overline{{\Varid{v}_{1},\Varid{v}_{2},\mathbin{...}}}){}\<[30]%
\>[30]{}\triangleq{}\bigcup\;\{\mskip1.5mu \Varid{addrs}^+_{\mu}(\Varid{v}_{1}),\Varid{addrs}^+_{\mu}(\Varid{v}_{2}),\mathbin{...}\mskip1.5mu\}{}\<[E]%
\\[\blanklineskip]%
\>[B]{}\Varid{addrs}^+_{\mu}(Ref^{\Red{\textsf{{\tiny TCB}}}}_{\textsc{fs}}\;\Varid{a}){}\<[27]%
\>[27]{}\triangleq{}\{\mskip1.5mu \Varid{a}\mskip1.5mu\}\cup\Varid{addrs}^+_{\mu}(\mu\;(\Varid{a})){}\<[E]%
\\
\>[B]{}\Varid{addrs}^+_{\mu}(\Varid{v}){}\<[27]%
\>[27]{}\triangleq{}\emptyset{}\<[E]%
\ColumnHook
\end{hscode}\resethooks
\begin{mathpar}
\inferrule[withRefs-Ctx]
{\ensuremath{\Sigma\mathrel{=}(\lcurr,\mu_\textsc{fi},\mu_\textsc{fs})}\\
 \ensuremath{\mu_\textsc{fs}'\mathrel{=}\{\mskip1.5mu \Varid{a}\;\mapsto\;\mu_\textsc{fs}\;(\Varid{a})\;|\;\Varid{a}\;\in\;\textrm{dom}\;\mu_\textsc{fs}\;\cap\;(\Varid{addrs}^+_{\mu_\textsc{fs}}(\Varid{v}))\mskip1.5mu\}}\\
 \ensuremath{\conf{\lcurr,\mu_\textsc{fi},\mu_\textsc{fs}'}{\textbf{\Blue{E}}\;[\mskip1.5mu \Varid{t}\mskip1.5mu]}\lto\conf{\lcurr',\mu_\textsc{fi}',\mu_\textsc{fs}''}{\textbf{\Blue{E}}\;[\mskip1.5mu \Varid{t}'\mskip1.5mu]}}\\
 \ensuremath{\Sigma''\mathrel{=}(\lcurr',\mu_\textsc{fi}',\mu_\textsc{fs}''\;\!\ltimes\!\;\mu_\textsc{fs})}\\
 \ensuremath{\Varid{v}'\mathrel{=}\Varid{addrs}^{-1}(\textrm{dom}\;\mu_\textsc{fs}'')}
}
{\ensuremath{\conf{\Sigma}{\textbf{\Blue{E}}\;[\mskip1.5mu \textbf{withRefs}_\textsc{fs}\;\Varid{v}\;\Varid{t}\mskip1.5mu]}\lto\conf{\Sigma''}{\textbf{\Blue{E}}\;[\mskip1.5mu \textbf{withRefs}_\textsc{fs}\;\Varid{v}'\;\Varid{t}'\mskip1.5mu]}}}
\and
\inferrule[withRefs-Done]
{ }
{\ensuremath{\conf{\Sigma}{\textbf{\Blue{E}}\;[\mskip1.5mu \textbf{withRefs}_\textsc{fs}\;\Varid{v}\;\Varid{v}'\mskip1.5mu]}\lto\conf{\Sigma}{\textbf{\Blue{E}}\;[\mskip1.5mu \Varid{v}'\mskip1.5mu]}}}
\and
\inferrule[Type-withRef]
{ \ensuremath{\Delta'\mathrel{=}\{\mskip1.5mu \Varid{a}\;\mapsto\;\Delta\;(\Varid{a})\;|\;\Varid{a}\;\in\;\textrm{dom}\;\Delta\;\cap\;(\Varid{addrs}(\Varid{v}))\mskip1.5mu\}}\\
  \ensuremath{\Delta',\Gamma\vdash\Varid{v}\mathbin{:}\overline{{Ref_{\textsc{fs}}\;\tau_{1},\mathbin{...}}}}\\
  \ensuremath{\Delta',\Gamma\vdash\Varid{t}\mathbin{:}\Conid{LIO}\;\tau}
}
{ \ensuremath{\Delta,\Gamma\vdash\textbf{withRefs}_\textsc{fs}\;\Varid{v}\;\Varid{t}\mathbin{:}\Conid{LIO}\;\tau} }
\end{mathpar}
\caption{Extending \liofs{} and \lioafs{} with \ensuremath{\textbf{withRefs}_\textsc{fs}}.\label{fig:sos:withRefs}}
\end{figure}

%\concept{auto upgrade}
%Inspired by~\citep{Hedin13},
We provide an extension to \liofs{} that can
be used to automatically upgrade references.
This extension, called \lioafs{}, is given in Figure~\ref{fig:sos:afs}.
Intuitively, whenever the current label is about to be raised, we first upgrade
all the references in the \ensuremath{\mu_\textsc{fs}} store %, with |upgradeM|, in the store
and then raise the current label.
Rule \ruleref{upgradeStore} upgrades every reference in the flow-sensitive store
\ensuremath{\mu_\textsc{fs}} by executing \ensuremath{\Varid{t}_{\mathrm{1}}\sequ \Varid{t}_{\mathrm{2}}\sequ \cdots\sequ \Varid{t}_\Varid{n}}, where
\ensuremath{\Varid{t}_\Varid{i}\mathrel{=}\textbf{upgrade}_\textsc{fs}\;(Ref^{\Red{\textsf{{\tiny TCB}}}}_{\textsc{fs}}\;\Varid{a}_\Varid{i})\;\Varid{l}}. Term \ensuremath{\Varid{t}\sequ \Varid{t}'} is
similar to bind except that it discards the result produced by
\ensuremath{\Varid{t}}.
%\footnote{Formally, |t >> t'| is defined as |t >>= \x.t'| where |x| does not
%  appear free in |t'|}
%
Since \ensuremath{\textbf{unlabel}} is the only function that raises the current label, we
augment the \ruleref{unlabel} rule with \ruleref{unlabel-au}, given in
Figure~\ref{fig:sos:afs}.
This ensures that as the computation progresses it does not ``lose''
write access to its references.
Returning to our logging example, with auto-upgrades the reference
used as the log never needs to be explicitly upgraded and can always
be written to---an interface expected of a log.

%\concept{store creep}
Recall  that \ensuremath{\textbf{toLabeled}} is used to avoid label creep by allowing code
to only temporarily raise the current label.
Unfortunately, with auto-upgrades, when the current label gets raised
within a \ensuremath{\textbf{toLabeled}} block, the upgrades of the flow-sensitive references remain even
after the current label is restored.
Thus, reading from any flow-sensitive reference after the \ensuremath{\textbf{toLabeled}}
block will raise the current label to (at least) the current label at the end
of the \ensuremath{\textbf{toLabeled}} block (since all references are upgraded every time 
the current label gets raised).
This can be used to carry out a \emph{poison pill}-like
attack~\cite{Breeze}, wherein the (usually untrusted) computation
executing within the \ensuremath{\textbf{toLabeled}} block will render the outer computation useless
via label creep.
(We note that this attack is possible in \liofs{} without the auto-upgrade, but
requires the attacker to manually insert all the upgrades.)

%\concept{withRefs}
To address this issue, we extend \liofs{} (and \lioafs{}) with \ensuremath{\textbf{withRefs}_\textsc{fs}\;\Varid{v}\;\Varid{t}}, which takes a bag (strict heterogeneous list) \ensuremath{\Varid{v}} of references
and a computation \ensuremath{\Varid{t}}, and executes \ensuremath{\Varid{t}} in a configuration where the
flow-sensitive reference store only contains the
subset of references \ensuremath{\Varid{v}} (and any nested references).
This extension and type rule \ruleref{Type-withRef}, which ensures
that a term cannot access a reference outside its store, are shown in
Figure~\ref{fig:sos:withRefs}.
%
% \Red{What about refs stored in a ref? I'm not on 100\% brain power now
%   but we should think about this case explicitly.}

% \Red{[Pablo] According to the typing rule,
%   references stored in references are accessible if the outer
%   reference is. I think this is fine... the user can always further
%   constrain the references with nested withRefs.  However, the
%   reduction rule (and embedding) executes |withRefs v t| in a memory with just the
%   references in |v|, which doesn't match this. We need to somehow hide
%   the disallowed references in the semantics, so that when
%   auto-upgrade kicks in, only the selected references get upgraded;
%   the problem is that now we are hiding too much. I think we could
%   have a closure operator addrs+ instead of just addrs, such that
%   addrs+(v) are the addresses of all the references reachable from
%   v. Then in the SOS rule for withRefs, we would compute mu' by
%   intersecting with addrs+(v) instead of just addrs(v).  }
%
A bag is either empty \ensuremath{\epsilon_{\,\overline{\tau, ...}}}, or it may contain a set of references
of (potentially) distinct types \ensuremath{\overline{{\Varid{v},\mathinner{\ldotp\ldotp}}}}.
Rules \ruleref{withRefs-Ctx} and \ruleref{withRefs-Done} precisely
define the semantics of this new primitive, where the meta-level
function \ensuremath{\Varid{addrs}(\cdot )} converts a bag of references to a set of their
corresponding addresses, \ensuremath{\Varid{addrs}^{-1}(\cdot )} performs the inverse conversion, and $\;\ensuremath{\!\ltimes\!}\;$
is used to merge the stores, giving preferences to the left-hand-side
store, i.e., when there is a discrepancy on a stored value between
both stores, it chooses the one appearing on the left-hand-side.
The function \ensuremath{\Varid{addrs}^+_{\mu}(\cdot )} computes the closure of \ensuremath{\Varid{addrs}(\cdot )} under store
\ensuremath{\mu}, so as to include the addresses of arbitrarily-nested
references. Note that if we did not include these addresses in the
restricted store \ensuremath{\mu_\textsc{fs}'}, evaluation might get stuck if the program
attempted a \ensuremath{\textbf{readRef}_{\textsc{fs}}} operation on a nested reference.
We note that \ruleref{withRefs-Ctx} is triggered until the term under
evaluation is reduced to a value, at which point \ruleref{withRefs-Done} is
triggered, returning said value; we specify this big-step rule in terms of
small-steps to facilitate the formalization of our concurrent calculus (see
Section~\ref{sec:conc}).
%O
Aside from the modeling of bags, the \ensuremath{\textbf{withRefs}_\textsc{fs}} primitive is
straightforward and mostly standard;
indeed, the programming paradigm is similar to that already present in
some mainstream languages (e.g., C++'s lambda closures require the
programmer to specify the captured references).
Lastly, we note that the poison pill attack can now be addressed by
simply wrapping \ensuremath{\textbf{toLabeled}} with \ensuremath{\textbf{withRefs}_\textsc{fs}}, which prevents
(untrusted) code within the \ensuremath{\textbf{toLabeled}} block from upgrading arbitrary
references.

%%%%%%%%%%%%%%%%%%%%%%%%%%%%%%%%%%%%%%%%%%%%%%%%%%%%%%%%%%%%%%%%%%%%%%%%%%%%
%%%%%%%%%%%%%%%%%%%%%%%%%%%%%%%%%%%%%%%%%%%%%%%%%%%%%%%%%%%%%%%%%%%%%%%%%%%%

%% \begin{definition}
%% We define the equivalence relation oblivious to flow-sensitive reference
%% implemnetation as:
%% \begin{itemize}
%% \item |LIORefTCB S t &=  LIORefTCB S t'| for any |t|, |t'|.
%% \item |t &=  t'| iff |t = t'|.
%% \item |LIOTCB t &=  LIOTCB t'| iff |t &=  t'|.
%% \item |LabeledTCB l t &= LabeledTCB l t'| iff |t &=  t'| and |l &=  l'|.
%% \item etc.
%% \end{itemize}
%% \end{definition}
%%
%% \begin{theorem}[Equivalence]
%% For all |t, v| in \lio{} and |t', v'| in \liofs{},
%% and |v &= v'|, it must be that
%% %
%% |conf (lcurr,mI) t  ==> conf (lcurr',...) diverge| $\Leftrightarrow$
%% |conf (lcurr,mI,emptyset) t' ==> conf (lcurr',...) diverge| and
%% %
%% |conf (lcurr,mI) t  ==> conf (lcurr',...) v| $\Leftrightarrow$
%% |conf (lcurr,mI,emptyset) t' ==> conf (lcurr',...) v'|.
%% \end{theorem}

% Local Variables:
% TeX-master: "main.tex"
% TeX-command-default: "Make"
% End:
% 2 Pages (Pablo/Ale) explain sequential

\section{Concurrency}
\label{sec:conc}

\concept{extending \lioafs{}}
In this section, we consider flow-sensitive references in the presence
of concurrency (e.g., a web application in which request-handling
threads share a common log).
Concretely, we extend our sequential \liofs{} and \lioafs{} calculi with
threads and a new terminal, \ensuremath{\textbf{forkLIO}}, for dynamically creating new threads, as
in the concurrent version of LIO~\cite{stefan:addressing-covert}.
Intuitively, this concurrent calculus \lioconc{} simply defines a scheduler over
sequential threads, such that taking a step in the concurrent calculus amounts
to taking a step in a sequential thread and context switching to a different
one.
For brevity, we restrict our discussion in this section to the case where the
underlying sequential calculus is \lioafs{}, since this calculus extends
\liofs.

\begin{figure}
\small
\begin{hscode}\SaveRestoreHook
\column{B}{@{}>{\hspre}l<{\hspost}@{}}%
\column{6}{@{}>{\hspre}l<{\hspost}@{}}%
\column{17}{@{}>{\hspre}l<{\hspost}@{}}%
\column{E}{@{}>{\hspre}l<{\hspost}@{}}%
\>[B]{}\Varid{t}{}\<[6]%
\>[6]{}\Coloneqq\cdots{}\<[17]%
\>[17]{}\mid \textbf{forkLIO}\;\Varid{t}\mid \hcancel{\textbf{toLabeled}\ t\ t}{}\<[E]%
\ColumnHook
\end{hscode}\resethooks
\begin{mathpar}
\inferrule[forkLIO]
{ }
{
\ensuremath{\conf{\Sigma}{\textbf{\Blue{E}}\;[\mskip1.5mu \textbf{forkLIO}\;\Varid{t}\mskip1.5mu]}\overset{\Varid{fork}(\Varid{t})}{\lto}\conf{\Sigma}{\textbf{\Blue{E}}\;[\mskip1.5mu \mathbf{return}\;()\mskip1.5mu]}}
}
\and
\inferrule[withRefs-Opt]
{
\ensuremath{\Varid{v}\mathrel{=}\Varid{addrs}^{-1}((\Varid{addrs}(\Varid{v}_{1}))\;\cap\;(\Varid{addrs}(\Varid{v}_{2})))}\\
\ensuremath{\conf{\Sigma}{\textbf{\Blue{E}}\;[\mskip1.5mu \textbf{withRefs}_\textsc{fs}\;\Varid{v}\;\Varid{t}\mskip1.5mu]\lto\conf{\Sigma'}{\textbf{\Blue{E}}\;[\mskip1.5mu \Varid{t}'\mskip1.5mu]}}}
}
{
\ensuremath{\conf{\Sigma}{\textbf{\Blue{E}}\;[\mskip1.5mu \textbf{withRefs}_\textsc{fs}\;\Varid{v}_{1}\;(\textbf{withRefs}_\textsc{fs}\;\Varid{v}_{2}\;\Varid{t})\mskip1.5mu]}\lto\conf{\Sigma'}{\textbf{\Blue{E}}\;[\mskip1.5mu \Varid{t}'\mskip1.5mu]}}
}
\and
\inferrule[T-step]
{
\ensuremath{\Sigma\mathrel{=}(\lcurr,\mu_\textsc{fi},\mu_\textsc{fs})}\\
\ensuremath{\conf{\Sigma}{\textbf{withRefs}_\textsc{fs}\;\Varid{v}\;\Varid{t}}\lto\conf{\Sigma'}{\Varid{t}'}}\\
\ensuremath{\Sigma'\mathrel{=}(\lcurr',\mu_\textsc{fi}',\mu_\textsc{fs}')}\\
\ensuremath{\Varid{v}'\mathrel{=}\Varid{addrs}^{-1}(\textrm{dom}\;\mu_\textsc{fs}')}
}
{
\ensuremath{\tconf{\mu_\textsc{fi}}{\mu_\textsc{fs}}{\thread{\lcurr}{\Varid{v}}{\Varid{t}},\Varid{k}_{2},\mathbin{...}}\lto\tconf{\mu_\textsc{fi}'}{\mu_\textsc{fs}'}{\Varid{k}_{2},\mathbin{...},\thread{\lcurr'}{\Varid{v}'}{\Varid{t}'}}}
}
\and
\inferrule[T-stuck]
{ }
{
\ensuremath{\tconf{\mu_\textsc{fi}}{\mu_\textsc{fs}}{\thread{\lcurr}{\Varid{v}}{{\Uparrow}},\Varid{k}_{2},\mathbin{...}}\lto\tconf{\mu_\textsc{fi}}{\mu_\textsc{fs}}{\Varid{k}_{2},\mathbin{...}}}
}
\and
\inferrule[T-done]
{ }
{
\ensuremath{\tconf{\mu_\textsc{fi}}{\mu_\textsc{fs}}{\thread{\lcurr}{\Varid{v}}{\Varid{v}'},\Varid{k}_{2},\mathbin{...}}\lto\tconf{\mu_\textsc{fi}}{\mu_\textsc{fs}}{\Varid{k}_{2},\mathbin{...}}}
}

\and
\inferrule[T-fork]
{
\ensuremath{\Sigma\mathrel{=}(\lcurr,\mu_\textsc{fi},\mu_\textsc{fs})}\\
\ensuremath{\conf{\Sigma}{\textbf{withRefs}_\textsc{fs}\;\Varid{v}\;\Varid{t}}\overset{\Varid{fork}(\Varid{t}')}{\lto}\conf{\Sigma'}{\Varid{t}''}}\\
\ensuremath{\Sigma'\mathrel{=}(\lcurr',\mu_\textsc{fi}',\mu_\textsc{fs}')}\\
\ensuremath{\Varid{v}'\mathrel{=}\Varid{addrs}^{-1}(\textrm{dom}\;\mu_\textsc{fs}')}\\
\ensuremath{\Varid{k}_\textrm{new}\mathrel{=}\thread{\lcurr'}{\Varid{v}'}{\Varid{t}'}}
}
{
\ensuremath{\tconf{\mu_\textsc{fi}}{\mu_\textsc{fs}}{\thread{\lcurr}{\Varid{v}}{\Varid{t}},\Varid{k}_{2},\mathbin{...}}\lto\tconf{\mu_\textsc{fi}'}{\mu_\textsc{fs}'}{\Varid{k}_{2},\mathbin{...},\thread{\lcurr'}{\Varid{v}'}{\Varid{t}''},\Varid{k}_\textrm{new}}}
}
\end{mathpar}
  \caption{Semantics for \lioconc{}, parametric in the flow-sensitivity policy,
  i.e., with and without auto-upgrade\label{fig:sos:conc}.}
\end{figure}

Figure~\ref{fig:sos:conc} shows our extended concurrent calculus, \lioconc{}.
A concurrent program configuration has the form \ensuremath{\tconf{\mu_\textsc{fi}}{\mu_\textsc{fs}}{\Varid{k}_{1},\Varid{k}_{2},\mathbin{...}}},
where \ensuremath{\mu_\textsc{fi}} and \ensuremath{\mu_\textsc{fs}} are respectively the flow-insensitive and flow-sensitive
stores shared by all the threads \ensuremath{\Varid{k}_{1},\Varid{k}_{2},\mathbin{...}} in the program.
Since the memory stores are global, a thread \ensuremath{\Varid{k}} is simply a tuple
encapsulating the current label of the thread \ensuremath{\lcurr}, the term under
evaluation \ensuremath{\Varid{t}}, and a bag of references \ensuremath{\Varid{v}} the thread may access, i.e.,
\ensuremath{\Varid{k}\mathrel{=}\thread{\lcurr}{\Varid{v}}{\Varid{t}}}.

The reduction rules for concurrent programs are mostly standard.
Rule \ruleref{T-step} specifies that if the first thread in the thread pool
takes a step in \lioafs{}, the whole concurrent program takes a step, moving
the thread to the end of the pool.
We note that the thread term \ensuremath{\Varid{t}} executed with its stored current label
\ensuremath{\lcurr}, and a subset of the flow-sensitive memory store, by wrapping it in
\ensuremath{\textbf{withRefs}_\textsc{fs}}.
While the use of \ensuremath{\textbf{withRefs}_\textsc{fs}} makes the extension straightforward, one
peculiarity arises: since \ruleref{T-step} always wraps the thread term \ensuremath{\Varid{t}}
with \ensuremath{\textbf{withRefs}_\textsc{fs}}, if \ensuremath{\Varid{t}} does not reduce in one step to a value, and instead
reduces to a term \ensuremath{\Varid{t}'}, the next time the thread is scheduled, we will superfluously
wrap \ensuremath{\textbf{withRefs}_\textsc{fs}\;\Varid{t}'} with yet another \ensuremath{\textbf{withRefs}_\textsc{fs}}---thus preventing the thread
from making progress!
To address this problem, we extend the calculus with rule \ruleref{withRefs-Opt} that
collapses nested \ensuremath{\textbf{withRefs}_\textsc{fs}} blocks.\footnote{
This change also requires modifying \ruleref{withRefs-Ctx} to not be triggered when
the term being evaluated is a \ensuremath{\textbf{withRefs}_\textsc{fs}} term.
}

Rules \ruleref{T-done} and \ruleref{T-stuck} specify that once a thread term
has reduced to a value or got stuck, which is represented by \ensuremath{{\Uparrow}}, the
scheduler removes it from the thread pool and schedules the next thread.

As shown in Figure~\ref{fig:sos:conc}, to allow for dynamic thread creation, we
extend \lioafs{}'s terms with \ensuremath{\textbf{forkLIO}}, and add a new reduction rule that sends
an event to the scheduler, specifying the term to execute in a new
thread.\footnote{
In fact, the reduction rule for \lioafs{} must be changed to account for events.
However, since \ensuremath{\Varid{fork}} is the only event in our system, we treat \ensuremath{\lto} as
implicitly carrying an empty event.
}
Rule \ruleref{T-fork} describes the corresponding scheduler rule, triggered
when a \ensuremath{\Varid{fork}\;(\Varid{t}')} event is received.
Here, we create a new thread \ensuremath{\Varid{k}_\textrm{new}} whose current label \ensuremath{\lcurr'} and partition
of the store, i.e., bag of references \ensuremath{\Varid{v}'}, is the same as that of the parent
thread; the term evaluated in the newly created thread is provided in the
event: \ensuremath{\Varid{t}'}.
Subsequently, we add the new thread to the thread pool.

The final modification in extending \lioafs{} to \lioconc{}
%given in
%Figure~\ref{fig:sos:conc}
is the removal of \ensuremath{\textbf{toLabeled}} from the underlying
calculus.
As described in~\cite{stefan:addressing-covert}, we must remove \ensuremath{\textbf{toLabeled}} to
guarantee termination-sensitive non-interference.
Importantly, however, \ensuremath{\textbf{forkLIO}} with synchronization primitives (e.g.,
flow-insensitive labeled MVars, as discussed
in~\cite{stefan:addressing-covert}) can be used to provide functionality
equivalent to that of \ensuremath{\textbf{toLabeled}}.
Due to space constraints we omit synchronization primitives from
\lioconc{};
we only remark that extending \lioconc{} to provide flow-sensitive
labeled MVars follows in a straightforward way.

\begin{figure}[t]%{0.5\columnwidth}
\small
\begin{hscode}\SaveRestoreHook
\column{B}{@{}>{\hspre}l<{\hspost}@{}}%
\column{3}{@{}>{\hspre}l<{\hspost}@{}}%
\column{9}{@{}>{\hspre}l<{\hspost}@{}}%
\column{17}{@{}>{\hspre}l<{\hspost}@{}}%
\column{E}{@{}>{\hspre}l<{\hspost}@{}}%
\>[B]{}\Varid{leakRef}\mathbin{::}Ref^{\Red{\textsf{{\tiny TCB}}}}_{{}}\;\Conid{Bool}\to \Conid{LIO}\;\Conid{Bool}{}\<[E]%
\\
\>[B]{}\Varid{leakRef}\;\Varid{href}\mathrel{=}\mathbf{do}{}\<[E]%
\\
\>[B]{}\hsindent{3}{}\<[3]%
\>[3]{}\Varid{tmp}{}\<[9]%
\>[9]{}\leftarrow \textbf{newRef}_{{}}\;\textbf{L}\;(){}\<[E]%
\\
\>[B]{}\hsindent{3}{}\<[3]%
\>[3]{}\textbf{forkLIO}\mathbin{\$}\mathbf{do}\;{}\<[17]%
\>[17]{}\Varid{h}\leftarrow \textbf{readRef}_{{}}\;\Varid{href}{}\<[E]%
\\
\>[17]{}\textbf{when}\;\Varid{h}\mathbin{\$}\textbf{writeRef}_{{}}\;\Varid{tmp}\;(){}\<[E]%
\\
\>[B]{}\hsindent{3}{}\<[3]%
\>[3]{}\Varid{delay}{}\<[E]%
\\
\>[B]{}\hsindent{3}{}\<[3]%
\>[3]{}\mathbf{return}\mathbin{\$}\textbf{labelOf}_{{}}\;\Varid{tmp}\equiv \textbf{H}{}\<[E]%
\ColumnHook
\end{hscode}\resethooks
\cut{$}
\caption{Attack on concurrent LIO with naive flow-sensitive reference
extension.
\label{fig:fs-conc-attack}}
\end{figure}
Since the flow-sensitive attack in Figure~\ref{fig:fs-attack} relied
on \ensuremath{\textbf{toLabeled}} to restore the current label, a natural question, given
that we remove \ensuremath{\textbf{toLabeled}}, is whether we can use the naive
flow-sensitive reference semantics of Section~\ref{sec:flow-sensitive}
for concurrent LIO.
As shown by the attack code in Figure~\ref{fig:fs-conc-attack}, in
which we use \ensuremath{\textbf{forkLIO}} instead of \ensuremath{\textbf{toLabeled}} to address a potential
label creep, the fundamental problem remains: the label on the
reference label is not protected!
This precisely motivated our principled approach of extending \lioafs{} to a
concurrent setting as opposed to extending concurrent LIO with flow-sensitive
references.

% Local Variables:
% TeX-master: "main.tex"
% TeX-command-default: "Make"
% End:
% 1 Page (Pablo/Ale), attack and discussion

  % 1 Page (Pablo/Ale)

\section{Formal results}
\label{sec:soundness}

In this section, we show that our flow-sensitive enforcement can be
embedded into the flow-insensitive version of LIO. Additionally, we
provide security guarantees in terms of non-interference definitions by reusing
previous results on LIO.

\subsection{Embedding into \lio{}}

Every flow-sensitive reference
with label \ensuremath{\Varid{l}_{\Varid{d}}} created in a context where the current label is \ensuremath{\Varid{l}_{\Varid{o}}} (and
thus stored in \ensuremath{\mu_\textsc{fs}} as \ensuremath{Lb^{\Red{\textsf{{\tiny TCB}}}}\;\Varid{l}_{\Varid{o}}\;(Lb^{\Red{\textsf{{\tiny TCB}}}}\;\Varid{l}_{\Varid{d}}\;\Varid{t})}), can be represented
by a flow-insensitive reference with label \ensuremath{\Varid{l}_{\Varid{o}}}, whose contents are
another flow-insensitive reference containing \ensuremath{\Varid{t}} and labeled \ensuremath{\Varid{l}_{\Varid{d}}}.

Figure~\ref{fig:fs-exts-semantics-impl} gives our encoding of
the flow-sensitive reference operations in terms of flow-insensitive
references. For a given
store \ensuremath{\Sigma}, we define the \ensuremath{\llbracket \mathbin{-}\rrbracket_{\textsc{fi}}^{\Sigma}} function, which given a term
\ensuremath{\Varid{t}} in \liofs{}, produces a term \ensuremath{\llbracket \Varid{t}\rrbracket_{\textsc{fi}}^{\Sigma}} in \lio{}, 
expanding the definitions of flow-sensitive operations in
terms of flow-insensitive ones. This function is applied
homomorphically in all other cases. We use the \ensuremath{\Conid{WrapRef}} constructor
to mark the flow-insensitive references that are being used to
represent flow-sensitive ones, so as to distinguish them from other
flow-insensitive references. The functions \ensuremath{\Varid{wrap}} and \ensuremath{\Varid{unwrap}} are
used to add and remove this boundary encoding.
In the embedding of \ensuremath{\textbf{writeRef}_{\textsc{fs}}}, we use \ensuremath{\textbf{toLabeled}} to make any
changes to the current label (possibly caused by reading the outer
reference) local to this operation. Inside \ensuremath{\textbf{toLabeled}}, the code
fetches the inner reference (\ensuremath{\textbf{readRef}_{\textsc{fi}}}), and then performs the
actual write of the new value.  If this fails, the computation
diverges, but, importantly, the current label was raised (with
\ensuremath{\textbf{readRef}_{\textsc{fi}}}) to reflect the fact that the label on the label of the
reference was observed.
The embedding of \ensuremath{\textbf{upgrade}_\textsc{fs}} relies on flow-insensitive primitives to implement
the \ensuremath{\textbf{upgrade}} operation. As in \ensuremath{\textbf{writeRef}_{\textsc{fs}}}, a \ensuremath{\textbf{toLabeled}} block is used to
delimit the taint on the current label. Inside the block, the code fetches the
inner reference (\ensuremath{\textbf{readRef}_{\textsc{fi}}}), which taints the current label with \ensuremath{\Varid{l'}}, and
makes a new reference \ensuremath{\Varid{n}} (\ensuremath{\textbf{newRef}_{\textsc{fi}}}) with the upgraded label (\ensuremath{\lcurr\;\lub\;(\Varid{l}\;\lub\;\textbf{labelOf}\;\Varid{i})}) and an undefined value (\ensuremath{\bot }). Observe that the
operation for creating the reference always succeeds since its label is above
the current label, i.e., \ensuremath{\lcurr}. Then, \ensuremath{\textbf{copyRef}} is used to copy the value of the
original inner reference into the new one, \ensuremath{\Varid{n}}. As before, this action always
succeeds because the label of the reference bound to \ensuremath{\Varid{i}} (\ensuremath{\textbf{labelOf}\;\Varid{i}}) is below
the label of the new reference \ensuremath{\Varid{n}}. Finally, the reference \ensuremath{\Varid{n}} is stored in
place of the original inner reference using \ensuremath{\textbf{writeRef}_{\cdot }}. Importantly, this
instruction only succeeds when the current label at the time of writing, i.e.,
\ensuremath{\Varid{lc}\;\lub\;\Varid{l'}} in Figure~\ref{fig:fs-exts-semantics-impl}, is below or
equal to \ensuremath{\Varid{l'}} (the label of the outer reference), i.e., \ensuremath{\Varid{lc}\;\lub\;\Varid{l'}\;\flows\;\Varid{l'}}. This restrictions holds when the current label at the
time of upgrade, i.e., \ensuremath{\Varid{lc}}, is below or equal to \ensuremath{\Varid{l'}}---effectively
encoding the non-sensitive upgrade policy for label changes. The
embedding of \ensuremath{\textbf{downgrade}_\textsc{fs}} follows similarly, except that the label of
the new reference is computed using \ensuremath{\glb} instead of \ensuremath{\lub} (to achieve
the downgrade), and the \ensuremath{\textbf{copyRef}} step is omitted, since the original
value must be destroyed.
We remark that
%For the sake of brevity, we do not explain the mapping in further
that the mapping mimics the behavior described by the rules in
Figure~\ref{fig:sos:fs}.
%\hl{(Pablo) Maybe `brevity' is not a good reason for
%  omitting something in a journal paper?}
%

\begin{figure}
\small
\begin{hscode}\SaveRestoreHook
\column{B}{@{}>{\hspre}l<{\hspost}@{}}%
\column{3}{@{}>{\hspre}l<{\hspost}@{}}%
\column{4}{@{}>{\hspre}l<{\hspost}@{}}%
\column{5}{@{}>{\hspre}l<{\hspost}@{}}%
\column{7}{@{}>{\hspre}l<{\hspost}@{}}%
\column{10}{@{}>{\hspre}l<{\hspost}@{}}%
\column{11}{@{}>{\hspre}l<{\hspost}@{}}%
\column{34}{@{}>{\hspre}l<{\hspost}@{}}%
\column{36}{@{}>{\hspre}l<{\hspost}@{}}%
\column{38}{@{}>{\hspre}l<{\hspost}@{}}%
\column{40}{@{}>{\hspre}l<{\hspost}@{}}%
\column{E}{@{}>{\hspre}l<{\hspost}@{}}%
\>[B]{}\Varid{wrap}\;\Varid{r}\triangleq{}\Conid{WrapRef}\;\Varid{r}{}\<[E]%
\\
\>[B]{}\Varid{unwrap}\;(\Conid{WrapRef}\;\Varid{r})\triangleq{}\Varid{r}{}\<[E]%
\\[\blanklineskip]%
\>[B]{}\llbracket Ref^{\Red{\textsf{{\tiny TCB}}}}_{\textsc{fs}}\;\Varid{r}\rrbracket_{\textsc{fi}}^{(\lcurr,\mu_\textsc{fi},\mu_\textsc{fs})}\triangleq{}\Varid{wrap}\;(Ref^{\Red{\textsf{{\tiny TCB}}}}_{\textsc{fi}}\;(\textbf{labelOf}_{\textsc{fi}}\;\mu_\textsc{fs}\;(\Varid{r}))\;\Varid{r}){}\<[E]%
\\[\blanklineskip]%
\>[B]{}\llbracket \textbf{newRef}_{\textsc{fs}}\rrbracket_{\textsc{fi}}^{\Sigma}\triangleq{}\lambda \Varid{l}\;\Varid{t}.\mathbf{do}{}\<[E]%
\\
\>[B]{}\hsindent{3}{}\<[3]%
\>[3]{}\Varid{i}{}\<[10]%
\>[10]{}\leftarrow \textbf{newRef}_{\textsc{fi}}\;\Varid{l}\;\Varid{t}{}\<[E]%
\\
\>[B]{}\hsindent{3}{}\<[3]%
\>[3]{}\lcurr{}\<[10]%
\>[10]{}\leftarrow \mathbf{getLabel}{}\<[E]%
\\
\>[B]{}\hsindent{3}{}\<[3]%
\>[3]{}\Varid{o}{}\<[10]%
\>[10]{}\leftarrow \textbf{newRef}_{\textsc{fi}}\;\lcurr\;\Varid{i}{}\<[E]%
\\
\>[B]{}\hsindent{3}{}\<[3]%
\>[3]{}\mathbf{return}\;(\Varid{wrap}\;\Varid{o}){}\<[E]%
\\[\blanklineskip]%
\>[B]{}\llbracket \textbf{readRef}_{\textsc{fs}}\rrbracket_{\textsc{fi}}^{\Sigma}\triangleq{}\lambda \Varid{r}.\textbf{readRef}_{\textsc{fi}}\;(\Varid{unwrap}\;\Varid{r})\bind \textbf{readRef}_{\textsc{fi}}{}\<[E]%
\\[\blanklineskip]%
\>[B]{}\llbracket \textbf{writeRef}_{\textsc{fs}}\rrbracket_{\textsc{fi}}^{\Sigma}\triangleq{}\lambda \Varid{r}\;\Varid{t}.\mathbf{let}\;\Varid{o}\mathrel{=}\Varid{unwrap}\;\Varid{r}\;\mathbf{in}\;\mathbf{do}{}\<[E]%
\\
\>[B]{}\hsindent{3}{}\<[3]%
\>[3]{}\lcurr\leftarrow \mathbf{getLabel}{}\<[E]%
\\
\>[B]{}\hsindent{3}{}\<[3]%
\>[3]{}\textbf{toLabeled}\;(\lcurr\;\lub\;(\textbf{labelOf}\;\Varid{o}))\mathbin{\$}\mathbf{do}{}\<[E]%
\\
\>[3]{}\hsindent{2}{}\<[5]%
\>[5]{}\Varid{i}\leftarrow \textbf{readRef}_{\textsc{fi}}\;\Varid{o}{}\<[E]%
\\
\>[3]{}\hsindent{2}{}\<[5]%
\>[5]{}\textbf{writeRef}_{\textsc{fi}}\;\Varid{i}\;\Varid{t}{}\<[E]%
\\[\blanklineskip]%
\>[B]{}\llbracket \textbf{labelOf}_{\textsc{fs}}\rrbracket_{\textsc{fi}}^{\Sigma}\triangleq{}\lambda \Varid{r}.{}\<[E]%
\\
\>[B]{}\hsindent{3}{}\<[3]%
\>[3]{}\textbf{readRef}_{\textsc{fi}}\;(\Varid{unwrap}\;\Varid{r})\bind \mathbf{return}.\textbf{labelOf}_{\textsc{fi}}{}\<[E]%
\\[\blanklineskip]%
\>[B]{}\llbracket \textbf{upgrade}_\textsc{fs}\rrbracket_{\textsc{fi}}^{\Sigma}\triangleq{}\lambda \Varid{r}\;\Varid{l}.\mathbf{let}\;{}\<[34]%
\>[34]{}\Varid{o}{}\<[38]%
\>[38]{}\mathrel{=}\Varid{unwrap}\;\Varid{r}{}\<[E]%
\\
\>[34]{}\Varid{l'}{}\<[38]%
\>[38]{}\mathrel{=}\textbf{labelOf}\;\Varid{o}\;\mathbf{in}\;\mathbf{do}{}\<[E]%
\\
\>[B]{}\hsindent{5}{}\<[5]%
\>[5]{}\Varid{lc}\leftarrow \mathbf{getLabel}{}\<[E]%
\\
\>[B]{}\hsindent{5}{}\<[5]%
\>[5]{}\textbf{toLabeled}\;(\Varid{lc}\;\lub\;\Varid{l'})\mathbin{\$}\mathbf{do}{}\<[E]%
\\
\>[5]{}\hsindent{2}{}\<[7]%
\>[7]{}\Varid{i}{}\<[11]%
\>[11]{}\leftarrow \textbf{readRef}_{\textsc{fi}}\;\Varid{o}{}\<[E]%
\\
\>[5]{}\hsindent{2}{}\<[7]%
\>[7]{}\lcurr\leftarrow \mathbf{getLabel}{}\<[E]%
\\
\>[5]{}\hsindent{2}{}\<[7]%
\>[7]{}\Varid{n}{}\<[10]%
\>[10]{}\leftarrow \textbf{newRef}_{\textsc{fi}}\;(\lcurr\;\lub\;(\Varid{l}\;\lub\;\textbf{labelOf}\;\Varid{i}))\;\bot {}\<[E]%
\\
\>[5]{}\hsindent{2}{}\<[7]%
\>[7]{}\textbf{copyRef}\;\Varid{i}\;\Varid{n}{}\<[E]%
\\
\>[5]{}\hsindent{2}{}\<[7]%
\>[7]{}\textbf{writeRef}_{\textsc{fi}}\;\Varid{o}\;\Varid{n}{}\<[E]%
\\[\blanklineskip]%
\>[B]{}\llbracket \textbf{downgrade}_\textsc{fs}\rrbracket_{\textsc{fi}}^{\Sigma}\triangleq{}\lambda \Varid{r}\;\Varid{l}.\mathbf{let}\;{}\<[36]%
\>[36]{}\Varid{o}{}\<[40]%
\>[40]{}\mathrel{=}\Varid{unwrap}\;\Varid{r}{}\<[E]%
\\
\>[36]{}\Varid{l'}{}\<[40]%
\>[40]{}\mathrel{=}\textbf{labelOf}\;\Varid{o}\;\mathbf{in}\;\mathbf{do}{}\<[E]%
\\
\>[B]{}\hsindent{5}{}\<[5]%
\>[5]{}\Varid{lc}\leftarrow \mathbf{getLabel}{}\<[E]%
\\
\>[B]{}\hsindent{5}{}\<[5]%
\>[5]{}\textbf{toLabeled}\;(\Varid{lc}\;\lub\;\Varid{l'})\mathbin{\$}\mathbf{do}{}\<[E]%
\\
\>[5]{}\hsindent{2}{}\<[7]%
\>[7]{}\Varid{i}{}\<[11]%
\>[11]{}\leftarrow \textbf{readRef}_{\textsc{fi}}\;\Varid{o}{}\<[E]%
\\
\>[5]{}\hsindent{2}{}\<[7]%
\>[7]{}\lcurr\leftarrow \mathbf{getLabel}{}\<[E]%
\\
\>[5]{}\hsindent{2}{}\<[7]%
\>[7]{}\Varid{n}{}\<[10]%
\>[10]{}\leftarrow \textbf{newRef}_{\textsc{fi}}\;(\lcurr\;\lub\;(\Varid{l}\;\glb\;\textbf{labelOf}\;\Varid{i}))\;\bot {}\<[E]%
\\
\>[5]{}\hsindent{2}{}\<[7]%
\>[7]{}\textbf{writeRef}_{\textsc{fi}}\;\Varid{o}\;\Varid{n}{}\<[E]%
\\[\blanklineskip]%
\>[B]{}\llbracket \textbf{withRefs}_\textsc{fs}\;\Varid{v}\;\Varid{t}\rrbracket_{\textsc{fi}}^{(\lcurr,\mu_\textsc{fi},\mu_\textsc{fs})}\triangleq{}\llbracket \Varid{t}\rrbracket_{\textsc{fi}}^{(\lcurr,\mu_\textsc{fi},\mu_\textsc{fs}')}{}\<[E]%
\\
\>[B]{}\hsindent{4}{}\<[4]%
\>[4]{}\mathbf{where}{}\<[E]%
\\
\>[4]{}\hsindent{1}{}\<[5]%
\>[5]{}\mu_\textsc{fs}'{}\<[10]%
\>[10]{}\mathrel{=}\{\mskip1.5mu \Varid{a}\;\mapsto\;\mu_\textsc{fs}\;(\Varid{a})\;|\;\Varid{a}\;\in\;\textrm{dom}\;\mu_\textsc{fs}\;\cap\;(\Varid{addrs}^+_{\mu_\textsc{fs}}(\Varid{v}))\mskip1.5mu\}{}\<[E]%
\ColumnHook
\end{hscode}\resethooks
\caption{Implementation mapping for flow-sensitive references. For all other terms, the function is applied homomorphically.\label{fig:fs-exts-semantics-impl}}
\end{figure}

We extend this definition naturally to convert \liofs{} environments
into \lio{} environments, by having \ensuremath{\llbracket (\lcurr,\mu_\textsc{fi},\mu_\textsc{fs})\rrbracket_{\textsc{fi}}\triangleq{}(\lcurr,\mu_\textsc{fi}')} where \ensuremath{\mu_\textsc{fi}'} is obtained by extending \ensuremath{\mu_\textsc{fi}} with the pair of bindings \ensuremath{\Varid{a}_\Varid{i}\;\mapsto\;Lb^{\Red{\textsf{{\tiny TCB}}}}\;\Varid{l}_{\Varid{o}}\;(Ref^{\Red{\textsf{{\tiny TCB}}}}_{\textsc{fi}}\;\Varid{l}_{\Varid{d}}\;\Varid{b}_{\Varid{i}}),\Varid{b}_{\Varid{i}}\;\mapsto\;(Lb^{\Red{\textsf{{\tiny TCB}}}}\;\Varid{l}_{\Varid{d}}\;\Varid{v})} (with \ensuremath{\Varid{b}_{\Varid{i}}} being a fresh name) for each binding of the form \ensuremath{\Varid{a}_\Varid{i}\;\mapsto\;Lb^{\Red{\textsf{{\tiny TCB}}}}\;\Varid{l}_{\Varid{o}}\;(Lb^{\Red{\textsf{{\tiny TCB}}}}\;\Varid{l}_{\Varid{d}}\;\Varid{v}_i)} in
\ensuremath{\mu_\textsc{fs}}.
%(\hl{TODO:}we also have to expand |a_i mapsto Lb l_o (Lb l_d v)| into two references |a_i mapsto Lb l_o (Ref l_d b_i)| and |b_i mapsto Lb l_d v|).
Note that the domains of \ensuremath{\mu_\textsc{fi}} and \ensuremath{\mu_\textsc{fs}} are disjoint because the
\ensuremath{\fresh{\cdot }} predicate that we use in the semantics is assumed to produce
globally unique addresses.

In order to prove that our implementation is correct with respect to
the semantics, we show that, if we take a program with flow-sensitive
operations, and expand those operations, replacing them by the code in
Figure~\ref{fig:fs-exts-semantics-impl}, then its behavior corresponds
with the flow-sensitive semantics.

% For the proof, we will require a lemma stating that, for every
% flow-sensitive reference, the label-on-the-label always flows to the
% label of the reference at every program point.

% \begin{lemma} Let |t| be a term in \liofs{}, |c = (lcurr, mI, mS)| an
%   initial environment and |r| a flow-sensitive reference in |c|.
%   Then, for every configuration |conf ((lcurr',mI',mS')) t'| reachable
%   from |conf c t|, if we consider |mS'(r) = LabeledTCB l_o (LabeledTCB
%   l_d v)|, we have that |l_o canFlowTo l_d|.
% \end{lemma}

% We can now state our equivalence theorem.

\begin{theorem}[Embedding \liofs{} in \lio{}]
\label{thm:eq} Let \ensuremath{\Varid{t}} be a well-typed term in \liofs{}.
  Then if \ensuremath{\conf{\Sigma}{\Varid{t}}\lto^*\conf{\Sigma'}{\Varid{v}}}, we have \ensuremath{\conf{\llbracket \Sigma\rrbracket_{\textsc{fi}}}{\llbracket \Varid{t}\rrbracket_{\textsc{fi}}^{\Sigma}}\lto^*\conf{\llbracket \Sigma'\rrbracket_{\textsc{fi}}}{\llbracket \Varid{v}\rrbracket_{\textsc{fi}}^{\Sigma}}}, and if
  \ensuremath{\conf{\Sigma}{\Varid{t}}\lto^*\conf{\Sigma'}{{\Uparrow}}}, then
  \ensuremath{\conf{\llbracket \Sigma\rrbracket_{\textsc{fi}}}{\llbracket \Varid{t}\rrbracket_{\textsc{fi}}^{\Sigma}}\lto^*\conf{\llbracket \Sigma'\rrbracket_{\textsc{fi}}}{{\Uparrow}}}.
  % \textrm{\hl{this is not right:
  %     the constructors for FS references are different (one wraps FI
  %     references, the other doesn't), so syntactic equiality is not
  %     right.  we need the equivalence relation I defined before}}
\end{theorem}
\begin{proof}
See Appendix~\ref{app:embedding}.
\end{proof}
While straight forward, this theorem highlights an important result: in
floating label systems, flow-sensitive references can be encoded in a calculus
with flow-insensitive references and explicitly labeled values.

\subsection{Security guarantees for \liofs{}, \lioafs{} and \lioconc{}}

From previous results~\cite{stefan:lio}, we know that LIO satisfies termination-insensitive
non-interference (TINI) in the sequential setting, and termination-sensitive
non-interference (TSNI) in the concurrent setting. By using the
embedding theorem we can extend these results for LIO with flow-sensitive references.

For completeness, we now present our
non-interference theorems, as straightforward applications of the theorems in
previous work. Our security results rely on the notion of
\ensuremath{\Varid{l}}-equivalence for terms and configurations, which captures the idea
of terms that cannot be distinguished by an attacker which can observe
data at level \ensuremath{\Varid{l}}. A pair of terms \ensuremath{\Varid{t}_{1},\Varid{t}_{2}} is said to be
\ensuremath{\Varid{l}}-equivalent (written \ensuremath{\Varid{t}_{1}\approx_{\Varid{l}}\Varid{t}_{2}}) if, after erasing all the
information more sensitive than \ensuremath{\Varid{l}} from \ensuremath{\Varid{t}_{1}} and \ensuremath{\Varid{t}_{2}}, we obtain
syntactically equivalent terms. This definition extends naturally to
configurations.

Intuitively, non-interference means that an attacker at level \ensuremath{\Varid{l}}
cannot distinguish among different runs of a program with
\ensuremath{\Varid{l}}-equivalent initial configurations.

\begin{theorem}[TINI for \liofs{}]
  Consider two well-typed terms \ensuremath{\Varid{t}_{1}} and \ensuremath{\Varid{t}_{2}} in
  \liofs{} which do not contain any \ensuremath{\cdot\;{}^{\Red{\textsf{{\tiny{TCB}}}}}} syntax nodes, such that \ensuremath{\Varid{t}_{1}\approx_{\Varid{l}}\Varid{t}_{2}}, where \ensuremath{\Varid{l}} is the attacker observation level.
  Let \ensuremath{\Sigma} be an initial environment, and let
\[
    \ensuremath{\conf{\Sigma}{\Varid{t}_{1}}\lto^*\conf{\Sigma_{1}}{\Varid{v}_{1}}}\mbox{ and }
    \ensuremath{\conf{\Sigma}{\Varid{t}_{2}}\lto^*\conf{\Sigma_{2}}{\Varid{v}_{2}}}
\]
  Then, we have that \ensuremath{\conf{\Sigma_{1}}{\Varid{v}_{1}}\approx_{\Varid{l}}\conf{\Sigma_{2}}{\Varid{v}_{2}}}.
\end{theorem}

\begin{proof}
  By expanding all the flow-sensitive operations in \ensuremath{\Varid{t}_{1}} and \ensuremath{\Varid{t}_{2}} using their
  definition given in Figure~\ref{fig:fs-exts-semantics-impl}, we get
  terms in \lio{}, which by Theorem~\ref{thm:eq} has equivalent
  semantics. Therefore, the result follows from the
  \lio{} TINI result of~\citep{stefan:lio}.
\end{proof}

\begin{corollary}[TINI for \lioafs{}]
  The previous non-interference result can be easily extended to \lioafs{}.
  In \lioafs{}, the \ensuremath{\textbf{unlabel}} operation
  triggers the automatic upgrades mechanism, which
  performs the \ensuremath{\textbf{upgrade}} operation for every flow-sensitive
  reference in scope before actually raising the current label.
  Regardless of how \ensuremath{\textbf{unlabel}} is used, we note that the resulting term
  (after inserting the necessary \ensuremath{\textbf{upgrade}}s), is just an \liofs{}
  term. Therefore, the main TINI result for \liofs{} applies.

 % Additionally, the |withRefs| primitive allows the programmer
 %  to restrict the set of references in scope, reducing the number of
 %  |upgrade|s that need to be performed when |unlabel| occurs.

% \hl{move
%    withRefs}
\end{corollary}

For the concurrent result, we need a supporting lemma which states
that the current label is always at least as sensitive as
the label of every reference in scope. Formally,

\begin{lemma} \label{lem:lolFlowslCurr}
  Let \ensuremath{\Varid{t}} be a well-typed term in \lioconc{}, \ensuremath{\Sigma\mathrel{=}(\lcurr,\mu_\textsc{fi},\mu_\textsc{fs})} be an initial
  environment, and \ensuremath{\Varid{a}} be the address of a flow-sensitive reference \ensuremath{\Varid{r}} in \ensuremath{\Sigma},
  where \ensuremath{\mu_\textsc{fs}\;(\Varid{a})\mathrel{=}Lb^{\Red{\textsf{{\tiny TCB}}}}\;\Varid{l}_{\Varid{o}}\;(Lb^{\Red{\textsf{{\tiny TCB}}}}\;\Varid{l}_{\Varid{d}}\;\Varid{v})}. Then, if \ensuremath{\conf{\Sigma}{\Varid{t}}\lto^*\conf{\lcurr',\mu_\textsc{fi}',\mu_\textsc{fs}'}{\Varid{t}'}}, we have that \ensuremath{\Varid{l}_{\Varid{o}}\;\flows\;\lcurr'}.
\end{lemma}

\begin{proof}
  Note that the result holds immediately after creating \ensuremath{\Varid{r}}, since
  the current label is the label on the label of \ensuremath{\Varid{r}}, i.e., \ensuremath{\Varid{l}_{\Varid{o}}\mathrel{=}\lcurr}. It is easy to show that \ensuremath{\Varid{l}_{\Varid{o}}} is immutable, since there are
  no reduction rules that modify it. Moreover, given that the current
  label is monotonic, the only way in which \ensuremath{\Varid{l}_{\Varid{o}}\;\flows\;\lcurr} can
  cease to hold is if \ensuremath{\Varid{r}} is accessed from a different thread. But in
  order to pass \ensuremath{\Varid{r}} to a different thread, a labeled object must be
  used as intermediary, and the label of such object would have to be
  at least \ensuremath{\lcurr}, the current label in the thread that created
  \ensuremath{\Varid{r}}. As a result, if we were to pass \ensuremath{\Varid{r}} to another thread in this
  way, then the target thread would also have to be tainted by
  \ensuremath{\lcurr}, and the result would still hold.
\end{proof}

We now prove our non-interference theorem for \lioconc. This result is
stronger than TINI, since it implies that there can be no termination
or internal timing leaks.

\begin{theorem}[TSNI for \lioconc{}]
  Consider two well-typed terms \ensuremath{\Varid{t}_{1}} and \ensuremath{\Varid{t}_{2}} in
  \lioconc{} which do not contain any \ensuremath{\cdot\;{}^{\Red{\textsf{{\tiny{TCB}}}}}} syntax nodes, such that \ensuremath{\Varid{t}_{1}\approx_{\Varid{l}}\Varid{t}_{2}}, where \ensuremath{\Varid{l}} is the attacker observation level.
  Let \ensuremath{\Sigma\mathrel{=}(\lcurr,\mu_\textsc{fi},\mu_\textsc{fs})} be an initial environment, and let
\[
 \ensuremath{\tconf{\mu_\textsc{fi}}{\mu_\textsc{fs}}{\thread{\lcurr}{\Varid{addrs}^{-1}(\textrm{dom}\;\mu_\textsc{fs})}{\Varid{t}_{1}}}\lto^*\Conid{M}_{1}}
\]
  Then, there exists some configuration \ensuremath{\Conid{M}_{2}} such that
\[
  \ensuremath{\tconf{\mu_\textsc{fi}}{\mu_\textsc{fs}}{\thread{\lcurr}{\Varid{addrs}^{-1}(\textrm{dom}\;\mu_\textsc{fs})}{\Varid{t}_{2}}}\lto^*\Conid{M}_{2}}
\]
  and \ensuremath{\Conid{M}_{1}\approx_{\Varid{l}}\Conid{M}_{2}}.
\end{theorem}

\begin{proof} From Lemma~\ref{lem:lolFlowslCurr} and
  looking at the embeddings of \ensuremath{\textbf{writeRef}_{\textsc{fs}}} and \ensuremath{\textbf{upgrade}_\textsc{fs}}, we note
  that the first \ensuremath{\textbf{readRef}_{\textsc{fi}}} operation in each \ensuremath{\textbf{toLabeled}} block will
  be trying to raise the current label to \ensuremath{\Varid{l}}. However, since \ensuremath{\Varid{l}\;\flows\;\lcurr}, these operations will never effectively raise the
  current label. This means that using \ensuremath{\textbf{toLabeled}} is not necessary to
  preserve the semantics, because there is no need to restore the
  current label afterwards. As a result, and after removing
  \ensuremath{\textbf{toLabeled}} in these two cases, we note that the embedding produces
  valid concurrent \lio{} terms (which does not have \ensuremath{\textbf{toLabeled}}).

%%   However, we note that all the critical sections in our operations consist in
%%   a |readRef I| followed by a |writeRef I|. Therefore, we argue that
%%   the problem could be addressed by replacing these operations with
%%   |atomicModifyRef I|, an atomic operation that performs the read and write and
%%   is part of the actual implementation of LIO, but which we elide in order to
%%   simplify our presentation.

  Finally, by expanding all the flow-sensitive operations in \ensuremath{\Varid{t}_{1}} and \ensuremath{\Varid{t}_{2}}
  using their definition given in
  Figure~\ref{fig:fs-exts-semantics-impl}, we get terms in concurrent
  \lio{}.  Therefore, the result follows from the
  termination-sensitive non-interference of concurrent
  \lio{}~\citep{stefan:addressing-covert}.
\end{proof}

The detailed proofs for the results in this section can be found in
Appendix~\ref{app:embedding}.
We lastly remark a limitation: while we preserve non-interference when
embedding the flow-sensitive calculus in the original LIO, our
embedding includes no synchronization to ensure atomicity of the
flow-sensitive operations, so certain interleaving that break semantic
equivalence are possible. 

\subsection{Permissiveness}
In Section~\ref{sec:related} we compare the permissiveness of our
system with previous flow-sensitive IFC systems. Here, we
solely remark that the above results imply that our
flow-sensitive calculus is as permissive as flow-insensitive
LIO. In particular, any flow-insensitive LIO program can be trivially
converted to a flow-sensitive LIO program (without auto-upgrades) by
using flow-sensitive references instead of flow-insensitive ones.
Since these references would never be upgraded, they will behave just
like their flow-insensitive counterparts.
This means that all existing LIO programs can be run in our
flow-sensitive monitor. This includes Hails~\cite{giffin:hails}, a
web framework using LIO, on top of which a number of applications
have been built (e.g., GitStar\footnote{\url{www.gitstar.com}}, a
code-hosting web platform,
LearnByHacking\footnote{\url{www.learnbyhacking.org}}, a
blog/tutorial platform similar to School of Haskell, and
LambdaChair~\cite{stefan:2012:arxiv-flexible}, an EasyChair-like
conference review system).

% GitStar's core could be run in our
%   flow-sensitive monitor; the full system uses other features like
%   privileges and exceptions that we have not directly implemented in
%   this version of the flow-sensitive monitor, but we see no major
%   challenges in incorporating flow-sensitivity to the full version of
%   LIO.

% \hl{NOTE}

% Why is |toLabeled| secure in the implementation of |writeRef S|? It
% would seem that it is impossible for the first |readRef I| to
% effectively raise the current label, i.e. if a reference with LOL |l|
% is in scope, it is always the case that |l canFlowTo lcurr|. Why? If a
% reference has LOL |l|, it must have been created in a context where
% the current label was |l|. From here, it is not possible to access the
% reference from a context with a lower current label. This would
% suggest that |toLabeled| is not necessary, since |readRef| is not
% changing the current label anyway.

% Local Variables:
% TeX-master: "main.lhs.tex"
% TeX-command-default: "Make"
% End:
% 2 Pages (Pablo/Ale)

\section{Comparison with other policies for label change}

In this section, we compare our enforcement mechanism with two
policies for label change: \emph{\nsu} (originally proposed by
Zdancewic~\cite{Zdancewic02programminglanguages}) and \emph{\pu}, a
more permissive version of the former by Austin and
Flanagan~\citep{Austin:Flanagan:PLAS09,Austin:Flanagan:PLAS10}.
We introduce a simple imperative language to simplify our comparison
with the languages implementing the aforementioned policies.  This simple
language has variables, \ensuremath{\mathbf{if}}-statements, a \ensuremath{\textbf{skip}} command that does nothing, and
an \ensuremath{\textbf{output}} command that is used to produce public outputs.  This language is
easily implemented in \liofs{} as syntactic sugar.  For example, a conditional
statement \ensuremath{\mathbf{if}\;\Conid{C};\Conid{A}\;\mathbf{else}\;\Conid{B}} is desugared to \ensuremath{\textbf{toLabeled}\;\textbf{H}\;(\mathbf{do}\;\Varid{b}\leftarrow \textbf{unlabel}\;\Conid{C};\mathbf{if}\;\Varid{b}\;\mathbf{then}\;\Conid{A}\;\mathbf{else}\;\Conid{B})}.  (Here, \ensuremath{\textbf{toLabeled}} is used to ensure that the current
label is restored after leaving the \ensuremath{\mathbf{if}}-statement.)

\subsection{\Nsu}
The {\nsu} discipline stops execution on any attempt to change the
label of a public variable inside a secret context. Our \liofs{}
calculus essentially implements this discipline as
well---see~\ruleref{upgradeRef} in Figure~\ref{fig:sos:fs}.  The
original presentation of {\nsu} allows for variables with a secret
label to be downgraded, as long as the original contents are
discarded. Our \liofs{} calculus similarly allows for this
with the \ensuremath{\textbf{downgrade}} operation.  Our approach differs
in also allowing code to explicitly upgrade a variable
before entering a secret context, permissively allowing writes to
originally-public variables in secret contexts.

\subsection{\Pu}

The {\pu} policy differs from {\nsu} in allowing code to
change the label of a public variables in secret contexts, but
subsequently disallowing branches on
such permissively-upgraded, or ``marked,'' variables.
When upgrading a public variable in a secret context, the
security label of the variable
is changed to \ensuremath{\textbf{P}} where \ensuremath{\textbf{L}\;\flows\;\textbf{H}\;\flows\;\textbf{P}}.

\begin{wrapfigure}{l}{4cm}
\small
\begin{hscode}\SaveRestoreHook
\column{B}{@{}>{\hspre}l<{\hspost}@{}}%
\column{4}{@{}>{\hspre}l<{\hspost}@{}}%
\column{E}{@{}>{\hspre}l<{\hspost}@{}}%
\>[B]{}\textbf{upgrade}\;\Varid{x}\;\textbf{H}{}\<[E]%
\\
\>[B]{}\mathbf{if}\;\Varid{h}{}\<[E]%
\\
\>[B]{}\hsindent{4}{}\<[4]%
\>[4]{}\Varid{x}\mathbin{:=}\Conid{True}{}\<[E]%
\\
\>[B]{}\mathbf{if}\;\Varid{x}{}\<[E]%
\\
\>[B]{}\hsindent{4}{}\<[4]%
\>[4]{}\textbf{skip}{}\<[E]%
\ColumnHook
\end{hscode}\resethooks
\caption{A secure program that {\liofs} accepts.\label{fig:pu-rejects}}
\end{wrapfigure}

In general, the permissiveness of our approach is incomparable with
that of {\pu}.
\begin{wrapfigure}{l}{4cm}
\vspace{-12pt}
\small
\begin{hscode}\SaveRestoreHook
\column{B}{@{}>{\hspre}l<{\hspost}@{}}%
\column{4}{@{}>{\hspre}l<{\hspost}@{}}%
\column{E}{@{}>{\hspre}l<{\hspost}@{}}%
\>[B]{}\mathbf{if}\;\Varid{h}{}\<[E]%
\\
\>[B]{}\hsindent{4}{}\<[4]%
\>[4]{}\Varid{x}\mathbin{:=}\Conid{True}{}\<[E]%
\\
\>[B]{}\mathbf{if}\;\Varid{x}{}\<[E]%
\\
\>[B]{}\hsindent{4}{}\<[4]%
\>[4]{}\textbf{skip}{}\<[E]%
\ColumnHook
\end{hscode}\resethooks
\caption{A secure program that {\pu} rejects and {\lioafs} accepts.\label{fig:pu-rejects2}}
\vspace{-10pt}
\end{wrapfigure}
For example, without the \ensuremath{\textbf{upgrade}} operation, \liofs{} is as
expressive as {\nsu}, and thus less permissive than {\pu}.
But, with \ensuremath{\textbf{upgrade}} we can write programs in \liofs{} that would be
rejected by a {\pu} monitor.
Figure~\ref{fig:pu-rejects} shows an example of one such case.
In the example, the \ensuremath{\textbf{upgrade}} operation is used to ensure that reference \ensuremath{\Varid{x}}, which
would be marked \ensuremath{\textbf{P}} by {\pu}, ends up as \ensuremath{\textbf{H}} in all runs;
without the \ensuremath{\textbf{upgrade}}, a {\pu} monitor would reject the branch on \ensuremath{\Varid{x}}.
By inserting \ensuremath{\textbf{upgrade}} operations in the ``right'' places, our approach can
become more flexible than {\pu}.

%% When we compare {\pu} with \liofs{} directly, we see that \liofs{} is
%% a calculus with an explicit |upgrade| operation on flow-sensitive
%% references, which {\pu} lacks. If we consider an \liofs{} program
%% without |upgrade| statements, which behaves like {\nsu}, we have that
%% {\pu} is strictly more permissive. However, it is possible to write
%% programs that \liofs{} accepts and {\pu} rejects, by using the
%% |upgrade| statement to ensure that a reference which would be marked
%% |marked| by {\pu} ends up as |high| in all
%% runs. Figure~\ref{fig:pu-rejects} shows an example of such a program.

Of course, the challenge lies in upgrading references in a permissive
fashion.
And automatically upgrading references whenever the current label is
raised is not necessarily more permissive than a {\pu} monitor.
Indeed, the permissiveness of {\pu} and \lioafs{} are also incomparable.

Figure~\ref{fig:pu-rejects2} shows a program that is rejected by {\pu}
but accepted by \lioafs{}. With {\pu}, the first branch on \ensuremath{\Varid{h}}
upgrades \ensuremath{\Varid{x}} to \ensuremath{\textbf{P}}, which leads to a failure when subsequently
trying to branch \ensuremath{\Varid{x}}. With \lioafs{}, on the other hand, reference \ensuremath{\Varid{x}} would be
upgraded to \ensuremath{\textbf{H}}, permissively allowing the second branch.

%% Let us now compare {\pu} with \lioafs{}, the version of \liofs{} which
%% includes auto-upgrades. In terms of permissiveness, {\pu} and \lioafs{} are
%% incomparable. Figure~\ref{fig:pu-rejects2} shows a secure program that
%% is rejected by {\pu} but accepted by \lioafs{}. The first branch
%% contains a sensitive upgrade of |x|, which in {\pu} would get label
%% |marked|, making the program abort in the subsequent branch on
%% |x|. However, in \lioafs{} the reference |x| would be upgraded to |H|,
%% and since the second branch does nothing, the program will not be
%% stopped.

%
Conversely, Figure~\ref{fig:pu-accepts} shows a secure program that is
accepted by {\pu} but rejected by \lioafs{}. In this program, there
are two variables in scope: \ensuremath{\Varid{x}} and \ensuremath{\Varid{y}}, both initially public.  In
the first secret conditional block we assign to \ensuremath{\Varid{y}}, which with {\pu}
only upgrades variable \ensuremath{\Varid{y}} to \ensuremath{\textbf{P}}; \ensuremath{\Varid{x}} remains \ensuremath{\textbf{L}} and thus the
second branch is executed, producing a public output.  With \lioafs{},
however, unlabeling \ensuremath{\Varid{h}} (an operation implicit in the first
conditional, which inspects \ensuremath{\Varid{h}}) auto-upgrades both \ensuremath{\Varid{x}} and \ensuremath{\Varid{y}} to
\ensuremath{\textbf{H}}. As a result, the current label at the point of the \ensuremath{\textbf{output}} is
\ensuremath{\textbf{H}}, causing a failure.

\begin{figure}[b]
\small
\begin{minipage}{.5\textwidth}
\begin{hscode}\SaveRestoreHook
\column{B}{@{}>{\hspre}l<{\hspost}@{}}%
\column{4}{@{}>{\hspre}l<{\hspost}@{}}%
\column{E}{@{}>{\hspre}l<{\hspost}@{}}%
\>[B]{}\Varid{x},\Varid{y}\mathbin{:=}\Conid{True}{}\<[E]%
\\
\>[B]{}\mathbf{if}\;\Varid{h}{}\<[E]%
\\
\>[B]{}\hsindent{4}{}\<[4]%
\>[4]{}\Varid{y}\mathbin{:=}\Conid{False}{}\<[E]%
\\
\>[B]{}\mathbf{if}\;\Varid{x}{}\<[E]%
\\
\>[B]{}\hsindent{4}{}\<[4]%
\>[4]{}\textbf{output}\;(\mathrm{1}){}\<[E]%
\ColumnHook
\end{hscode}\resethooks
\caption{A secure program that {\pu} accepts and {\lioafs}
  rejects.\label{fig:pu-accepts}}
\end{minipage}
\begin{minipage}{.5\textwidth}
\small
\begin{hscode}\SaveRestoreHook
\column{B}{@{}>{\hspre}l<{\hspost}@{}}%
\column{3}{@{}>{\hspre}l<{\hspost}@{}}%
\column{4}{@{}>{\hspre}l<{\hspost}@{}}%
\column{6}{@{}>{\hspre}l<{\hspost}@{}}%
\column{E}{@{}>{\hspre}l<{\hspost}@{}}%
\>[B]{}\Varid{x},\Varid{y}\mathbin{:=}\Conid{False}{}\<[E]%
\\
\>[B]{}\textbf{withRefs}_\textsc{fs}\;(\Varid{y})\;\{\mskip1.5mu {}\<[E]%
\\
\>[B]{}\hsindent{3}{}\<[3]%
\>[3]{}\mathbf{if}\;\Varid{h}{}\<[E]%
\\
\>[3]{}\hsindent{3}{}\<[6]%
\>[6]{}\Varid{y}\mathbin{:=}\Conid{True}{}\<[E]%
\\
\>[B]{}\mskip1.5mu\}{}\<[E]%
\\
\>[B]{}\mathbf{if}\;\Varid{x}{}\<[E]%
\\
\>[B]{}\hsindent{4}{}\<[4]%
\>[4]{}\textbf{output}\;(\mathrm{1}){}\<[E]%
\ColumnHook
\end{hscode}\resethooks
\caption{A secure program that {\pu} and {\lioafs} with \ensuremath{\textbf{withRefs}_\textsc{fs}}
  accept.\label{fig:pu-accepts2}}
\vspace{-30pt}
\end{minipage}
\end{figure}
One way to address the permissiveness issues of \lioafs{} is by using
\ensuremath{\textbf{withRefs}_\textsc{fs}} to delimit the scope of the upgrades.
Figure~\ref{fig:pu-accepts2} shows a modified version of the previous
example, where \ensuremath{\Varid{y}} is explicitly marked as the only variable that
should be upgraded in the first branch. As a consequence, \ensuremath{\Varid{x}} does not get
upgraded and the program does not fail---the \ensuremath{\textbf{output}} operation is
allowed.  More generally, if there is enough static information to guide
the use of \ensuremath{\textbf{withRefs}_\textsc{fs}}, we believe that \lioafs{} could match (or exceed) the
permissiveness of {\pu}.
% Ale: I turn it down a bit that "can be more permissive than
% permissive-upgrade" since people might think that static would never be more
% permissive than dynamic systems.

\section{Related work}
\label{sec:related}

%% Label on the label idea
The \emph{label on the label} could be seen as a fixed label that dictates which
principals can read or modify the policy (inner label) of a flow-sensitive
entity. In a different setting, trust management frameworks have
explored this idea~\cite{4556677}, where role-based rules are labeled to
restrict the view on policies---the mere presence of certain policies could
become inappropriate conduits of information.

Several authors propose an
\emph{existence security label} to remove leaks due to the termination covert
channel~\cite{DBLP:conf/csfw/RafnssonHS12,Rafnson:2013} or certain behaviors
in dynamic nested data structures~\cite{Russo:2009,DBLP:conf/csfw/HedinS12}. While the
existence security label and the label on the label are structurally isomorphic,
they are used for different purposes and in different scenarios, e.g., 
inspecting labels is not allowed in~\cite{Russo:2009,DBLP:conf/csfw/HedinS12,DBLP:conf/csfw/RafnssonHS12,Rafnson:2013}.

%% Flow-sensitive monitors
Hunt and Sands~\citep{Hunt:2006} show the equivalence (modulo code
transformation) between flow-sensitive and flow-insensitive type-systems.
%Until this work, there was not a similar result concerning purely dynamic
% monitors.
Russo and Sabelfeld~\citep{Russo:2010} formally pin down
the challenge of mutable labels when using purely dynamic monitors. They prove that
monitors require static analysis in order to be more permissive than traditional
flow-sensitive type-systems.
%\hl{Pablo, you could take the text below as an inspiration for the NSU/PU section}
%Targeting purely dynamic monitors, Austin and
%Flanagan~\citep{Austin:Flanagan:PLAS09,Austin:Flanagan:PLAS10} describes the
%label-change policies \emph{\nsu} (originally proposed by
%  Zdancewic~\cite{Zdancewic02programminglanguages}) and \emph{\pu}, where the
%latter is provably more permissive than the former, i.e., it rejects fewer
%programs. The {\nsu} discipline stops execution on any attempt to change the
%label of a public variable inside a secret context.
%We saw that {\pu} allows to change the label of a public variable inside a
%secret context, marking the altered variables so that the program cannot
%subsequently branch on them. The marking consists in replacing the security
%label of the variables with |marked|, where |low canFlowTo high canFlowTo
%marked|.
Austin and Flanagan propose a \emph{privatization} operation to boost the
permissiveness of {\pu}. This technique has been recently generalized to
arbitrary lattices~\citep{Bichhawat:2014}.  Moreover, the privatization
operation can only enforce non-interference when outputs are suppressed after
branching on a marked flow-sensitive reference. Unfortunately, none of the
mentioned work consider flow-sensitive in the presence of concurrency.  In fact,
the notion of {\pu} does not easily generalize to the concurrent setting, as
this would require tracking occurrences of branches across threads.

Recently, Hritcu et al.~\citep{Breeze} propose a floating-label
system called Breeze. Like LIO, Breeze allows changes in the current
context label (i.e., \ensuremath{\Varid{pc}}) and only considered values with
flow-insensitive labels. Given the design similarities with
LIO~\citep{stefan:lio}, we believe that our results could be easily
adapted to Breeze.

%% JSFlow
Hedin et al. \citep{Hedin13} recently developed JSFlow, an IFC flow-sensitive
monitor for JavaScript. The monitor uses the {\nsu} label changing
policy. To overcome some of the restrictions imposed by this discipline, the
primitive \textbf{upgrade} is introduced to explicitly change labels. Our
upgrade operation resembles that proposed by Hedin et al. Moreover, the
extension to \ensuremath{\textbf{unlabel}} as described Section \ref{sec:flow-sensitive} can be seen as an
automatic application of \textbf{upgrade} every time that the current label gets
raised. Using testing, Birgisson et al.~\citep{Arnar2012} automatically insert
\textbf{upgrade} instructions to boost the permissiveness of {\nsu}.
We further extend this concept of (automatic) \textbf{upgrade}s to a concurrent setting.

The Operating System IFC community has also treated the mutable label
problem in the presence of purely dynamic monitors.
Specifically, IFC OSes such as
Asbestos~\cite{Efstathopoulos:2005}, HiStar~\cite{zeldovich:histar},
and Flume~\cite{krohn:flume} distinguish between subjects
(processes), and objects (files, sockets, etc.) such that the
security labels for objects are immutable, while subject labels change
according to the sensitivity of data being read.
As in language-based IFC systems, changing the label of subjects and
object can become a covert channel, if not handled appropriated.
Hence, HiStar and Flume require that the label of a subject be
done explicitly by the subject code.
Asbestos, on the other hand, allowed (unsafe) changes to labels as the result
of receiving messages under specific and safe conditions.
Our work extends these concepts with a level of indirection to allow
for changes in object labels.

%% OS-like work for the browser
Coarse-grained IFC enforcements, similar to the ones found in IFC OS
work, have been applied to web browsers. e.g, BFlow~\citep{Yip:2009}
and COWL~\cite{stefan:2014:protecting} track the flow of information at the granularity of
protection zones and context, respectively. Both can be understood as
tracking IFC at the iframe-level granularity. As in LIO, an iframe's
label, i.e., a subject's label, must be explicitly updated. While our
techniques can be applied to COWL, BFlow
does not have fine-grained labeled objects; hence the flow-sensitivity
result is only applicable at the protection zone level.  By
taking a more fine grained approach, the DOM-tree could be thought of
as being composed of flow-sensitive objects, whose security
labels change according to the dynamic behavior of the web
page, as done in~\citep{Russo:2009}.

%% Logic / verification approaches
Hoare-like logics for IFC are often
flow-sensitive~\citep[e.g.][]{Amtoft:2006,Nanevski:2011}. Different from dynamic
approaches, these logics have the ability to observe all the execution paths and
safely approximate label changes. As a result, no leaks due to label changes are
present in provably secure programs.
%% Hybrid approaches
Le Guernic et al.~\citep{LeGuernic:2006,Guernic:2007:ACM} combine dynamic and
static checks in a flow-sensitive execution monitor.
For a flow-sensitive type-system, Foster et
al.~\cite{Foster:2002:FTQ:512529.512531} propose a \textbf{restrict}
primitive that limits the use of variables' aliases in a block of
code. Our \ensuremath{\textbf{withRefs}_\textsc{fs}} is similar to \textbf{restrict} in being used to
increase the permissiveness of the analysis.

% Local Variables:
% TeX-master: "main.lhs.tex"
% TeX-command-default: "Make"
% End:
% 1 Page (Ale)

\section{Conclusions}
\label{sec:conclusion}

We presented an extension of LIO with flow-sensitive references.
As in previous flow-sensitive work, our approach allows secure label changes
using \ensuremath{\textbf{upgrade}} and \ensuremath{\textbf{downgrade}} operations, as a way to boost the permissiveness
of the IFC system, i.e., \ensuremath{\textbf{upgrade}} can be used to allow for the encoding of
programs that would otherwise be rejected by the IFC monitor.
Since manually inserting \ensuremath{\textbf{upgrade}} operations can be cumbersome, we
extend the calculus to automatically insert upgrades whenever the
current label is raised, while still giving programmers fine-grained
control over which references untrusted code can upgrade.
Importantly, our approach extends to a concurrent setting.
To the best of our knowledge, this is the first work to address the
problem of flow-sensitive label changes for a concurrent, dynamic IFC
language.
A further insight of this work was to show that, by leveraging nested
labeled objects, both the sequential and concurrent calculi with
flow-sensitive references can be encoded using only flow-insensitive
constructs.
As a consequence, our soundness proof can be reduced to an invocation
of previous results for LIO.

% Local Variables:
% TeX-master: "main.tex"
% TeX-command-default: "Make"
% End:
% 1/2 Page Total so far: 11 1/2

\ifextended

\appendix

\section{Semantics for the base calculus}
\label{sec:app:sem}

\begin{figure}[t] % sos:rules
\small
\begin{mathpar}
\inferrule[app]
{ } { \ensuremath{\Red{E}\;[\mskip1.5mu (\lambda \Varid{x}.\Varid{t}_{1})\;\Varid{t}_{2}\mskip1.5mu]\lto\Red{E}\;[\mskip1.5mu \{\mskip1.5mu \Varid{t}_{2}\mathbin{/}\Varid{x}\mskip1.5mu\}\;\Varid{t}_{1}\mskip1.5mu]} }
\and
\inferrule[fix]
{ } { \ensuremath{\Red{E}\;[\mskip1.5mu \mathbf{fix}\;(\lambda \Varid{x}.\Varid{t})\mskip1.5mu]\lto\Red{E}\;[\mskip1.5mu \{\mskip1.5mu \mathbf{fix}\;(\lambda \Varid{x}.\Varid{t})\mathbin{/}\Varid{x}\mskip1.5mu\}\;\Varid{t}}] }
\and
\inferrule[ifTrue]
{ } { \ensuremath{\Red{E}\;[\mskip1.5mu \mathbf{if}\;\Conid{True}\;\mathbf{then}\;\Varid{t}_{2}\;\mathbf{else}\;\Varid{t}_{3}\mskip1.5mu]\lto\Red{E}\;[\mskip1.5mu \Varid{t}_{2}\mskip1.5mu]} }
\and
\inferrule[ifFalse]
{ } { \ensuremath{\Red{E}\;[\mskip1.5mu \mathbf{if}\;\Conid{False}\;\mathbf{then}\;\Varid{t}_{2}\;\mathbf{else}\;\Varid{t}_{3}\mskip1.5mu]\lto\Red{E}\;[\mskip1.5mu \Varid{t}_{3}\mskip1.5mu]} }
\and
\inferrule[labelOp]
{ \ensuremath{\Varid{v}\mathrel{=}\llbracket \Varid{l}_{1}\;\otimes\;\Varid{l}_{2}\rrbracket_{\lattice}} }
{ \ensuremath{\Red{E}\;[\mskip1.5mu \Varid{l}_{1}\;\otimes\;\Varid{l}_{2}\mskip1.5mu]\lto\Red{E}\;[\mskip1.5mu \Varid{v}\mskip1.5mu]} }
\and
\inferrule[return]
{ } { \ensuremath{\conf{\Sigma}{\textbf{\Blue{E}}\;[\mskip1.5mu \mathbf{return}\;\Varid{t}\mskip1.5mu]}\lto\conf{\Sigma}{\textbf{\Blue{E}}\;[\mskip1.5mu LIO^{\Red{\textsf{{\tiny TCB}}}}\;\Varid{t}\mskip1.5mu]}} }
\and
\inferrule[bind]
{ }
{ \ensuremath{\conf{\Sigma}{\textbf{\Blue{E}}\;[\mskip1.5mu (LIO^{\Red{\textsf{{\tiny TCB}}}}\;\Varid{t}_{1})\bind \Varid{t}_{2}\mskip1.5mu]}\lto\conf{\Sigma}{\textbf{\Blue{E}}\;[\mskip1.5mu \Varid{t}_{2}\;\Varid{t}_{1}\mskip1.5mu]}} }
\end{mathpar}
\caption{Reduction rules for standard \lio{} terms.\label{fig:sos:rules}}
\end{figure}

The reduction rules for pure and monadic terms are given in
Figure~\ref{fig:sos:rules}.
We define substitution \ensuremath{\{\mskip1.5mu \Varid{t}_{2}\mathbin{/}\Varid{x}\mskip1.5mu\}\;\Varid{t}_{1}} in the
usual way: homomorphic on all operators and renaming bound names to avoid
captures.
Our label operations \ensuremath{\lub}, \ensuremath{\glb}, and \ensuremath{\flows} rely on the label-specific
implementation of these lattice operators, as used in the premise of rule
\ruleref{labelOp}; we use the meta-level partial function \ensuremath{\llbracket \cdot\rrbracket_{\lattice}}, which
maps terms to values, to precisely capture this implementation detail.

%\concept{pure red}
%
The reduction rules for pure terms are standard. For instance, in rule
\ruleref{ifTrue}, when the branch has a true condition, i.e., \ensuremath{\Red{E}\;[\mskip1.5mu \mathbf{if}\;\Conid{True}\;\mathbf{then}\;\Varid{t}_{2}\;\mathbf{else}\;\Varid{t}_{3}\mskip1.5mu]}, it reduces to the then branch (\ensuremath{\Red{E}\;[\mskip1.5mu \Varid{t}_{2}\mskip1.5mu]}).  The
rest are self-explanatory and we do not discuss them any further.
%

%\concept{monad red}
%
%The reduction rules for monadic terms deserve some attention.
%
Since all the IFC checks are performed by individual LIO terms, the
definition for \ensuremath{\mathbf{return}} and \ensuremath{(\bind )} are trivial.
The former simply reduces to a monadic value by wrapping the term with
the \ensuremath{LIO^{\Red{\textsf{{\tiny TCB}}}}} constructor, while the latter evaluates the left-hand
term and supplies the result to the right-hand term, as usual.

\section{Attack on naive flow-sensitive references}
\label{sec:app:fs-attack2}
%\begin{wrapfigure}{r}{0.55\columnwidth}
\begin{figure}
%\vspace{-20pt}
\small
\begin{hscode}\SaveRestoreHook
\column{B}{@{}>{\hspre}l<{\hspost}@{}}%
\column{3}{@{}>{\hspre}l<{\hspost}@{}}%
\column{9}{@{}>{\hspre}l<{\hspost}@{}}%
\column{24}{@{}>{\hspre}l<{\hspost}@{}}%
\column{E}{@{}>{\hspre}l<{\hspost}@{}}%
\>[B]{}\Varid{leakRef}\mathbin{::}Ref_{\textsf{{\tiny FS}}}\;{}\;\Conid{Bool}\to \Conid{LIO}\;\Conid{Bool}{}\<[E]%
\\
\>[B]{}\Varid{leakRef}\;\Varid{href}\mathrel{=}\mathbf{do}{}\<[E]%
\\
\>[B]{}\hsindent{3}{}\<[3]%
\>[3]{}\Varid{lref}{}\<[9]%
\>[9]{}\leftarrow \textbf{newRef}_{{}}\;\textbf{L}\;\Conid{True}{}\<[E]%
\\
\>[B]{}\hsindent{3}{}\<[3]%
\>[3]{}\Varid{tmp}{}\<[9]%
\>[9]{}\leftarrow \textbf{newRef}_{{}}\;\textbf{L}\;\Conid{False}{}\<[E]%
\\
\>[B]{}\hsindent{3}{}\<[3]%
\>[3]{}\textbf{toLabeled}\;\textbf{H}\mathbin{\$}\mathbf{do}\;{}\<[24]%
\>[24]{}\Varid{h}\leftarrow \textbf{readRef}_{{}}\;\Varid{href}{}\<[E]%
\\
\>[24]{}\textbf{when}\;\Varid{h}\mathbin{\$}\textbf{writeRef}_{{}}\;\Varid{tmp}\;\Conid{True}{}\<[E]%
\\
\>[B]{}\hsindent{3}{}\<[3]%
\>[3]{}\textbf{toLabeled}\;\textbf{H}\mathbin{\$}\mathbf{do}\;{}\<[24]%
\>[24]{}\Varid{t}\leftarrow \textbf{readRef}_{{}}\;\Varid{tmp}{}\<[E]%
\\
\>[24]{}\textbf{when}\;(\neg \;\Varid{t})\mathbin{\$}\textbf{writeRef}_{{}}\;\Varid{lref}\;\Conid{False}{}\<[E]%
\\
\>[B]{}\hsindent{3}{}\<[3]%
\>[3]{}\textbf{readRef}_{{}}\;\Varid{lref}{}\<[E]%
\ColumnHook
\end{hscode}\resethooks
\vspace{-5pt}
\caption{An attack in LIO with naive flow-sensitive reference
extension without \ensuremath{\textbf{labelOf}_{{}}}.\label{fig:fs-attack2}}
%\end{wrapfigure}
\end{figure}
As in the attack of Figure~\ref{fig:fs-attack}, the \ensuremath{\Varid{leakRef}} of
Figure~\ref{fig:fs-attack2} can be used to leak the value stored in a \ensuremath{\textbf{H}}
reference \ensuremath{\Varid{href}}, while keeping the current label \ensuremath{\textbf{L}}, without using \ensuremath{\textbf{labelOf}_{{}}}.
Internally, the value is leaked into public reference \ensuremath{\Varid{lref}} by leveraging the
fact that, based on a secret value, the label of a public reference (\ensuremath{\Varid{tmp}}) can
be changed (or not).
In the first \ensuremath{\textbf{toLabeled}} block, if \ensuremath{\Varid{h}\equiv \Conid{True}}, then the label of \ensuremath{\Varid{tmp}} is raised
to \ensuremath{\textbf{H}} and its value is set to \ensuremath{\Conid{True}}.
In the second \ensuremath{\textbf{toLabeled}} block, we read \ensuremath{\Varid{tmp}}, which may raise the current
label to \ensuremath{\textbf{H}} if the secret is \ensuremath{\Conid{True}} (and thus \ensuremath{\Varid{tmp}} was upgraded).
Indeed, if the secret is \ensuremath{\Conid{True}} (and thus \ensuremath{\Varid{t}\equiv \Conid{True}}) we leave the public
reference intact: \ensuremath{\Conid{True}}.
However, if the secret is \ensuremath{\Conid{False}}, the \ensuremath{\Varid{tmp}} reference is not modified in the
first \ensuremath{\textbf{toLabeled}} block and thus when reading it in the second \ensuremath{\textbf{toLabeled}}
block, the current label remains \ensuremath{\textbf{L}}, and since \ensuremath{\Varid{t}\equiv \Conid{False}}, we write \ensuremath{\Conid{False}}
into the public reference.
In both cases the value stored in \ensuremath{\Varid{lref}} corresponds to that of \ensuremath{\Varid{href}}, yet
leaving the current label and the label of \ensuremath{\Varid{lref}} intact (\ensuremath{\textbf{L}}).

\section{Embedding Theorem}
\label{app:embedding}

In this section we prove that the embedding from \liofs{} into \lio{}
preserves semantics.

We will use the following lemma for single \liofs{} steps:

\begin{lemma} [Single-step embedding]
\label{lem:step-embedding}
Let \ensuremath{\Varid{t}} be a well-typed term in \liofs{}.  Then if \ensuremath{\conf{\Sigma}{\Varid{t}}\lto\conf{\Sigma'}{\Varid{t}'}}, then there is a configuration \ensuremath{\Conid{Y}} such that \ensuremath{\conf{\llbracket \Sigma\rrbracket_{\textsc{fi}}}{\llbracket \Varid{t}\rrbracket_{\textsc{fi}}^{\Sigma}}\lto^*\Conid{Y}} and \ensuremath{\conf{\llbracket \Sigma'\rrbracket_{\textsc{fi}}}{\llbracket \Varid{t}'\rrbracket_{\textsc{fi}}^{\Sigma}}\lto^*\Conid{Y}}, i.e. \ensuremath{\conf{\llbracket \Sigma\rrbracket_{\textsc{fi}}}{\llbracket \Varid{t}\rrbracket_{\textsc{fi}}^{\Sigma}}} and \ensuremath{\conf{\llbracket \Sigma'\rrbracket_{\textsc{fi}}}{\llbracket \Varid{t}'\rrbracket_{\textsc{fi}}^{\Sigma}}} are
$\beta$-equivalent.
\end{lemma}

\begin{proof}

  Case analysis on the next redex in \ensuremath{\Varid{t}}. Most cases show a stronger
  version of the lemma, i.e. that \ensuremath{\conf{\Sigma}{\Varid{t}}\lto\conf{\Sigma'}{\Varid{t}'}} implies
  \ensuremath{\conf{\llbracket \Sigma\rrbracket_{\textsc{fi}}}{\llbracket \Varid{t}\rrbracket_{\textsc{fi}}^{\Sigma}}\lto^*\conf{\llbracket \Sigma'\rrbracket_{\textsc{fi}}}{\llbracket \Varid{t}'\rrbracket_{\textsc{fi}}^{\Sigma}}}.

{\bf Case} \ensuremath{\textbf{\Blue{E}}\;[\mskip1.5mu \textbf{newRef}_{\textsc{fs}}\;\Varid{l}\;\Varid{t}\mskip1.5mu]}.

We have \ensuremath{\conf{\Sigma}{\textbf{\Blue{E}}\;[\mskip1.5mu \textbf{newRef}_{\textsc{fs}}\;\Varid{l}\;\Varid{t}\mskip1.5mu]}\lto\conf{\Sigma'}{\textbf{\Blue{E}}\;[\mskip1.5mu \mathbf{return}\;(Ref^{\Red{\textsf{{\tiny TCB}}}}_{\textsc{fs}}\;\Varid{a})\mskip1.5mu]}}, where \ensuremath{\Sigma'\mathrel{=}\Sigma\;[\mskip1.5mu \mu_\textsc{fs}\;\mapsto\;\Sigma.\mu_\textsc{fs}\;[\mskip1.5mu \Varid{a}\;\mapsto\;Lb^{\Red{\textsf{{\tiny TCB}}}}\;\Sigma.\lcurr\;(Lb^{\Red{\textsf{{\tiny TCB}}}}\;\Varid{l}\;\Varid{t})\mskip1.5mu]\mskip1.5mu]}, and we know that \ensuremath{\Sigma.\lcurr\;\flows\;\Varid{l}}.

Let \ensuremath{\Sigma_{1}\mathrel{=}\llbracket \Sigma\rrbracket_{\textsc{fi}}}. We argue

\begin{hscode}\SaveRestoreHook
\column{B}{@{}>{\hspre}l<{\hspost}@{}}%
\column{7}{@{}>{\hspre}l<{\hspost}@{}}%
\column{38}{@{}>{\hspre}l<{\hspost}@{}}%
\column{E}{@{}>{\hspre}l<{\hspost}@{}}%
\>[7]{}\conf{\Sigma_{1}}{\llbracket \textbf{\Blue{E}}\;[\mskip1.5mu \textbf{newRef}_{\textsc{fs}}\;\Varid{l}\;\Varid{t}\mskip1.5mu]\rrbracket_{\textsc{fi}}^{\Sigma}}{}\<[E]%
\\
\>[B]{}\lto{}\<[7]%
\>[7]{}\conf{\Sigma_{1}'}{(\llbracket \textbf{\Blue{E}}\rrbracket_{\textsc{fi}}^{\Sigma})\;[\mskip1.5mu \mathbf{do}\;\lcurr\leftarrow \mathbf{getLabel};\mathbf{return}\;(\Varid{wrap}\;(Ref^{\Red{\textsf{{\tiny TCB}}}}_{\textsc{fi}}\;\lcurr\;\Varid{a}))\mskip1.5mu]}{}\<[E]%
\\
\>[7]{}(\Sigma_{1}'\mathrel{=}\Sigma_{1}\;[\mskip1.5mu \mu_\textsc{fi}\;\mapsto\;\Sigma_{1}.\mu_\textsc{fi}\;[\mskip1.5mu \Varid{i}\;\mapsto\;Lb^{\Red{\textsf{{\tiny TCB}}}}\;\Varid{l}\;\Varid{t}\mskip1.5mu]\mskip1.5mu]){}\<[E]%
\\
\>[B]{}\lto{}\<[7]%
\>[7]{}\conf{\Sigma_{1}''}{(\llbracket \textbf{\Blue{E}}\rrbracket_{\textsc{fi}}^{\Sigma})\;[\mskip1.5mu \mathbf{return}\;(\Varid{wrap}\;(Ref^{\Red{\textsf{{\tiny TCB}}}}_{\textsc{fi}}\;\lcurr\;\Varid{a}))\mskip1.5mu]}{}\<[E]%
\\
\>[7]{}(\Sigma_{1}''\mathrel{=}\Sigma_{1}\;[\mskip1.5mu \mu_\textsc{fi}\;\mapsto\;\Sigma_{1}.\mu_\textsc{fi}\;[\mskip1.5mu {}\<[38]%
\>[38]{}\Varid{i}\;\mapsto\;Lb^{\Red{\textsf{{\tiny TCB}}}}\;\Varid{l}\;\Varid{t};{}\<[E]%
\\
\>[38]{}\Varid{a}\;\mapsto\;Lb^{\Red{\textsf{{\tiny TCB}}}}\;\lcurr\;(Ref^{\Red{\textsf{{\tiny TCB}}}}_{\textsc{fi}}\;\Varid{l}\;\Varid{i})\mskip1.5mu]\mskip1.5mu]){}\<[E]%
\ColumnHook
\end{hscode}\resethooks

We now have to check that \ensuremath{\llbracket \mathbf{return}\;(Ref^{\Red{\textsf{{\tiny TCB}}}}_{\textsc{fs}}\;\Varid{a})\rrbracket_{\textsc{fi}}^{\Sigma'}\mathrel{=}\mathbf{return}\;(\Varid{wrap}\;(Ref^{\Red{\textsf{{\tiny TCB}}}}_{\textsc{fi}}\;\lcurr\;\Varid{a}))} and \ensuremath{\llbracket \Sigma'\rrbracket_{\textsc{fi}}\mathrel{=}\Sigma_{1}''}, which follow directly from
the definition of \ensuremath{\llbracket \cdot \rrbracket_{\textsc{fi}}} for references and states.
\medskip

{\bf Case} \ensuremath{\textbf{\Blue{E}}\;[\mskip1.5mu \textbf{readRef}_{\textsc{fs}}\;(Ref^{\Red{\textsf{{\tiny TCB}}}}_{\textsc{fs}}\;\Varid{a})\mskip1.5mu]}.

We have \ensuremath{\conf{\Sigma}{\textbf{\Blue{E}}\;[\mskip1.5mu \textbf{readRef}_{\textsc{fs}}\;(Ref^{\Red{\textsf{{\tiny TCB}}}}_{\textsc{fs}}\;\Varid{a})\mskip1.5mu]}\lto\conf{\Sigma}{\textbf{\Blue{E}}\;[\mskip1.5mu \textbf{unlabel}\;(Lb^{\Red{\textsf{{\tiny TCB}}}}\;(\Varid{l}\;\lub\;\Varid{l'})\;\Varid{t})\mskip1.5mu]}}, where \ensuremath{\Sigma.\mu_\textsc{fs}\;(\Varid{a})\mathrel{=}Lb^{\Red{\textsf{{\tiny TCB}}}}\;\Varid{l}\;(Lb^{\Red{\textsf{{\tiny TCB}}}}\;\Varid{l'}\;\Varid{t})}.

Let \ensuremath{\Sigma_{1}\mathrel{=}\llbracket \Sigma\rrbracket_{\textsc{fi}}}. We argue

\begin{hscode}\SaveRestoreHook
\column{B}{@{}>{\hspre}l<{\hspost}@{}}%
\column{7}{@{}>{\hspre}l<{\hspost}@{}}%
\column{E}{@{}>{\hspre}l<{\hspost}@{}}%
\>[7]{}\conf{\Sigma_{1}}{\llbracket \textbf{\Blue{E}}\;[\mskip1.5mu \textbf{readRef}_{\textsc{fs}}\;(Ref^{\Red{\textsf{{\tiny TCB}}}}_{\textsc{fs}}\;\Varid{a})\mskip1.5mu]\rrbracket_{\textsc{fi}}^{\Sigma}}{}\<[E]%
\\
\>[B]{}\lto{}\<[7]%
\>[7]{}\conf{\Sigma_{1}}{(\llbracket \textbf{\Blue{E}}\rrbracket_{\textsc{fi}}^{\Sigma})\;[\mskip1.5mu \textbf{readRef}_{\textsc{fi}}\;(\llbracket Ref^{\Red{\textsf{{\tiny TCB}}}}_{\textsc{fs}}\;\Varid{a}\rrbracket_{\textsc{fi}}^{\cdot })\bind \textbf{readRef}_{\textsc{fi}}\mskip1.5mu]}{}\<[E]%
\\
\>[B]{}\lto{}\<[7]%
\>[7]{}\conf{\Sigma_{1}}{(\llbracket \textbf{\Blue{E}}\rrbracket_{\textsc{fi}}^{\Sigma})\;[\mskip1.5mu \textbf{unlabel}\;(\Sigma_{1}.\mu_\textsc{fi}\;(\Varid{a}))\bind \textbf{readRef}_{\textsc{fi}}\mskip1.5mu]}{}\<[E]%
\\
\>[B]{}\lto{}\<[7]%
\>[7]{}\conf{\Sigma_{1}'}{(\llbracket \textbf{\Blue{E}}\rrbracket_{\textsc{fi}}^{\Sigma})\;[\mskip1.5mu \mathbf{return}\;(Ref^{\Red{\textsf{{\tiny TCB}}}}_{\textsc{fi}}\;\Varid{l'}\;\Varid{i})\bind \textbf{readRef}_{\textsc{fi}}\mskip1.5mu]}{}\<[E]%
\\
\>[7]{}(\Sigma_{1}'\mathrel{=}\Sigma_{1}\;[\mskip1.5mu \lcurr\;\mapsto\;\Sigma_{1}.\lcurr\;\lub\;\Varid{l}\mskip1.5mu]){}\<[E]%
\\
\>[B]{}\lto{}\<[7]%
\>[7]{}\conf{\Sigma_{1}'}{(\llbracket \textbf{\Blue{E}}\rrbracket_{\textsc{fi}}^{\Sigma})\;[\mskip1.5mu \textbf{readRef}_{\textsc{fi}}\;(Ref^{\Red{\textsf{{\tiny TCB}}}}_{\textsc{fi}}\;\Varid{l'}\;\Varid{i})\mskip1.5mu]}{}\<[E]%
\\
\>[B]{}\lto{}\<[7]%
\>[7]{}\conf{\Sigma_{1}'}{(\llbracket \textbf{\Blue{E}}\rrbracket_{\textsc{fi}}^{\Sigma})\;[\mskip1.5mu \textbf{unlabel}\;(\Sigma_{1}'.\mu_\textsc{fi}\;(\Varid{i}))\mskip1.5mu]}{}\<[E]%
\\
\>[B]{}\lto{}\<[7]%
\>[7]{}\conf{\Sigma_{1}''}{(\llbracket \textbf{\Blue{E}}\rrbracket_{\textsc{fi}}^{\Sigma})\;[\mskip1.5mu \mathbf{return}\;\Varid{t}\mskip1.5mu]}{}\<[E]%
\\
\>[7]{}(\Sigma_{1}''\mathrel{=}\Sigma_{1}'\;[\mskip1.5mu \lcurr\;\mapsto\;\Sigma_{1}'.\lcurr\;\lub\;\Varid{l'}\mskip1.5mu]){}\<[E]%
\ColumnHook
\end{hscode}\resethooks

Now if we consider \ensuremath{\conf{\llbracket \Sigma\rrbracket_{\textsc{fi}}}{\llbracket \textbf{\Blue{E}}\;[\mskip1.5mu \textbf{unlabel}\;(Lb^{\Red{\textsf{{\tiny TCB}}}}\;(\Varid{l}\;\lub\;\Varid{l'})\;\Varid{t})\mskip1.5mu]\rrbracket_{\textsc{fi}}^{\cdot }}}, we have

\begin{hscode}\SaveRestoreHook
\column{B}{@{}>{\hspre}l<{\hspost}@{}}%
\column{7}{@{}>{\hspre}l<{\hspost}@{}}%
\column{E}{@{}>{\hspre}l<{\hspost}@{}}%
\>[7]{}\conf{\llbracket \Sigma\rrbracket_{\textsc{fi}}}{\llbracket \textbf{\Blue{E}}\;[\mskip1.5mu \textbf{unlabel}\;(Lb^{\Red{\textsf{{\tiny TCB}}}}\;(\Varid{l}\;\lub\;\Varid{l'})\;\Varid{t})\mskip1.5mu]\rrbracket_{\textsc{fi}}^{\cdot }}{}\<[E]%
\\
\>[B]{}\lto{}\<[7]%
\>[7]{}\conf{\Sigma_{2}}{(\llbracket \textbf{\Blue{E}}\rrbracket_{\textsc{fi}}^{\Sigma})\;[\mskip1.5mu \mathbf{return}\;\Varid{t}\mskip1.5mu]{}\<[E]%
\\
\>[7]{}(\Sigma_{2}\mathrel{=}\Sigma_{2}\;[\mskip1.5mu \lcurr\;\mapsto\;(\llbracket \Sigma\rrbracket_{\textsc{fi}}).\lcurr\;\lub\;\Varid{l}\;\lub\;\Varid{l'}\mskip1.5mu])}{}\<[E]%
\ColumnHook
\end{hscode}\resethooks

Note that \ensuremath{\Sigma_{1}''.\lcurr\mathrel{=}(\llbracket \Sigma\rrbracket_{\textsc{fi}}).\lcurr\;\lub\;\Varid{l}\;\lub\;\Varid{l'}\mathrel{=}\Sigma_{2}.\lcurr}.

{\bf Case} \ensuremath{\textbf{\Blue{E}}\;[\mskip1.5mu \textbf{writeRef}_{\textsc{fs}}\;(Ref^{\Red{\textsf{{\tiny TCB}}}}_{\textsc{fs}}\;\Varid{a})\;\Varid{t}\mskip1.5mu]}.

We have \ensuremath{\conf{\Sigma}{\textbf{\Blue{E}}\;[\mskip1.5mu \textbf{writeRef}_{\textsc{fs}}\;(Ref^{\Red{\textsf{{\tiny TCB}}}}_{\textsc{fs}}\;\Varid{a})\;\Varid{t}\mskip1.5mu]}\lto\conf{\Sigma'}{\textbf{\Blue{E}}\;[\mskip1.5mu \mathbf{return}\;()\mskip1.5mu]}}, where \ensuremath{\Sigma.\mu_\textsc{fs}\;(\Varid{a})\mathrel{=}Lb^{\Red{\textsf{{\tiny TCB}}}}\;\Varid{l}\;(Lb^{\Red{\textsf{{\tiny TCB}}}}\;\Varid{l'}\;\Varid{v})}, \ensuremath{\Sigma'\mathrel{=}\Sigma\;[\mskip1.5mu \mu_\textsc{fs}\;\mapsto\;\Sigma.\mu_\textsc{fs}\;[\mskip1.5mu \Varid{a}\;\mapsto\;Lb^{\Red{\textsf{{\tiny TCB}}}}\;\Varid{l}\;(Lb^{\Red{\textsf{{\tiny TCB}}}}\;\Varid{l'}\;\Varid{v})\mskip1.5mu]\mskip1.5mu]},
and we know that \ensuremath{\Sigma.\lcurr\;\flows\;\Varid{l}\;\lub\;\Varid{l'}}.

Let \ensuremath{\Sigma_{1}\mathrel{=}\llbracket \Sigma\rrbracket_{\textsc{fi}}}. Then there exists a function \ensuremath{\mu} such that:

\begin{hscode}\SaveRestoreHook
\column{B}{@{}>{\hspre}l<{\hspost}@{}}%
\column{7}{@{}>{\hspre}l<{\hspost}@{}}%
\column{E}{@{}>{\hspre}l<{\hspost}@{}}%
\>[7]{}\conf{\Sigma_{1}}{\llbracket \textbf{\Blue{E}}\;[\mskip1.5mu \textbf{upgrade}_\textsc{fs}\;(Ref^{\Red{\textsf{{\tiny TCB}}}}_{\textsc{fs}}\;\Varid{a})\;\Varid{l'}\mskip1.5mu]\rrbracket_{\textsc{fi}}^{\Sigma}}{}\<[E]%
\\
\>[B]{}\lto{}\<[7]%
\>[7]{}\conf{\Sigma_{1}}{(\llbracket \textbf{\Blue{E}}\rrbracket_{\textsc{fi}}^{\Sigma})\;[\mskip1.5mu \textbf{toLabeled}\;(\Sigma_{1}.\lcurr\;\lub\;\Varid{l})\;(\mu\;\Sigma_{1}.\lcurr)\mskip1.5mu]}{}\<[E]%
\ColumnHook
\end{hscode}\resethooks

We now step through the evaluation of \ensuremath{\conf{\Sigma_{1}}{\mu\;\Sigma_{1}.\lcurr}}, as follows:

\begin{hscode}\SaveRestoreHook
\column{B}{@{}>{\hspre}l<{\hspost}@{}}%
\column{7}{@{}>{\hspre}l<{\hspost}@{}}%
\column{E}{@{}>{\hspre}l<{\hspost}@{}}%
\>[7]{}\conf{\Sigma_{1}}{\mu\;\Sigma_{1}.\lcurr}{}\<[E]%
\\
\>[B]{}\lto{}\<[7]%
\>[7]{}\conf{\Sigma_{1}}{\;\mathbf{do}\;\Varid{i}\leftarrow \textbf{readRef}_{\llbracket Ref^{\Red{\textsf{{\tiny TCB}}}}_{\textsc{fs}}\;\Varid{a}\rrbracket_{\textsc{fi}}^{\cdot }};\textbf{writeRef}_{\textsc{fi}}\;\Varid{i}\;\Varid{t}}{}\<[E]%
\\
\>[B]{}\lto{}\<[7]%
\>[7]{}\conf{\Sigma_{1}'}{\textbf{writeRef}_{\textsc{fi}}\;\Varid{i}\;\Varid{t}}{}\<[E]%
\\
\>[7]{}(\Sigma_{1}'\mathrel{=}\Sigma_{1}\;[\mskip1.5mu \lcurr\;\mapsto\;\Sigma_{1}.\lcurr\;\lub\;\Varid{l}\mskip1.5mu]){}\<[E]%
\\
\>[B]{}\lto\conf{\Sigma_{1}''}{\mathbf{return}\;()}\;{}\<[E]%
\\
\>[B]{}\hsindent{7}{}\<[7]%
\>[7]{}(\Sigma_{1}''\mathrel{=}\Sigma_{1}'\;[\mskip1.5mu \mu_\textsc{fi}\;\mapsto\;\Sigma_{1}'.\mu_\textsc{fi}\;[\mskip1.5mu \Varid{i}\;\mapsto\;Lb^{\Red{\textsf{{\tiny TCB}}}}\;\Varid{l'}\;\Varid{t}\mskip1.5mu]\mskip1.5mu]){}\<[E]%
\ColumnHook
\end{hscode}\resethooks

Finally, this allows us to conclude (from the rule for \ensuremath{\textbf{toLabeled}}), that

\begin{hscode}\SaveRestoreHook
\column{B}{@{}>{\hspre}l<{\hspost}@{}}%
\column{E}{@{}>{\hspre}l<{\hspost}@{}}%
\>[B]{}\conf{\Sigma_{1}}{(\llbracket \textbf{\Blue{E}}\rrbracket_{\textsc{fi}}^{\Sigma})\;[\mskip1.5mu \textbf{toLabeled}\;(\Sigma_{1}.\lcurr\;\lub\;\Varid{l})\;(\mu\;\Sigma_{1}.\lcurr)\mskip1.5mu]}{}\<[E]%
\\
\>[B]{}\lto\conf{\Sigma_{2}}{(\llbracket \textbf{\Blue{E}}\rrbracket_{\textsc{fi}}^{\Sigma})\;[\mskip1.5mu \mathbf{return}\;()\mskip1.5mu]}{}\<[E]%
\ColumnHook
\end{hscode}\resethooks

where \ensuremath{\Sigma_{2}\mathrel{=}(\Sigma_{1}.\lcurr,\Sigma_{1}''.\mu_\textsc{fi})}. Now we can check that \ensuremath{\llbracket \Sigma'\rrbracket_{\textsc{fi}}\mathrel{=}\Sigma_{2}} from the definition of \ensuremath{\llbracket \cdot \rrbracket_{\textsc{fi}}} for states.

{\bf Case} \ensuremath{\textbf{\Blue{E}}\;[\mskip1.5mu \textbf{upgrade}_\textsc{fs}\;(Ref^{\Red{\textsf{{\tiny TCB}}}}_{\textsc{fs}}\;\Varid{a})\;\Varid{l'}\mskip1.5mu]}.

We have \ensuremath{\conf{\Sigma}{\textbf{\Blue{E}}\;[\mskip1.5mu \textbf{upgrade}_\textsc{fs}\;(Ref^{\Red{\textsf{{\tiny TCB}}}}_{\textsc{fs}}\;\Varid{a})\;\Varid{l'}\mskip1.5mu]}\lto\conf{\Sigma'}{\textbf{\Blue{E}}\;[\mskip1.5mu \mathbf{return}\;()\mskip1.5mu]}}, where \ensuremath{\Sigma.\mu_\textsc{fs}\;(\Varid{a})\mathrel{=}Lb^{\Red{\textsf{{\tiny TCB}}}}\;\Varid{l}\;(Lb^{\Red{\textsf{{\tiny TCB}}}}\;\Varid{l''}\;\Varid{v})},
\ensuremath{\Sigma'\mathrel{=}\Sigma\;[\mskip1.5mu \mu_\textsc{fs}\;\mapsto\;\Sigma.\mu_\textsc{fs}\;[\mskip1.5mu \Varid{a}\;\mapsto\;Lb^{\Red{\textsf{{\tiny TCB}}}}\;\Varid{l}\;(Lb^{\Red{\textsf{{\tiny TCB}}}}\;(\Varid{l}\;\lub\;\Varid{l''}\;\lub\;\Varid{l'})\;\Varid{v})\mskip1.5mu]\mskip1.5mu]}, and we know that \ensuremath{\Sigma.\lcurr\;\flows\;\Varid{l}}.

Let \ensuremath{\Sigma_{1}\mathrel{=}\llbracket \Sigma\rrbracket_{\textsc{fi}}}. Then there exists a function \ensuremath{\mu} such that:

\begin{hscode}\SaveRestoreHook
\column{B}{@{}>{\hspre}l<{\hspost}@{}}%
\column{7}{@{}>{\hspre}l<{\hspost}@{}}%
\column{E}{@{}>{\hspre}l<{\hspost}@{}}%
\>[7]{}\conf{\Sigma_{1}}{\llbracket \textbf{\Blue{E}}\;[\mskip1.5mu \textbf{upgrade}_\textsc{fs}\;(Ref^{\Red{\textsf{{\tiny TCB}}}}_{\textsc{fs}}\;\Varid{a})\;\Varid{l'}\mskip1.5mu]\rrbracket_{\textsc{fi}}^{\Sigma}}{}\<[E]%
\\
\>[B]{}\lto{}\<[7]%
\>[7]{}\conf{\Sigma_{1}}{(\llbracket \textbf{\Blue{E}}\rrbracket_{\textsc{fi}}^{\Sigma})\;[\mskip1.5mu \textbf{toLabeled}\;(\Sigma_{1}.\lcurr\;\lub\;\Varid{l})\;(\mu\;\Sigma_{1}.\lcurr)\mskip1.5mu]}{}\<[E]%
\ColumnHook
\end{hscode}\resethooks

We now step through the evaluation of \ensuremath{\conf{\Sigma_{1}}{\mu\;\Sigma_{1}.\lcurr}}, as follows:

\begin{hscode}\SaveRestoreHook
\column{B}{@{}>{\hspre}l<{\hspost}@{}}%
\column{7}{@{}>{\hspre}l<{\hspost}@{}}%
\column{35}{@{}>{\hspre}l<{\hspost}@{}}%
\column{51}{@{}>{\hspre}l<{\hspost}@{}}%
\column{E}{@{}>{\hspre}l<{\hspost}@{}}%
\>[7]{}\conf{\Sigma_{1}}{\mu\;\Sigma_{1}.\lcurr}{}\<[E]%
\\
\>[B]{}\lto{}\<[7]%
\>[7]{}\conf{\Sigma_{1}}{\;\mathbf{do}\;\Varid{i}\leftarrow \textbf{readRef}_{\llbracket Ref^{\Red{\textsf{{\tiny TCB}}}}_{\textsc{fs}}\;\Varid{a}\rrbracket_{\textsc{fi}}^{\cdot }};\Varid{lc}\leftarrow \mathbf{getLabel};\mathbin{...}}{}\<[E]%
\\
\>[B]{}\lto{}\<[7]%
\>[7]{}\conf{\Sigma_{1}'}{\;\mathbf{do}\;\Varid{lc}\leftarrow \mathbf{getLabel};\Varid{n}\leftarrow \textbf{newRef}_{\textsc{fi}}\;{}\<[51]%
\>[51]{}(\Varid{lc}\;\lub\;(\Varid{l'}\;\lub\;\Varid{l}))\;\bot ;\mathbin{...}}{}\<[E]%
\\
\>[7]{}(\Sigma_{1}'\mathrel{=}\Sigma_{1}\;[\mskip1.5mu \lcurr\;\mapsto\;\Sigma_{1}.\lcurr\;\lub\;\Varid{l}\mskip1.5mu]){}\<[E]%
\\
\>[B]{}\lto{}\<[7]%
\>[7]{}\conf{\Sigma_{1}'}{\;\mathbf{do}\;\Varid{n}\leftarrow \textbf{newRef}_{\textsc{fi}}\;{}\<[35]%
\>[35]{}(\Varid{lc}\;\lub\;(\Varid{l'}\;\lub\;\Varid{l}))\;\bot ;\textbf{copyRef}\;\Varid{i}\;\Varid{n};\mathbin{...}}{}\<[E]%
\\
\>[B]{}\lto{}\<[7]%
\>[7]{}\conf{\Sigma_{1}''}{\;\mathbf{do}\;\textbf{copyRef}\;\Varid{i}\;\Varid{n};\textbf{writeRef}_{\textsc{fi}}\;(\llbracket Ref^{\Red{\textsf{{\tiny TCB}}}}_{\textsc{fs}}\;\Varid{a}\rrbracket_{\textsc{fi}}^{\cdot })\;\Varid{n}}{}\<[E]%
\\
\>[7]{}(\Sigma_{1}''\mathrel{=}\Sigma_{1}'\;[\mskip1.5mu \mu_\textsc{fi}\;\mapsto\;\Sigma_{1}'.\mu_\textsc{fi}\;[\mskip1.5mu \Varid{n}\;\mapsto\;Lb^{\Red{\textsf{{\tiny TCB}}}}\;(\Varid{lc}\;\lub\;(\Varid{l'}\;\lub\;\Varid{l}))\;\bot \mskip1.5mu]\mskip1.5mu]){}\<[E]%
\\
\>[B]{}\lto{}\<[7]%
\>[7]{}\conf{\Sigma_{1}'''}{\textbf{writeRef}_{\textsc{fi}}\;(\llbracket Ref^{\Red{\textsf{{\tiny TCB}}}}_{\textsc{fs}}\;\Varid{a}\rrbracket_{\textsc{fi}}^{\cdot })\;\Varid{n}}{}\<[E]%
\\
\>[7]{}(\Sigma_{1}'''\mathrel{=}\Sigma_{1}''\;[\mskip1.5mu \mu_\textsc{fi}\;\mapsto\;\Sigma_{1}''.\mu_\textsc{fi}\;[\mskip1.5mu \Varid{n}\;\mapsto\;Lb^{\Red{\textsf{{\tiny TCB}}}}\;(\Varid{lc}\;\lub\;(\Varid{l'}\;\lub\;\Varid{l}))\;\Varid{v}\mskip1.5mu]\mskip1.5mu]){}\<[E]%
\\
\>[B]{}\lto\conf{\Sigma_{1}'''}{\mathbf{return}\;()}\;{}\<[E]%
\\
\>[B]{}\hsindent{7}{}\<[7]%
\>[7]{}(\Sigma_{1}''''\mathrel{=}\Sigma_{1}'''\;[\mskip1.5mu \mu_\textsc{fi}\;\mapsto\;\Sigma_{1}'''.\mu_\textsc{fi}\;[\mskip1.5mu \Varid{a}\;\mapsto\;Lb^{\Red{\textsf{{\tiny TCB}}}}\;\Varid{l}\;(Ref^{\Red{\textsf{{\tiny TCB}}}}_{\textsc{fi}}\;(\Varid{lc}\;\lub\;\Varid{l'}\;\lub\;\Varid{l})\;\Varid{n})\mskip1.5mu]\mskip1.5mu]){}\<[E]%
\ColumnHook
\end{hscode}\resethooks

Finally, this allows us to conclude (from the rule for \ensuremath{\textbf{toLabeled}}), that

\begin{hscode}\SaveRestoreHook
\column{B}{@{}>{\hspre}l<{\hspost}@{}}%
\column{E}{@{}>{\hspre}l<{\hspost}@{}}%
\>[B]{}\conf{\Sigma_{1}}{(\llbracket \textbf{\Blue{E}}\rrbracket_{\textsc{fi}}^{\Sigma})\;[\mskip1.5mu \textbf{toLabeled}\;\Varid{l'}\;(\mu\;\Sigma_{1}.\lcurr)\mskip1.5mu]}{}\<[E]%
\\
\>[B]{}\lto\conf{\Sigma_{2}}{(\llbracket \textbf{\Blue{E}}\rrbracket_{\textsc{fi}}^{\Sigma})\;[\mskip1.5mu \mathbf{return}\;()\mskip1.5mu]}{}\<[E]%
\ColumnHook
\end{hscode}\resethooks

where \ensuremath{\Sigma_{2}\mathrel{=}(\Sigma_{1}.\lcurr,\Sigma_{1}''''.\mu_\textsc{fi})}. Now we can check that \ensuremath{\llbracket \Sigma'\rrbracket_{\textsc{fi}}\mathrel{=}\Sigma_{2}} from the definition of \ensuremath{\llbracket \cdot \rrbracket_{\textsc{fi}}} for states.

{\bf Case} \ensuremath{\textbf{\Blue{E}}\;[\mskip1.5mu \textbf{downgrade}_\textsc{fs}\;(Ref^{\Red{\textsf{{\tiny TCB}}}}_{\textsc{fs}}\;\Varid{a})\;\Varid{l'}\mskip1.5mu]}.

We have \ensuremath{\conf{\Sigma}{\textbf{\Blue{E}}\;[\mskip1.5mu \textbf{downgrade}_\textsc{fs}\;(Ref^{\Red{\textsf{{\tiny TCB}}}}_{\textsc{fs}}\;\Varid{a})\;\Varid{l'}\mskip1.5mu]}\lto\conf{\Sigma'}{\textbf{\Blue{E}}\;[\mskip1.5mu \mathbf{return}\;()\mskip1.5mu]}}, where \ensuremath{\Sigma.\mu_\textsc{fs}\;(\Varid{a})\mathrel{=}Lb^{\Red{\textsf{{\tiny TCB}}}}\;\Varid{l}\;(Lb^{\Red{\textsf{{\tiny TCB}}}}\;\Varid{l''}\;\Varid{v})},
\ensuremath{\Sigma'\mathrel{=}\Sigma\;[\mskip1.5mu \mu_\textsc{fs}\;\mapsto\;\Sigma.\mu_\textsc{fs}\;[\mskip1.5mu \Varid{a}\;\mapsto\;Lb^{\Red{\textsf{{\tiny TCB}}}}\;\Varid{l}\;(Lb^{\Red{\textsf{{\tiny TCB}}}}\;(\Varid{l}\;\lub\;\Varid{l''}\;\glb\;\Varid{l'})\;\bot )\mskip1.5mu]\mskip1.5mu]}, and we know that \ensuremath{\Sigma.\lcurr\;\flows\;\Varid{l}}.

Let \ensuremath{\Sigma_{1}\mathrel{=}\llbracket \Sigma\rrbracket_{\textsc{fi}}}. Then there exists a function \ensuremath{\mu} such that:

\begin{hscode}\SaveRestoreHook
\column{B}{@{}>{\hspre}l<{\hspost}@{}}%
\column{7}{@{}>{\hspre}l<{\hspost}@{}}%
\column{E}{@{}>{\hspre}l<{\hspost}@{}}%
\>[7]{}\conf{\Sigma_{1}}{\llbracket \textbf{\Blue{E}}\;[\mskip1.5mu \textbf{downgrade}_\textsc{fs}\;(Ref^{\Red{\textsf{{\tiny TCB}}}}_{\textsc{fs}}\;\Varid{a})\;\Varid{l'}\mskip1.5mu]\rrbracket_{\textsc{fi}}^{\Sigma}}{}\<[E]%
\\
\>[B]{}\lto{}\<[7]%
\>[7]{}\conf{\Sigma_{1}}{(\llbracket \textbf{\Blue{E}}\rrbracket_{\textsc{fi}}^{\Sigma})\;[\mskip1.5mu \textbf{toLabeled}\;(\Sigma_{1}.\lcurr\;\lub\;\Varid{l})\;(\mu\;\Sigma_{1}.\lcurr)\mskip1.5mu]}{}\<[E]%
\ColumnHook
\end{hscode}\resethooks

We now step through the evaluation of \ensuremath{\conf{\Sigma_{1}}{\mu\;\Sigma_{1}.\lcurr}}, as follows:

\begin{hscode}\SaveRestoreHook
\column{B}{@{}>{\hspre}l<{\hspost}@{}}%
\column{7}{@{}>{\hspre}l<{\hspost}@{}}%
\column{13}{@{}>{\hspre}l<{\hspost}@{}}%
\column{35}{@{}>{\hspre}l<{\hspost}@{}}%
\column{51}{@{}>{\hspre}l<{\hspost}@{}}%
\column{E}{@{}>{\hspre}l<{\hspost}@{}}%
\>[7]{}\conf{\Sigma_{1}}{\mu\;\Sigma_{1}.\lcurr}{}\<[E]%
\\
\>[B]{}\lto{}\<[7]%
\>[7]{}\conf{\Sigma_{1}}{\;\mathbf{do}\;\Varid{i}\leftarrow \textbf{readRef}_{\llbracket Ref^{\Red{\textsf{{\tiny TCB}}}}_{\textsc{fs}}\;\Varid{a}\rrbracket_{\textsc{fi}}^{\cdot }};\Varid{lc}\leftarrow \Varid{getlabel};\mathbin{...}}{}\<[E]%
\\
\>[B]{}\lto{}\<[7]%
\>[7]{}\conf{\Sigma_{1}'}{\;\mathbf{do}\;\Varid{lc}\leftarrow \mathbf{getLabel};\Varid{n}\leftarrow \textbf{newRef}_{\textsc{fi}}\;{}\<[51]%
\>[51]{}(\Varid{lc}\;\lub\;(\Varid{l'}\;\glb\;\Varid{l}))\;\bot ;\mathbin{...}}{}\<[E]%
\\
\>[7]{}(\Sigma_{1}'\mathrel{=}\Sigma_{1}\;[\mskip1.5mu \lcurr\;\mapsto\;\Sigma_{1}.\lcurr\;\lub\;\Varid{l}\mskip1.5mu]){}\<[E]%
\\
\>[B]{}\lto{}\<[7]%
\>[7]{}\conf{\Sigma_{1}'}{\;\mathbf{do}\;\Varid{n}\leftarrow \textbf{newRef}_{\textsc{fi}}\;{}\<[35]%
\>[35]{}(\Varid{lc}\;\lub\;(\Varid{l'}\;\glb\;\Varid{l}))\;\bot ;\textbf{writeRef}_{\textsc{fi}}\;(\llbracket Ref^{\Red{\textsf{{\tiny TCB}}}}_{\textsc{fs}}\;\Varid{a}\rrbracket_{\textsc{fi}}^{\cdot })\;\Varid{n}}{}\<[E]%
\\
\>[B]{}\lto{}\<[7]%
\>[7]{}\conf{\Sigma_{1}''}{\textbf{writeRef}_{\textsc{fi}}\;(\llbracket Ref^{\Red{\textsf{{\tiny TCB}}}}_{\textsc{fs}}\;\Varid{a}\rrbracket_{\textsc{fi}}^{\cdot })\;\Varid{n}}{}\<[E]%
\\
\>[7]{}(\Sigma_{1}''\mathrel{=}\Sigma_{1}'\;[\mskip1.5mu \mu_\textsc{fi}\;\mapsto\;\Sigma_{1}'.\mu_\textsc{fi}\;[\mskip1.5mu \Varid{n}\;\mapsto\;Lb^{\Red{\textsf{{\tiny TCB}}}}\;(\Varid{lc}\;\lub\;(\Varid{l'}\;\glb\;\Varid{l}))\;\bot \mskip1.5mu]\mskip1.5mu]){}\<[E]%
\\
\>[B]{}\lto\conf{\Sigma_{1}'''}{\mathbf{return}\;()}{}\<[E]%
\\[\blanklineskip]%
\ColumnHook
\end{hscode}\resethooks

Finally, this allows us to conclude (from the rule for \ensuremath{\textbf{toLabeled}}), that

\begin{hscode}\SaveRestoreHook
\column{B}{@{}>{\hspre}l<{\hspost}@{}}%
\column{E}{@{}>{\hspre}l<{\hspost}@{}}%
\>[B]{}\conf{\Sigma_{1}}{(\llbracket \textbf{\Blue{E}}\rrbracket_{\textsc{fi}}^{\Sigma})\;[\mskip1.5mu \textbf{toLabeled}\;\Varid{l'}\;(\mu\;\Sigma_{1}.\lcurr)\mskip1.5mu]}{}\<[E]%
\\
\>[B]{}\lto\conf{\Sigma_{2}}{(\llbracket \textbf{\Blue{E}}\rrbracket_{\textsc{fi}}^{\Sigma})\;[\mskip1.5mu \mathbf{return}\;()\mskip1.5mu]}{}\<[E]%
\ColumnHook
\end{hscode}\resethooks

where \ensuremath{\Sigma_{2}\mathrel{=}(\Sigma_{1}.\lcurr,\Sigma_{1}'''.\mu_\textsc{fi})}. Now we can check that \ensuremath{\llbracket \Sigma'\rrbracket_{\textsc{fi}}\mathrel{=}\Sigma_{2}} from the definition of \ensuremath{\llbracket \cdot \rrbracket_{\textsc{fi}}} for states.

\end{proof}

Now we can state the main theorem of this section.

{\bf Theorem.} [Embedding \liofs{} in \lio{}]
Let \ensuremath{\Varid{t}} be a well-typed term in \liofs{}.
  Then if \ensuremath{\conf{\Sigma}{\Varid{t}}\lto^*\conf{\Sigma'}{\Varid{v}}}, we have \ensuremath{\conf{\llbracket \Sigma\rrbracket_{\textsc{fi}}}{\llbracket \Varid{t}\rrbracket_{\textsc{fi}}^{\Sigma}}\lto^*\conf{\llbracket \Sigma'\rrbracket_{\textsc{fi}}}{\llbracket \Varid{v}\rrbracket_{\textsc{fi}}^{\Sigma}}}, and if
  \ensuremath{\conf{\Sigma}{\Varid{t}}\lto^*\conf{\Sigma'}{{\Uparrow}}}, then
  \ensuremath{\conf{\llbracket \Sigma\rrbracket_{\textsc{fi}}}{\llbracket \Varid{t}\rrbracket_{\textsc{fi}}^{\Sigma}}\lto^*\conf{\llbracket \Sigma'\rrbracket_{\textsc{fi}}}{{\Uparrow}}}.
  % \textrm{\hl{this is not right:
  %     the constructors for FS references are different (one wraps FI
  %     references, the other doesn't), so syntactic equiality is not
  %     right.  we need the equivalence relation I defined before}}

\begin{proof}
  By induction on the number of steps in \ensuremath{\conf{\Sigma}{\Varid{t}}\lto^*\conf{\Sigma'}{\Varid{v}}},
  using Lemma~\ref{lem:step-embedding} and uniqueness of normal forms
  in \lio{}.
\end{proof}
\fi

\section*{Acknowledgments}
We thank our colleagues in the ProSec group at Chalmers, Stefan Heule,
David Mazi{\`e}res, and Edward Z. Yang for the useful discussions.
We thank the anonymous reviewers for constructive feedback on an
earlier version of this work.
This work was funded by DARPA CRASH under contract \#N66001-10-2-4088,
the Swedish research agency VR, and a grant from Mozilla.
Deian Stefan was supported by the DoD through the NDSEG Fellowship
Program.

{\frenchspacing
\bibliographystyle{abbrvnat}
\bibliography{conferences,dm,local}
}
\balance

\appendix

\section{Semantics for the base calculus}
\label{sec:app:sem}

\begin{figure}[t] % sos:rules
\small
\begin{mathpar}
\inferrule[app]
{ } { \ensuremath{\Red{E}\;[\mskip1.5mu (\lambda \Varid{x}.\Varid{t}_{1})\;\Varid{t}_{2}\mskip1.5mu]\lto\Red{E}\;[\mskip1.5mu \{\mskip1.5mu \Varid{t}_{2}\mathbin{/}\Varid{x}\mskip1.5mu\}\;\Varid{t}_{1}\mskip1.5mu]} }
\and
\inferrule[fix]
{ } { \ensuremath{\Red{E}\;[\mskip1.5mu \mathbf{fix}\;(\lambda \Varid{x}.\Varid{t})\mskip1.5mu]\lto\Red{E}\;[\mskip1.5mu \{\mskip1.5mu \mathbf{fix}\;(\lambda \Varid{x}.\Varid{t})\mathbin{/}\Varid{x}\mskip1.5mu\}\;\Varid{t}}] }
\and
\inferrule[ifTrue]
{ } { \ensuremath{\Red{E}\;[\mskip1.5mu \mathbf{if}\;\Conid{True}\;\mathbf{then}\;\Varid{t}_{2}\;\mathbf{else}\;\Varid{t}_{3}\mskip1.5mu]\lto\Red{E}\;[\mskip1.5mu \Varid{t}_{2}\mskip1.5mu]} }
\and
\inferrule[ifFalse]
{ } { \ensuremath{\Red{E}\;[\mskip1.5mu \mathbf{if}\;\Conid{False}\;\mathbf{then}\;\Varid{t}_{2}\;\mathbf{else}\;\Varid{t}_{3}\mskip1.5mu]\lto\Red{E}\;[\mskip1.5mu \Varid{t}_{3}\mskip1.5mu]} }
\and
\inferrule[labelOp]
{ \ensuremath{\Varid{v}\mathrel{=}\llbracket \Varid{l}_{1}\;\otimes\;\Varid{l}_{2}\rrbracket_{\lattice}} }
{ \ensuremath{\Red{E}\;[\mskip1.5mu \Varid{l}_{1}\;\otimes\;\Varid{l}_{2}\mskip1.5mu]\lto\Red{E}\;[\mskip1.5mu \Varid{v}\mskip1.5mu]} }
\and
\inferrule[return]
{ } { \ensuremath{\conf{\Sigma}{\textbf{\Blue{E}}\;[\mskip1.5mu \mathbf{return}\;\Varid{t}\mskip1.5mu]}\lto\conf{\Sigma}{\textbf{\Blue{E}}\;[\mskip1.5mu LIO^{\Red{\textsf{{\tiny TCB}}}}\;\Varid{t}\mskip1.5mu]}} }
\and
\inferrule[bind]
{ }
{ \ensuremath{\conf{\Sigma}{\textbf{\Blue{E}}\;[\mskip1.5mu (LIO^{\Red{\textsf{{\tiny TCB}}}}\;\Varid{t}_{1})\bind \Varid{t}_{2}\mskip1.5mu]}\lto\conf{\Sigma}{\textbf{\Blue{E}}\;[\mskip1.5mu \Varid{t}_{2}\;\Varid{t}_{1}\mskip1.5mu]}} }
\end{mathpar}
\caption{Reduction rules for standard \lio{} terms.\label{fig:sos:rules}}
\end{figure}

The reduction rules for pure and monadic terms are given in
Figure~\ref{fig:sos:rules}.
We define substitution \ensuremath{\{\mskip1.5mu \Varid{t}_{2}\mathbin{/}\Varid{x}\mskip1.5mu\}\;\Varid{t}_{1}} in the
usual way: homomorphic on all operators and renaming bound names to avoid
captures.
Our label operations \ensuremath{\lub}, \ensuremath{\glb}, and \ensuremath{\flows} rely on the label-specific
implementation of these lattice operators, as used in the premise of rule
\ruleref{labelOp}; we use the meta-level partial function \ensuremath{\llbracket \cdot\rrbracket_{\lattice}}, which
maps terms to values, to precisely capture this implementation detail.

%\concept{pure red}
%
The reduction rules for pure terms are standard. For instance, in rule
\ruleref{ifTrue}, when the branch has a true condition, i.e., \ensuremath{\Red{E}\;[\mskip1.5mu \mathbf{if}\;\Conid{True}\;\mathbf{then}\;\Varid{t}_{2}\;\mathbf{else}\;\Varid{t}_{3}\mskip1.5mu]}, it reduces to the then branch (\ensuremath{\Red{E}\;[\mskip1.5mu \Varid{t}_{2}\mskip1.5mu]}).  The
rest are self-explanatory and we do not discuss them any further.
%

%\concept{monad red}
%
%The reduction rules for monadic terms deserve some attention.
%
Since all the IFC checks are performed by individual LIO terms, the
definition for \ensuremath{\mathbf{return}} and \ensuremath{(\bind )} are trivial.
The former simply reduces to a monadic value by wrapping the term with
the \ensuremath{LIO^{\Red{\textsf{{\tiny TCB}}}}} constructor, while the latter evaluates the left-hand
term and supplies the result to the right-hand term, as usual.

\section{Attack on naive flow-sensitive references}
\label{sec:app:fs-attack2}
%\begin{wrapfigure}{r}{0.55\columnwidth}
\begin{figure}
%\vspace{-20pt}
\small
\begin{hscode}\SaveRestoreHook
\column{B}{@{}>{\hspre}l<{\hspost}@{}}%
\column{3}{@{}>{\hspre}l<{\hspost}@{}}%
\column{9}{@{}>{\hspre}l<{\hspost}@{}}%
\column{24}{@{}>{\hspre}l<{\hspost}@{}}%
\column{E}{@{}>{\hspre}l<{\hspost}@{}}%
\>[B]{}\Varid{leakRef}\mathbin{::}Ref_{\textsf{{\tiny FS}}}\;{}\;\Conid{Bool}\to \Conid{LIO}\;\Conid{Bool}{}\<[E]%
\\
\>[B]{}\Varid{leakRef}\;\Varid{href}\mathrel{=}\mathbf{do}{}\<[E]%
\\
\>[B]{}\hsindent{3}{}\<[3]%
\>[3]{}\Varid{lref}{}\<[9]%
\>[9]{}\leftarrow \textbf{newRef}_{{}}\;\textbf{L}\;\Conid{True}{}\<[E]%
\\
\>[B]{}\hsindent{3}{}\<[3]%
\>[3]{}\Varid{tmp}{}\<[9]%
\>[9]{}\leftarrow \textbf{newRef}_{{}}\;\textbf{L}\;\Conid{False}{}\<[E]%
\\
\>[B]{}\hsindent{3}{}\<[3]%
\>[3]{}\textbf{toLabeled}\;\textbf{H}\mathbin{\$}\mathbf{do}\;{}\<[24]%
\>[24]{}\Varid{h}\leftarrow \textbf{readRef}_{{}}\;\Varid{href}{}\<[E]%
\\
\>[24]{}\textbf{when}\;\Varid{h}\mathbin{\$}\textbf{writeRef}_{{}}\;\Varid{tmp}\;\Conid{True}{}\<[E]%
\\
\>[B]{}\hsindent{3}{}\<[3]%
\>[3]{}\textbf{toLabeled}\;\textbf{H}\mathbin{\$}\mathbf{do}\;{}\<[24]%
\>[24]{}\Varid{t}\leftarrow \textbf{readRef}_{{}}\;\Varid{tmp}{}\<[E]%
\\
\>[24]{}\textbf{when}\;(\neg \;\Varid{t})\mathbin{\$}\textbf{writeRef}_{{}}\;\Varid{lref}\;\Conid{False}{}\<[E]%
\\
\>[B]{}\hsindent{3}{}\<[3]%
\>[3]{}\textbf{readRef}_{{}}\;\Varid{lref}{}\<[E]%
\ColumnHook
\end{hscode}\resethooks
\vspace{-5pt}
\caption{An attack in LIO with naive flow-sensitive reference
extension without \ensuremath{\textbf{labelOf}_{{}}}.\label{fig:fs-attack2}}
%\end{wrapfigure}
\end{figure}
As in the attack of Figure~\ref{fig:fs-attack}, the \ensuremath{\Varid{leakRef}} of
Figure~\ref{fig:fs-attack2} can be used to leak the value stored in a \ensuremath{\textbf{H}}
reference \ensuremath{\Varid{href}}, while keeping the current label \ensuremath{\textbf{L}}, without using \ensuremath{\textbf{labelOf}_{{}}}.
Internally, the value is leaked into public reference \ensuremath{\Varid{lref}} by leveraging the
fact that, based on a secret value, the label of a public reference (\ensuremath{\Varid{tmp}}) can
be changed (or not).
In the first \ensuremath{\textbf{toLabeled}} block, if \ensuremath{\Varid{h}\equiv \Conid{True}}, then the label of \ensuremath{\Varid{tmp}} is raised
to \ensuremath{\textbf{H}} and its value is set to \ensuremath{\Conid{True}}.
In the second \ensuremath{\textbf{toLabeled}} block, we read \ensuremath{\Varid{tmp}}, which may raise the current
label to \ensuremath{\textbf{H}} if the secret is \ensuremath{\Conid{True}} (and thus \ensuremath{\Varid{tmp}} was upgraded).
Indeed, if the secret is \ensuremath{\Conid{True}} (and thus \ensuremath{\Varid{t}\equiv \Conid{True}}) we leave the public
reference intact: \ensuremath{\Conid{True}}.
However, if the secret is \ensuremath{\Conid{False}}, the \ensuremath{\Varid{tmp}} reference is not modified in the
first \ensuremath{\textbf{toLabeled}} block and thus when reading it in the second \ensuremath{\textbf{toLabeled}}
block, the current label remains \ensuremath{\textbf{L}}, and since \ensuremath{\Varid{t}\equiv \Conid{False}}, we write \ensuremath{\Conid{False}}
into the public reference.
In both cases the value stored in \ensuremath{\Varid{lref}} corresponds to that of \ensuremath{\Varid{href}}, yet
leaving the current label and the label of \ensuremath{\Varid{lref}} intact (\ensuremath{\textbf{L}}).

\section{Embedding Theorem}
\label{app:embedding}

In this section we prove that the embedding from \liofs{} into \lio{}
preserves semantics.

We will use the following lemma for single \liofs{} steps:

\begin{lemma} [Single-step embedding]
\label{lem:step-embedding}
Let \ensuremath{\Varid{t}} be a well-typed term in \liofs{}.  Then if \ensuremath{\conf{\Sigma}{\Varid{t}}\lto\conf{\Sigma'}{\Varid{t}'}}, then there is a configuration \ensuremath{\Conid{Y}} such that \ensuremath{\conf{\llbracket \Sigma\rrbracket_{\textsc{fi}}}{\llbracket \Varid{t}\rrbracket_{\textsc{fi}}^{\Sigma}}\lto^*\Conid{Y}} and \ensuremath{\conf{\llbracket \Sigma'\rrbracket_{\textsc{fi}}}{\llbracket \Varid{t}'\rrbracket_{\textsc{fi}}^{\Sigma}}\lto^*\Conid{Y}}, i.e. \ensuremath{\conf{\llbracket \Sigma\rrbracket_{\textsc{fi}}}{\llbracket \Varid{t}\rrbracket_{\textsc{fi}}^{\Sigma}}} and \ensuremath{\conf{\llbracket \Sigma'\rrbracket_{\textsc{fi}}}{\llbracket \Varid{t}'\rrbracket_{\textsc{fi}}^{\Sigma}}} are
$\beta$-equivalent.
\end{lemma}

\begin{proof}

  Case analysis on the next redex in \ensuremath{\Varid{t}}. Most cases show a stronger
  version of the lemma, i.e. that \ensuremath{\conf{\Sigma}{\Varid{t}}\lto\conf{\Sigma'}{\Varid{t}'}} implies
  \ensuremath{\conf{\llbracket \Sigma\rrbracket_{\textsc{fi}}}{\llbracket \Varid{t}\rrbracket_{\textsc{fi}}^{\Sigma}}\lto^*\conf{\llbracket \Sigma'\rrbracket_{\textsc{fi}}}{\llbracket \Varid{t}'\rrbracket_{\textsc{fi}}^{\Sigma}}}.

{\bf Case} \ensuremath{\textbf{\Blue{E}}\;[\mskip1.5mu \textbf{newRef}_{\textsc{fs}}\;\Varid{l}\;\Varid{t}\mskip1.5mu]}.

We have \ensuremath{\conf{\Sigma}{\textbf{\Blue{E}}\;[\mskip1.5mu \textbf{newRef}_{\textsc{fs}}\;\Varid{l}\;\Varid{t}\mskip1.5mu]}\lto\conf{\Sigma'}{\textbf{\Blue{E}}\;[\mskip1.5mu \mathbf{return}\;(Ref^{\Red{\textsf{{\tiny TCB}}}}_{\textsc{fs}}\;\Varid{a})\mskip1.5mu]}}, where \ensuremath{\Sigma'\mathrel{=}\Sigma\;[\mskip1.5mu \mu_\textsc{fs}\;\mapsto\;\Sigma.\mu_\textsc{fs}\;[\mskip1.5mu \Varid{a}\;\mapsto\;Lb^{\Red{\textsf{{\tiny TCB}}}}\;\Sigma.\lcurr\;(Lb^{\Red{\textsf{{\tiny TCB}}}}\;\Varid{l}\;\Varid{t})\mskip1.5mu]\mskip1.5mu]}, and we know that \ensuremath{\Sigma.\lcurr\;\flows\;\Varid{l}}.

Let \ensuremath{\Sigma_{1}\mathrel{=}\llbracket \Sigma\rrbracket_{\textsc{fi}}}. We argue

\begin{hscode}\SaveRestoreHook
\column{B}{@{}>{\hspre}l<{\hspost}@{}}%
\column{7}{@{}>{\hspre}l<{\hspost}@{}}%
\column{38}{@{}>{\hspre}l<{\hspost}@{}}%
\column{E}{@{}>{\hspre}l<{\hspost}@{}}%
\>[7]{}\conf{\Sigma_{1}}{\llbracket \textbf{\Blue{E}}\;[\mskip1.5mu \textbf{newRef}_{\textsc{fs}}\;\Varid{l}\;\Varid{t}\mskip1.5mu]\rrbracket_{\textsc{fi}}^{\Sigma}}{}\<[E]%
\\
\>[B]{}\lto{}\<[7]%
\>[7]{}\conf{\Sigma_{1}'}{(\llbracket \textbf{\Blue{E}}\rrbracket_{\textsc{fi}}^{\Sigma})\;[\mskip1.5mu \mathbf{do}\;\lcurr\leftarrow \mathbf{getLabel};\mathbf{return}\;(\Varid{wrap}\;(Ref^{\Red{\textsf{{\tiny TCB}}}}_{\textsc{fi}}\;\lcurr\;\Varid{a}))\mskip1.5mu]}{}\<[E]%
\\
\>[7]{}(\Sigma_{1}'\mathrel{=}\Sigma_{1}\;[\mskip1.5mu \mu_\textsc{fi}\;\mapsto\;\Sigma_{1}.\mu_\textsc{fi}\;[\mskip1.5mu \Varid{i}\;\mapsto\;Lb^{\Red{\textsf{{\tiny TCB}}}}\;\Varid{l}\;\Varid{t}\mskip1.5mu]\mskip1.5mu]){}\<[E]%
\\
\>[B]{}\lto{}\<[7]%
\>[7]{}\conf{\Sigma_{1}''}{(\llbracket \textbf{\Blue{E}}\rrbracket_{\textsc{fi}}^{\Sigma})\;[\mskip1.5mu \mathbf{return}\;(\Varid{wrap}\;(Ref^{\Red{\textsf{{\tiny TCB}}}}_{\textsc{fi}}\;\lcurr\;\Varid{a}))\mskip1.5mu]}{}\<[E]%
\\
\>[7]{}(\Sigma_{1}''\mathrel{=}\Sigma_{1}\;[\mskip1.5mu \mu_\textsc{fi}\;\mapsto\;\Sigma_{1}.\mu_\textsc{fi}\;[\mskip1.5mu {}\<[38]%
\>[38]{}\Varid{i}\;\mapsto\;Lb^{\Red{\textsf{{\tiny TCB}}}}\;\Varid{l}\;\Varid{t};{}\<[E]%
\\
\>[38]{}\Varid{a}\;\mapsto\;Lb^{\Red{\textsf{{\tiny TCB}}}}\;\lcurr\;(Ref^{\Red{\textsf{{\tiny TCB}}}}_{\textsc{fi}}\;\Varid{l}\;\Varid{i})\mskip1.5mu]\mskip1.5mu]){}\<[E]%
\ColumnHook
\end{hscode}\resethooks

We now have to check that \ensuremath{\llbracket \mathbf{return}\;(Ref^{\Red{\textsf{{\tiny TCB}}}}_{\textsc{fs}}\;\Varid{a})\rrbracket_{\textsc{fi}}^{\Sigma'}\mathrel{=}\mathbf{return}\;(\Varid{wrap}\;(Ref^{\Red{\textsf{{\tiny TCB}}}}_{\textsc{fi}}\;\lcurr\;\Varid{a}))} and \ensuremath{\llbracket \Sigma'\rrbracket_{\textsc{fi}}\mathrel{=}\Sigma_{1}''}, which follow directly from
the definition of \ensuremath{\llbracket \cdot \rrbracket_{\textsc{fi}}} for references and states.
\medskip

{\bf Case} \ensuremath{\textbf{\Blue{E}}\;[\mskip1.5mu \textbf{readRef}_{\textsc{fs}}\;(Ref^{\Red{\textsf{{\tiny TCB}}}}_{\textsc{fs}}\;\Varid{a})\mskip1.5mu]}.

We have \ensuremath{\conf{\Sigma}{\textbf{\Blue{E}}\;[\mskip1.5mu \textbf{readRef}_{\textsc{fs}}\;(Ref^{\Red{\textsf{{\tiny TCB}}}}_{\textsc{fs}}\;\Varid{a})\mskip1.5mu]}\lto\conf{\Sigma}{\textbf{\Blue{E}}\;[\mskip1.5mu \textbf{unlabel}\;(Lb^{\Red{\textsf{{\tiny TCB}}}}\;(\Varid{l}\;\lub\;\Varid{l'})\;\Varid{t})\mskip1.5mu]}}, where \ensuremath{\Sigma.\mu_\textsc{fs}\;(\Varid{a})\mathrel{=}Lb^{\Red{\textsf{{\tiny TCB}}}}\;\Varid{l}\;(Lb^{\Red{\textsf{{\tiny TCB}}}}\;\Varid{l'}\;\Varid{t})}.

Let \ensuremath{\Sigma_{1}\mathrel{=}\llbracket \Sigma\rrbracket_{\textsc{fi}}}. We argue

\begin{hscode}\SaveRestoreHook
\column{B}{@{}>{\hspre}l<{\hspost}@{}}%
\column{7}{@{}>{\hspre}l<{\hspost}@{}}%
\column{E}{@{}>{\hspre}l<{\hspost}@{}}%
\>[7]{}\conf{\Sigma_{1}}{\llbracket \textbf{\Blue{E}}\;[\mskip1.5mu \textbf{readRef}_{\textsc{fs}}\;(Ref^{\Red{\textsf{{\tiny TCB}}}}_{\textsc{fs}}\;\Varid{a})\mskip1.5mu]\rrbracket_{\textsc{fi}}^{\Sigma}}{}\<[E]%
\\
\>[B]{}\lto{}\<[7]%
\>[7]{}\conf{\Sigma_{1}}{(\llbracket \textbf{\Blue{E}}\rrbracket_{\textsc{fi}}^{\Sigma})\;[\mskip1.5mu \textbf{readRef}_{\textsc{fi}}\;(\llbracket Ref^{\Red{\textsf{{\tiny TCB}}}}_{\textsc{fs}}\;\Varid{a}\rrbracket_{\textsc{fi}}^{\cdot })\bind \textbf{readRef}_{\textsc{fi}}\mskip1.5mu]}{}\<[E]%
\\
\>[B]{}\lto{}\<[7]%
\>[7]{}\conf{\Sigma_{1}}{(\llbracket \textbf{\Blue{E}}\rrbracket_{\textsc{fi}}^{\Sigma})\;[\mskip1.5mu \textbf{unlabel}\;(\Sigma_{1}.\mu_\textsc{fi}\;(\Varid{a}))\bind \textbf{readRef}_{\textsc{fi}}\mskip1.5mu]}{}\<[E]%
\\
\>[B]{}\lto{}\<[7]%
\>[7]{}\conf{\Sigma_{1}'}{(\llbracket \textbf{\Blue{E}}\rrbracket_{\textsc{fi}}^{\Sigma})\;[\mskip1.5mu \mathbf{return}\;(Ref^{\Red{\textsf{{\tiny TCB}}}}_{\textsc{fi}}\;\Varid{l'}\;\Varid{i})\bind \textbf{readRef}_{\textsc{fi}}\mskip1.5mu]}{}\<[E]%
\\
\>[7]{}(\Sigma_{1}'\mathrel{=}\Sigma_{1}\;[\mskip1.5mu \lcurr\;\mapsto\;\Sigma_{1}.\lcurr\;\lub\;\Varid{l}\mskip1.5mu]){}\<[E]%
\\
\>[B]{}\lto{}\<[7]%
\>[7]{}\conf{\Sigma_{1}'}{(\llbracket \textbf{\Blue{E}}\rrbracket_{\textsc{fi}}^{\Sigma})\;[\mskip1.5mu \textbf{readRef}_{\textsc{fi}}\;(Ref^{\Red{\textsf{{\tiny TCB}}}}_{\textsc{fi}}\;\Varid{l'}\;\Varid{i})\mskip1.5mu]}{}\<[E]%
\\
\>[B]{}\lto{}\<[7]%
\>[7]{}\conf{\Sigma_{1}'}{(\llbracket \textbf{\Blue{E}}\rrbracket_{\textsc{fi}}^{\Sigma})\;[\mskip1.5mu \textbf{unlabel}\;(\Sigma_{1}'.\mu_\textsc{fi}\;(\Varid{i}))\mskip1.5mu]}{}\<[E]%
\\
\>[B]{}\lto{}\<[7]%
\>[7]{}\conf{\Sigma_{1}''}{(\llbracket \textbf{\Blue{E}}\rrbracket_{\textsc{fi}}^{\Sigma})\;[\mskip1.5mu \mathbf{return}\;\Varid{t}\mskip1.5mu]}{}\<[E]%
\\
\>[7]{}(\Sigma_{1}''\mathrel{=}\Sigma_{1}'\;[\mskip1.5mu \lcurr\;\mapsto\;\Sigma_{1}'.\lcurr\;\lub\;\Varid{l'}\mskip1.5mu]){}\<[E]%
\ColumnHook
\end{hscode}\resethooks

Now if we consider \ensuremath{\conf{\llbracket \Sigma\rrbracket_{\textsc{fi}}}{\llbracket \textbf{\Blue{E}}\;[\mskip1.5mu \textbf{unlabel}\;(Lb^{\Red{\textsf{{\tiny TCB}}}}\;(\Varid{l}\;\lub\;\Varid{l'})\;\Varid{t})\mskip1.5mu]\rrbracket_{\textsc{fi}}^{\cdot }}}, we have

\begin{hscode}\SaveRestoreHook
\column{B}{@{}>{\hspre}l<{\hspost}@{}}%
\column{7}{@{}>{\hspre}l<{\hspost}@{}}%
\column{E}{@{}>{\hspre}l<{\hspost}@{}}%
\>[7]{}\conf{\llbracket \Sigma\rrbracket_{\textsc{fi}}}{\llbracket \textbf{\Blue{E}}\;[\mskip1.5mu \textbf{unlabel}\;(Lb^{\Red{\textsf{{\tiny TCB}}}}\;(\Varid{l}\;\lub\;\Varid{l'})\;\Varid{t})\mskip1.5mu]\rrbracket_{\textsc{fi}}^{\cdot }}{}\<[E]%
\\
\>[B]{}\lto{}\<[7]%
\>[7]{}\conf{\Sigma_{2}}{(\llbracket \textbf{\Blue{E}}\rrbracket_{\textsc{fi}}^{\Sigma})\;[\mskip1.5mu \mathbf{return}\;\Varid{t}\mskip1.5mu]{}\<[E]%
\\
\>[7]{}(\Sigma_{2}\mathrel{=}\Sigma_{2}\;[\mskip1.5mu \lcurr\;\mapsto\;(\llbracket \Sigma\rrbracket_{\textsc{fi}}).\lcurr\;\lub\;\Varid{l}\;\lub\;\Varid{l'}\mskip1.5mu])}{}\<[E]%
\ColumnHook
\end{hscode}\resethooks

Note that \ensuremath{\Sigma_{1}''.\lcurr\mathrel{=}(\llbracket \Sigma\rrbracket_{\textsc{fi}}).\lcurr\;\lub\;\Varid{l}\;\lub\;\Varid{l'}\mathrel{=}\Sigma_{2}.\lcurr}.

{\bf Case} \ensuremath{\textbf{\Blue{E}}\;[\mskip1.5mu \textbf{writeRef}_{\textsc{fs}}\;(Ref^{\Red{\textsf{{\tiny TCB}}}}_{\textsc{fs}}\;\Varid{a})\;\Varid{t}\mskip1.5mu]}.

We have \ensuremath{\conf{\Sigma}{\textbf{\Blue{E}}\;[\mskip1.5mu \textbf{writeRef}_{\textsc{fs}}\;(Ref^{\Red{\textsf{{\tiny TCB}}}}_{\textsc{fs}}\;\Varid{a})\;\Varid{t}\mskip1.5mu]}\lto\conf{\Sigma'}{\textbf{\Blue{E}}\;[\mskip1.5mu \mathbf{return}\;()\mskip1.5mu]}}, where \ensuremath{\Sigma.\mu_\textsc{fs}\;(\Varid{a})\mathrel{=}Lb^{\Red{\textsf{{\tiny TCB}}}}\;\Varid{l}\;(Lb^{\Red{\textsf{{\tiny TCB}}}}\;\Varid{l'}\;\Varid{v})}, \ensuremath{\Sigma'\mathrel{=}\Sigma\;[\mskip1.5mu \mu_\textsc{fs}\;\mapsto\;\Sigma.\mu_\textsc{fs}\;[\mskip1.5mu \Varid{a}\;\mapsto\;Lb^{\Red{\textsf{{\tiny TCB}}}}\;\Varid{l}\;(Lb^{\Red{\textsf{{\tiny TCB}}}}\;\Varid{l'}\;\Varid{v})\mskip1.5mu]\mskip1.5mu]},
and we know that \ensuremath{\Sigma.\lcurr\;\flows\;\Varid{l}\;\lub\;\Varid{l'}}.

Let \ensuremath{\Sigma_{1}\mathrel{=}\llbracket \Sigma\rrbracket_{\textsc{fi}}}. Then there exists a function \ensuremath{\mu} such that:

\begin{hscode}\SaveRestoreHook
\column{B}{@{}>{\hspre}l<{\hspost}@{}}%
\column{7}{@{}>{\hspre}l<{\hspost}@{}}%
\column{E}{@{}>{\hspre}l<{\hspost}@{}}%
\>[7]{}\conf{\Sigma_{1}}{\llbracket \textbf{\Blue{E}}\;[\mskip1.5mu \textbf{upgrade}_\textsc{fs}\;(Ref^{\Red{\textsf{{\tiny TCB}}}}_{\textsc{fs}}\;\Varid{a})\;\Varid{l'}\mskip1.5mu]\rrbracket_{\textsc{fi}}^{\Sigma}}{}\<[E]%
\\
\>[B]{}\lto{}\<[7]%
\>[7]{}\conf{\Sigma_{1}}{(\llbracket \textbf{\Blue{E}}\rrbracket_{\textsc{fi}}^{\Sigma})\;[\mskip1.5mu \textbf{toLabeled}\;(\Sigma_{1}.\lcurr\;\lub\;\Varid{l})\;(\mu\;\Sigma_{1}.\lcurr)\mskip1.5mu]}{}\<[E]%
\ColumnHook
\end{hscode}\resethooks

We now step through the evaluation of \ensuremath{\conf{\Sigma_{1}}{\mu\;\Sigma_{1}.\lcurr}}, as follows:

\begin{hscode}\SaveRestoreHook
\column{B}{@{}>{\hspre}l<{\hspost}@{}}%
\column{7}{@{}>{\hspre}l<{\hspost}@{}}%
\column{E}{@{}>{\hspre}l<{\hspost}@{}}%
\>[7]{}\conf{\Sigma_{1}}{\mu\;\Sigma_{1}.\lcurr}{}\<[E]%
\\
\>[B]{}\lto{}\<[7]%
\>[7]{}\conf{\Sigma_{1}}{\;\mathbf{do}\;\Varid{i}\leftarrow \textbf{readRef}_{\llbracket Ref^{\Red{\textsf{{\tiny TCB}}}}_{\textsc{fs}}\;\Varid{a}\rrbracket_{\textsc{fi}}^{\cdot }};\textbf{writeRef}_{\textsc{fi}}\;\Varid{i}\;\Varid{t}}{}\<[E]%
\\
\>[B]{}\lto{}\<[7]%
\>[7]{}\conf{\Sigma_{1}'}{\textbf{writeRef}_{\textsc{fi}}\;\Varid{i}\;\Varid{t}}{}\<[E]%
\\
\>[7]{}(\Sigma_{1}'\mathrel{=}\Sigma_{1}\;[\mskip1.5mu \lcurr\;\mapsto\;\Sigma_{1}.\lcurr\;\lub\;\Varid{l}\mskip1.5mu]){}\<[E]%
\\
\>[B]{}\lto\conf{\Sigma_{1}''}{\mathbf{return}\;()}\;{}\<[E]%
\\
\>[B]{}\hsindent{7}{}\<[7]%
\>[7]{}(\Sigma_{1}''\mathrel{=}\Sigma_{1}'\;[\mskip1.5mu \mu_\textsc{fi}\;\mapsto\;\Sigma_{1}'.\mu_\textsc{fi}\;[\mskip1.5mu \Varid{i}\;\mapsto\;Lb^{\Red{\textsf{{\tiny TCB}}}}\;\Varid{l'}\;\Varid{t}\mskip1.5mu]\mskip1.5mu]){}\<[E]%
\ColumnHook
\end{hscode}\resethooks

Finally, this allows us to conclude (from the rule for \ensuremath{\textbf{toLabeled}}), that

\begin{hscode}\SaveRestoreHook
\column{B}{@{}>{\hspre}l<{\hspost}@{}}%
\column{E}{@{}>{\hspre}l<{\hspost}@{}}%
\>[B]{}\conf{\Sigma_{1}}{(\llbracket \textbf{\Blue{E}}\rrbracket_{\textsc{fi}}^{\Sigma})\;[\mskip1.5mu \textbf{toLabeled}\;(\Sigma_{1}.\lcurr\;\lub\;\Varid{l})\;(\mu\;\Sigma_{1}.\lcurr)\mskip1.5mu]}{}\<[E]%
\\
\>[B]{}\lto\conf{\Sigma_{2}}{(\llbracket \textbf{\Blue{E}}\rrbracket_{\textsc{fi}}^{\Sigma})\;[\mskip1.5mu \mathbf{return}\;()\mskip1.5mu]}{}\<[E]%
\ColumnHook
\end{hscode}\resethooks

where \ensuremath{\Sigma_{2}\mathrel{=}(\Sigma_{1}.\lcurr,\Sigma_{1}''.\mu_\textsc{fi})}. Now we can check that \ensuremath{\llbracket \Sigma'\rrbracket_{\textsc{fi}}\mathrel{=}\Sigma_{2}} from the definition of \ensuremath{\llbracket \cdot \rrbracket_{\textsc{fi}}} for states.

{\bf Case} \ensuremath{\textbf{\Blue{E}}\;[\mskip1.5mu \textbf{upgrade}_\textsc{fs}\;(Ref^{\Red{\textsf{{\tiny TCB}}}}_{\textsc{fs}}\;\Varid{a})\;\Varid{l'}\mskip1.5mu]}.

We have \ensuremath{\conf{\Sigma}{\textbf{\Blue{E}}\;[\mskip1.5mu \textbf{upgrade}_\textsc{fs}\;(Ref^{\Red{\textsf{{\tiny TCB}}}}_{\textsc{fs}}\;\Varid{a})\;\Varid{l'}\mskip1.5mu]}\lto\conf{\Sigma'}{\textbf{\Blue{E}}\;[\mskip1.5mu \mathbf{return}\;()\mskip1.5mu]}}, where \ensuremath{\Sigma.\mu_\textsc{fs}\;(\Varid{a})\mathrel{=}Lb^{\Red{\textsf{{\tiny TCB}}}}\;\Varid{l}\;(Lb^{\Red{\textsf{{\tiny TCB}}}}\;\Varid{l''}\;\Varid{v})},
\ensuremath{\Sigma'\mathrel{=}\Sigma\;[\mskip1.5mu \mu_\textsc{fs}\;\mapsto\;\Sigma.\mu_\textsc{fs}\;[\mskip1.5mu \Varid{a}\;\mapsto\;Lb^{\Red{\textsf{{\tiny TCB}}}}\;\Varid{l}\;(Lb^{\Red{\textsf{{\tiny TCB}}}}\;(\Varid{l}\;\lub\;\Varid{l''}\;\lub\;\Varid{l'})\;\Varid{v})\mskip1.5mu]\mskip1.5mu]}, and we know that \ensuremath{\Sigma.\lcurr\;\flows\;\Varid{l}}.

Let \ensuremath{\Sigma_{1}\mathrel{=}\llbracket \Sigma\rrbracket_{\textsc{fi}}}. Then there exists a function \ensuremath{\mu} such that:

\begin{hscode}\SaveRestoreHook
\column{B}{@{}>{\hspre}l<{\hspost}@{}}%
\column{7}{@{}>{\hspre}l<{\hspost}@{}}%
\column{E}{@{}>{\hspre}l<{\hspost}@{}}%
\>[7]{}\conf{\Sigma_{1}}{\llbracket \textbf{\Blue{E}}\;[\mskip1.5mu \textbf{upgrade}_\textsc{fs}\;(Ref^{\Red{\textsf{{\tiny TCB}}}}_{\textsc{fs}}\;\Varid{a})\;\Varid{l'}\mskip1.5mu]\rrbracket_{\textsc{fi}}^{\Sigma}}{}\<[E]%
\\
\>[B]{}\lto{}\<[7]%
\>[7]{}\conf{\Sigma_{1}}{(\llbracket \textbf{\Blue{E}}\rrbracket_{\textsc{fi}}^{\Sigma})\;[\mskip1.5mu \textbf{toLabeled}\;(\Sigma_{1}.\lcurr\;\lub\;\Varid{l})\;(\mu\;\Sigma_{1}.\lcurr)\mskip1.5mu]}{}\<[E]%
\ColumnHook
\end{hscode}\resethooks

We now step through the evaluation of \ensuremath{\conf{\Sigma_{1}}{\mu\;\Sigma_{1}.\lcurr}}, as follows:

\begin{hscode}\SaveRestoreHook
\column{B}{@{}>{\hspre}l<{\hspost}@{}}%
\column{7}{@{}>{\hspre}l<{\hspost}@{}}%
\column{35}{@{}>{\hspre}l<{\hspost}@{}}%
\column{51}{@{}>{\hspre}l<{\hspost}@{}}%
\column{E}{@{}>{\hspre}l<{\hspost}@{}}%
\>[7]{}\conf{\Sigma_{1}}{\mu\;\Sigma_{1}.\lcurr}{}\<[E]%
\\
\>[B]{}\lto{}\<[7]%
\>[7]{}\conf{\Sigma_{1}}{\;\mathbf{do}\;\Varid{i}\leftarrow \textbf{readRef}_{\llbracket Ref^{\Red{\textsf{{\tiny TCB}}}}_{\textsc{fs}}\;\Varid{a}\rrbracket_{\textsc{fi}}^{\cdot }};\Varid{lc}\leftarrow \mathbf{getLabel};\mathbin{...}}{}\<[E]%
\\
\>[B]{}\lto{}\<[7]%
\>[7]{}\conf{\Sigma_{1}'}{\;\mathbf{do}\;\Varid{lc}\leftarrow \mathbf{getLabel};\Varid{n}\leftarrow \textbf{newRef}_{\textsc{fi}}\;{}\<[51]%
\>[51]{}(\Varid{lc}\;\lub\;(\Varid{l'}\;\lub\;\Varid{l}))\;\bot ;\mathbin{...}}{}\<[E]%
\\
\>[7]{}(\Sigma_{1}'\mathrel{=}\Sigma_{1}\;[\mskip1.5mu \lcurr\;\mapsto\;\Sigma_{1}.\lcurr\;\lub\;\Varid{l}\mskip1.5mu]){}\<[E]%
\\
\>[B]{}\lto{}\<[7]%
\>[7]{}\conf{\Sigma_{1}'}{\;\mathbf{do}\;\Varid{n}\leftarrow \textbf{newRef}_{\textsc{fi}}\;{}\<[35]%
\>[35]{}(\Varid{lc}\;\lub\;(\Varid{l'}\;\lub\;\Varid{l}))\;\bot ;\textbf{copyRef}\;\Varid{i}\;\Varid{n};\mathbin{...}}{}\<[E]%
\\
\>[B]{}\lto{}\<[7]%
\>[7]{}\conf{\Sigma_{1}''}{\;\mathbf{do}\;\textbf{copyRef}\;\Varid{i}\;\Varid{n};\textbf{writeRef}_{\textsc{fi}}\;(\llbracket Ref^{\Red{\textsf{{\tiny TCB}}}}_{\textsc{fs}}\;\Varid{a}\rrbracket_{\textsc{fi}}^{\cdot })\;\Varid{n}}{}\<[E]%
\\
\>[7]{}(\Sigma_{1}''\mathrel{=}\Sigma_{1}'\;[\mskip1.5mu \mu_\textsc{fi}\;\mapsto\;\Sigma_{1}'.\mu_\textsc{fi}\;[\mskip1.5mu \Varid{n}\;\mapsto\;Lb^{\Red{\textsf{{\tiny TCB}}}}\;(\Varid{lc}\;\lub\;(\Varid{l'}\;\lub\;\Varid{l}))\;\bot \mskip1.5mu]\mskip1.5mu]){}\<[E]%
\\
\>[B]{}\lto{}\<[7]%
\>[7]{}\conf{\Sigma_{1}'''}{\textbf{writeRef}_{\textsc{fi}}\;(\llbracket Ref^{\Red{\textsf{{\tiny TCB}}}}_{\textsc{fs}}\;\Varid{a}\rrbracket_{\textsc{fi}}^{\cdot })\;\Varid{n}}{}\<[E]%
\\
\>[7]{}(\Sigma_{1}'''\mathrel{=}\Sigma_{1}''\;[\mskip1.5mu \mu_\textsc{fi}\;\mapsto\;\Sigma_{1}''.\mu_\textsc{fi}\;[\mskip1.5mu \Varid{n}\;\mapsto\;Lb^{\Red{\textsf{{\tiny TCB}}}}\;(\Varid{lc}\;\lub\;(\Varid{l'}\;\lub\;\Varid{l}))\;\Varid{v}\mskip1.5mu]\mskip1.5mu]){}\<[E]%
\\
\>[B]{}\lto\conf{\Sigma_{1}'''}{\mathbf{return}\;()}\;{}\<[E]%
\\
\>[B]{}\hsindent{7}{}\<[7]%
\>[7]{}(\Sigma_{1}''''\mathrel{=}\Sigma_{1}'''\;[\mskip1.5mu \mu_\textsc{fi}\;\mapsto\;\Sigma_{1}'''.\mu_\textsc{fi}\;[\mskip1.5mu \Varid{a}\;\mapsto\;Lb^{\Red{\textsf{{\tiny TCB}}}}\;\Varid{l}\;(Ref^{\Red{\textsf{{\tiny TCB}}}}_{\textsc{fi}}\;(\Varid{lc}\;\lub\;\Varid{l'}\;\lub\;\Varid{l})\;\Varid{n})\mskip1.5mu]\mskip1.5mu]){}\<[E]%
\ColumnHook
\end{hscode}\resethooks

Finally, this allows us to conclude (from the rule for \ensuremath{\textbf{toLabeled}}), that

\begin{hscode}\SaveRestoreHook
\column{B}{@{}>{\hspre}l<{\hspost}@{}}%
\column{E}{@{}>{\hspre}l<{\hspost}@{}}%
\>[B]{}\conf{\Sigma_{1}}{(\llbracket \textbf{\Blue{E}}\rrbracket_{\textsc{fi}}^{\Sigma})\;[\mskip1.5mu \textbf{toLabeled}\;\Varid{l'}\;(\mu\;\Sigma_{1}.\lcurr)\mskip1.5mu]}{}\<[E]%
\\
\>[B]{}\lto\conf{\Sigma_{2}}{(\llbracket \textbf{\Blue{E}}\rrbracket_{\textsc{fi}}^{\Sigma})\;[\mskip1.5mu \mathbf{return}\;()\mskip1.5mu]}{}\<[E]%
\ColumnHook
\end{hscode}\resethooks

where \ensuremath{\Sigma_{2}\mathrel{=}(\Sigma_{1}.\lcurr,\Sigma_{1}''''.\mu_\textsc{fi})}. Now we can check that \ensuremath{\llbracket \Sigma'\rrbracket_{\textsc{fi}}\mathrel{=}\Sigma_{2}} from the definition of \ensuremath{\llbracket \cdot \rrbracket_{\textsc{fi}}} for states.

{\bf Case} \ensuremath{\textbf{\Blue{E}}\;[\mskip1.5mu \textbf{downgrade}_\textsc{fs}\;(Ref^{\Red{\textsf{{\tiny TCB}}}}_{\textsc{fs}}\;\Varid{a})\;\Varid{l'}\mskip1.5mu]}.

We have \ensuremath{\conf{\Sigma}{\textbf{\Blue{E}}\;[\mskip1.5mu \textbf{downgrade}_\textsc{fs}\;(Ref^{\Red{\textsf{{\tiny TCB}}}}_{\textsc{fs}}\;\Varid{a})\;\Varid{l'}\mskip1.5mu]}\lto\conf{\Sigma'}{\textbf{\Blue{E}}\;[\mskip1.5mu \mathbf{return}\;()\mskip1.5mu]}}, where \ensuremath{\Sigma.\mu_\textsc{fs}\;(\Varid{a})\mathrel{=}Lb^{\Red{\textsf{{\tiny TCB}}}}\;\Varid{l}\;(Lb^{\Red{\textsf{{\tiny TCB}}}}\;\Varid{l''}\;\Varid{v})},
\ensuremath{\Sigma'\mathrel{=}\Sigma\;[\mskip1.5mu \mu_\textsc{fs}\;\mapsto\;\Sigma.\mu_\textsc{fs}\;[\mskip1.5mu \Varid{a}\;\mapsto\;Lb^{\Red{\textsf{{\tiny TCB}}}}\;\Varid{l}\;(Lb^{\Red{\textsf{{\tiny TCB}}}}\;(\Varid{l}\;\lub\;\Varid{l''}\;\glb\;\Varid{l'})\;\bot )\mskip1.5mu]\mskip1.5mu]}, and we know that \ensuremath{\Sigma.\lcurr\;\flows\;\Varid{l}}.

Let \ensuremath{\Sigma_{1}\mathrel{=}\llbracket \Sigma\rrbracket_{\textsc{fi}}}. Then there exists a function \ensuremath{\mu} such that:

\begin{hscode}\SaveRestoreHook
\column{B}{@{}>{\hspre}l<{\hspost}@{}}%
\column{7}{@{}>{\hspre}l<{\hspost}@{}}%
\column{E}{@{}>{\hspre}l<{\hspost}@{}}%
\>[7]{}\conf{\Sigma_{1}}{\llbracket \textbf{\Blue{E}}\;[\mskip1.5mu \textbf{downgrade}_\textsc{fs}\;(Ref^{\Red{\textsf{{\tiny TCB}}}}_{\textsc{fs}}\;\Varid{a})\;\Varid{l'}\mskip1.5mu]\rrbracket_{\textsc{fi}}^{\Sigma}}{}\<[E]%
\\
\>[B]{}\lto{}\<[7]%
\>[7]{}\conf{\Sigma_{1}}{(\llbracket \textbf{\Blue{E}}\rrbracket_{\textsc{fi}}^{\Sigma})\;[\mskip1.5mu \textbf{toLabeled}\;(\Sigma_{1}.\lcurr\;\lub\;\Varid{l})\;(\mu\;\Sigma_{1}.\lcurr)\mskip1.5mu]}{}\<[E]%
\ColumnHook
\end{hscode}\resethooks

We now step through the evaluation of \ensuremath{\conf{\Sigma_{1}}{\mu\;\Sigma_{1}.\lcurr}}, as follows:

\begin{hscode}\SaveRestoreHook
\column{B}{@{}>{\hspre}l<{\hspost}@{}}%
\column{7}{@{}>{\hspre}l<{\hspost}@{}}%
\column{13}{@{}>{\hspre}l<{\hspost}@{}}%
\column{35}{@{}>{\hspre}l<{\hspost}@{}}%
\column{51}{@{}>{\hspre}l<{\hspost}@{}}%
\column{E}{@{}>{\hspre}l<{\hspost}@{}}%
\>[7]{}\conf{\Sigma_{1}}{\mu\;\Sigma_{1}.\lcurr}{}\<[E]%
\\
\>[B]{}\lto{}\<[7]%
\>[7]{}\conf{\Sigma_{1}}{\;\mathbf{do}\;\Varid{i}\leftarrow \textbf{readRef}_{\llbracket Ref^{\Red{\textsf{{\tiny TCB}}}}_{\textsc{fs}}\;\Varid{a}\rrbracket_{\textsc{fi}}^{\cdot }};\Varid{lc}\leftarrow \Varid{getlabel};\mathbin{...}}{}\<[E]%
\\
\>[B]{}\lto{}\<[7]%
\>[7]{}\conf{\Sigma_{1}'}{\;\mathbf{do}\;\Varid{lc}\leftarrow \mathbf{getLabel};\Varid{n}\leftarrow \textbf{newRef}_{\textsc{fi}}\;{}\<[51]%
\>[51]{}(\Varid{lc}\;\lub\;(\Varid{l'}\;\glb\;\Varid{l}))\;\bot ;\mathbin{...}}{}\<[E]%
\\
\>[7]{}(\Sigma_{1}'\mathrel{=}\Sigma_{1}\;[\mskip1.5mu \lcurr\;\mapsto\;\Sigma_{1}.\lcurr\;\lub\;\Varid{l}\mskip1.5mu]){}\<[E]%
\\
\>[B]{}\lto{}\<[7]%
\>[7]{}\conf{\Sigma_{1}'}{\;\mathbf{do}\;\Varid{n}\leftarrow \textbf{newRef}_{\textsc{fi}}\;{}\<[35]%
\>[35]{}(\Varid{lc}\;\lub\;(\Varid{l'}\;\glb\;\Varid{l}))\;\bot ;\textbf{writeRef}_{\textsc{fi}}\;(\llbracket Ref^{\Red{\textsf{{\tiny TCB}}}}_{\textsc{fs}}\;\Varid{a}\rrbracket_{\textsc{fi}}^{\cdot })\;\Varid{n}}{}\<[E]%
\\
\>[B]{}\lto{}\<[7]%
\>[7]{}\conf{\Sigma_{1}''}{\textbf{writeRef}_{\textsc{fi}}\;(\llbracket Ref^{\Red{\textsf{{\tiny TCB}}}}_{\textsc{fs}}\;\Varid{a}\rrbracket_{\textsc{fi}}^{\cdot })\;\Varid{n}}{}\<[E]%
\\
\>[7]{}(\Sigma_{1}''\mathrel{=}\Sigma_{1}'\;[\mskip1.5mu \mu_\textsc{fi}\;\mapsto\;\Sigma_{1}'.\mu_\textsc{fi}\;[\mskip1.5mu \Varid{n}\;\mapsto\;Lb^{\Red{\textsf{{\tiny TCB}}}}\;(\Varid{lc}\;\lub\;(\Varid{l'}\;\glb\;\Varid{l}))\;\bot \mskip1.5mu]\mskip1.5mu]){}\<[E]%
\\
\>[B]{}\lto\conf{\Sigma_{1}'''}{\mathbf{return}\;()}{}\<[E]%
\\[\blanklineskip]%
\ColumnHook
\end{hscode}\resethooks

Finally, this allows us to conclude (from the rule for \ensuremath{\textbf{toLabeled}}), that

\begin{hscode}\SaveRestoreHook
\column{B}{@{}>{\hspre}l<{\hspost}@{}}%
\column{E}{@{}>{\hspre}l<{\hspost}@{}}%
\>[B]{}\conf{\Sigma_{1}}{(\llbracket \textbf{\Blue{E}}\rrbracket_{\textsc{fi}}^{\Sigma})\;[\mskip1.5mu \textbf{toLabeled}\;\Varid{l'}\;(\mu\;\Sigma_{1}.\lcurr)\mskip1.5mu]}{}\<[E]%
\\
\>[B]{}\lto\conf{\Sigma_{2}}{(\llbracket \textbf{\Blue{E}}\rrbracket_{\textsc{fi}}^{\Sigma})\;[\mskip1.5mu \mathbf{return}\;()\mskip1.5mu]}{}\<[E]%
\ColumnHook
\end{hscode}\resethooks

where \ensuremath{\Sigma_{2}\mathrel{=}(\Sigma_{1}.\lcurr,\Sigma_{1}'''.\mu_\textsc{fi})}. Now we can check that \ensuremath{\llbracket \Sigma'\rrbracket_{\textsc{fi}}\mathrel{=}\Sigma_{2}} from the definition of \ensuremath{\llbracket \cdot \rrbracket_{\textsc{fi}}} for states.

\end{proof}

Now we can state the main theorem of this section.

{\bf Theorem.} [Embedding \liofs{} in \lio{}]
Let \ensuremath{\Varid{t}} be a well-typed term in \liofs{}.
  Then if \ensuremath{\conf{\Sigma}{\Varid{t}}\lto^*\conf{\Sigma'}{\Varid{v}}}, we have \ensuremath{\conf{\llbracket \Sigma\rrbracket_{\textsc{fi}}}{\llbracket \Varid{t}\rrbracket_{\textsc{fi}}^{\Sigma}}\lto^*\conf{\llbracket \Sigma'\rrbracket_{\textsc{fi}}}{\llbracket \Varid{v}\rrbracket_{\textsc{fi}}^{\Sigma}}}, and if
  \ensuremath{\conf{\Sigma}{\Varid{t}}\lto^*\conf{\Sigma'}{{\Uparrow}}}, then
  \ensuremath{\conf{\llbracket \Sigma\rrbracket_{\textsc{fi}}}{\llbracket \Varid{t}\rrbracket_{\textsc{fi}}^{\Sigma}}\lto^*\conf{\llbracket \Sigma'\rrbracket_{\textsc{fi}}}{{\Uparrow}}}.
  % \textrm{\hl{this is not right:
  %     the constructors for FS references are different (one wraps FI
  %     references, the other doesn't), so syntactic equiality is not
  %     right.  we need the equivalence relation I defined before}}

\begin{proof}
  By induction on the number of steps in \ensuremath{\conf{\Sigma}{\Varid{t}}\lto^*\conf{\Sigma'}{\Varid{v}}},
  using Lemma~\ref{lem:step-embedding} and uniqueness of normal forms
  in \lio{}.
\end{proof}

\end{document}